\documentclass[10pt,A4paper]{article}

\usepackage[normalem]{ulem}
\usepackage{amssymb}
\usepackage{amsmath}
\usepackage{amsthm}
\usepackage{latexsym}
\usepackage[dvips]{epsfig}
\usepackage{enumerate}
\usepackage{mathrsfs}
\usepackage{eufrak}
\usepackage{bm}
\usepackage{tikz}
\usepackage{authblk}
\usepackage{cancel}
\usepackage{color}
\definecolor{pink}{rgb}{1,0,1}
\newcommand{\pink}[1]{{\textcolor{pink}{#1}}}

\allowdisplaybreaks

\theoremstyle{plain}
\newtheorem{proposition}{Proposition}
\newtheorem{lemma}{Lemma}
\newtheorem{theorem}{Theorem}
\newtheorem{assumption}{Assumption}

\newtheorem{corollary}{Corollary}

\newtheorem{remark}{Remark}

\def\bmg{{\bm g}}
\def\bmh{{\bm h}}

\def\bml{{\bm l}}
\def\bmn{{\bm n}}
\def\bmm{{\bm m}}

\def\bmK{{\bm K}}

\def\p{\partial}
\def\nablasl{/\kern-0.58em\nabla}
\def\Deltasl{/\kern-0.58em\Delta}

\def\bmpartial{{\bm \partial}}

\def\bmsigma{{\bm \sigma}}

\begin{document}

\title{\textbf{Revisiting the characteristic initial value
    problem for the vacuum Einstein field equations}}

\author[,1]{David Hilditch \footnote{E-mail
    address:{\tt\,david.hilditch@tecnico.ulisboa.pt}}}
\author[,2]{Juan A. Valiente Kroon \footnote{E-mail
    address:{\tt\,j.a.valiente-kroon@qmul.ac.uk}}}
\author[,2]{Peng Zhao \footnote{E-mail
    address:{\tt\,p.zhao@qmul.ac.uk}}}

\affil[1]{CENTRA, Departamento de F\'isica, Instituto Superior
  T\'ecnico – IST, Universidade de Lisboa – UL, Avenida Rovisco Pais
  1, 1049 Lisboa, Portugal.}

\affil[2]{School of Mathematical Sciences, Queen Mary, University of
  London, Mile End Road, London E1 4NS, United Kingdom.}

\maketitle

\begin{abstract}
Using the Newman-Penrose formalism we study the characteristic initial
value problem in vacuum General Relativity. We work in a gauge
suggested by Stewart, and following the strategy taken in the work of
Luk, demonstrate local existence of solutions in a neighbourhood of
the set on which data are given. These data are given on intersecting
null hypersurfaces. Existence near their intersection is achieved by
combining the observation that the field equations are symmetric
hyperbolic in this gauge with the results of Rendall. To obtain
existence all the way along the null-hypersurfaces themselves, a
bootstrap argument involving the Newman-Penrose variables is
performed.
\end{abstract}

\section{Introduction}

The simplest setups of partial differential equations (PDEs) are of
course the boundary value and Cauchy / initial value problems
(IVPs). For hyperbolic PDEs the IVP is of particular interest since it
naturally forms a well-posed problem. Rather than specifying data just
on a spacelike hypersurface as in the IVP however, we can consider
additionally the initial boundary value problem. In this setup might
have, for example, a compact spatial domain and then choose suitable
boundary conditions on a timelike worldtube at the perimeter of that
domain. A third possibility, that we consider in the present work, is
the characteristic initial value problem (CIVP). Here data are
specified on characteristic surfaces of the equations under
consideration. In the context of general relativity (GR) these
surfaces are null slices.

In GR the CIVP has a long history which dates back at least to the
pioneering work by Bondi and collaborators on gravitational waves
---see~\cite{BonBurMet62,Sac62c}. The analysis in this work is based
on the observation that in coordinates (\emph{Bondi coordinates})
adapted to the geometry of outgoing light cones, the Einstein
equations give rise to a hierarchy of equations which can be formally
solved in sequence if certain pieces of data are provided. These ideas
were formalised in subsequent work by Sachs ---see~\cite{Sac62b}. The
CIVP was reconsidered by Newman \& Penrose in their more geometric
reformulation of the original analysis of gravitational radiation by
Bondi and collaborators ---see~\cite{NewPen62}, which also contains
the original formulation of the frame formulation of the Einstein
field equations known as the \emph{Newman-Penrose (NP) formalism}. The
work by Newman \& Penrose identifies particular components of the Weyl
tensor (expressed in terms of a null frame) as the key pieces of free
data to be specified on the characteristic hypersurfaces. The CIVP
setup also underlies subsequent work by Penrose on the properties of
massless spin fields and his approach of \emph{exact sequences of
  fields} ---see~\cite{Pen65a}. The common theme in this early work on
the CIVP in GR is that is mainly concerned with the structural
(i.e. algebraic) properties of the system of equations and does not
systematically address the issue of existence and uniqueness of
solutions.

Pioneering work on technical issues concerning the existence and
uniqueness of solutions to the characteristic problem for the Einstein
field equations can be found in the analysis of M\"uller zu Hagen and
Seifert~\cite{MulSei77}. These ideas were brought to fruition in the
work of Friedrich ---see~\cite{Fri81a}. There, it was shown that the
formulation of the characteristic problem by Newman \& Penrose implies
a symmetric hyperbolic evolution system for which known techniques
from the theory of PDEs can be applied. In particular, Friedrich shows
the local existence of solution near the intersection of the
characteristic hypersurfaces under the assumption of analyticity of
the freely specifiable data. This method was extended in subsequent
work to characteristic problems for a conformal representation of the
Einstein field equations (the \emph{the conformal Einstein field
  equations}) ---see~\cite{Fri81b,Fri82}. Among other things, this
work demonstrates the mathematical consistency of the work on the
nature of gravitational waves by Bondi and collaborators and Newman \&
Penrose. The formulation of the CIVP for the Einstein equations using
the NP formalism was further developed as a possible pathway towards
numerical simulations of the Einstein field equations~\cite{SteFri82}
---see also~\cite{IsaWelWin83} for an alternative formulation for
numerics using the Bondi approach to the characteristic problem, and
also influenced work on the nature and classification of caustics in
Relativity~\cite{FriSte83}.

A major milestone in the analysis of the problem came with the
influential work by Rendall on the \emph{reduction} of the CIVP to a
standard IVP~\cite{Ren90}, whose well-posedness is guaranteed by the
classical results of Choquet-Bruhat~\cite{Fou52}. In particular this
reduction provides an improved version of the local existence theorem
for the CIVP for the Einstein field equations which only requires a
finite level of differentiability of the initial data. Rendall's
method was subsequently used to obtain a smooth data version Friedrich
local existence result for the asymptotic CIVP for the conformal
Einstein field equations. Ideas arising from the CIVP underline and
permeate the fundamental work by Christodoulou \& Klainermann and on
the non-linear stability for the Einstein field
equations~\cite{ChrKla90,ChrKla93}. In particular, Christodoulou \&
Klainermann make use of a null frame formalism related to that of
Newman \& Penrose. Moreover, their analysis systematically exploits
the nonlinear structure of the Einstein field equations when expressed
in terms of such a null frame.

The structural properties identified in the analysis by Christodoulou
\& Klainermann paved the way for an improved local existence result
for the CIVP for the Einstein equations. Working in a gauge adopted
from Christodoulou's work on the formation of black
holes~\cite{Chr08}, which explicitly employs double-null coordinates,
such an improved result has been given by Luk~\cite{Luk12}. This work
guarantees an existence domain no longer restricted to a neighbourhood
of the intersection of the initial null hypersurfaces but that
stretches along them. Recently, Luk's analysis has been extended so
that the existence interval extends arbitrarily along the null
hypersurfaces and, thus, the solution contains a piece of infinity
---see \cite{LiZhu18}. An alternative approach to an improve local
existence result for the CIVP has been pursued by Chru{\'s}ciel and
collaborators ---see~\cite{ChoChrMar11,ChrPae12,CabChrWaf14} This
approach makes use of second order evolution equations for which well
developed theory of the CIVP exists ---see e.g.~\cite{Cag80,Cag81}.

Presently we are interested in two follow-up questions for which the
work of Rendall~\cite{Ren90} and Luk~\cite{Luk12} are most
relevant. Firstly, how do the aforementioned results look when
expressed in the language of the Newman-Penrose formalism? Following
long-term existence results in harmonic gauge~\cite{LinRod05}, it is
apparent that a variety of formulations of GR exhibit desirable
structure in their nonlinearities. Second, we are therefore curious as
to the robustness of this `null-structure' under changes of gauge.  We
hence give a formulation of the CIVP heavily influenced by that of
Stewart~\cite{Ste91}, and demonstrate for that formulation local
existence in a full neighbourhood of the initial null surfaces. In
first instance, the argument here provided gives an improved local
existence result along one of the initial hypersurfaces. This argument
can be adapted, \emph{mutatis mutandi}, to obtain improved local
existence along the other initial hypersurface ---see
Figure~\ref{Fig:ComparisonExistenceDomains}, (b). For conciseness, we
restrict our discussion to the neighbourhood of only one of the
hypersurfaces.  A tertiary aim in translating to the NP formalism is
to allow for the arguments and methods employed with Christodoulou's
formulation to be reformed for application elsewhere. Our interest in
understanding the structural properties of the NP field equations is
what drives us to consider the approach to an improved local existence
result for the CIVP pursued by Luk rather than the one followed by
Chru{\'s}ciel and collaborators. In the future we hope that this will
permit us to obtain similar results for the conformal field
equations~\cite{Fri81b,Fri82}. Regarding the question of robustness of
the nonlinearities, our work serves only as a stepping stone for a
more detailed investigation. Nevertheless it is worth stressing that
our gauge differs from that used elsewhere, and that the
nonlinearities of the equations do retain sufficient structure for us
to successfully manage.

\begin{figure}[t]
\centering
\includegraphics[width=\textwidth]{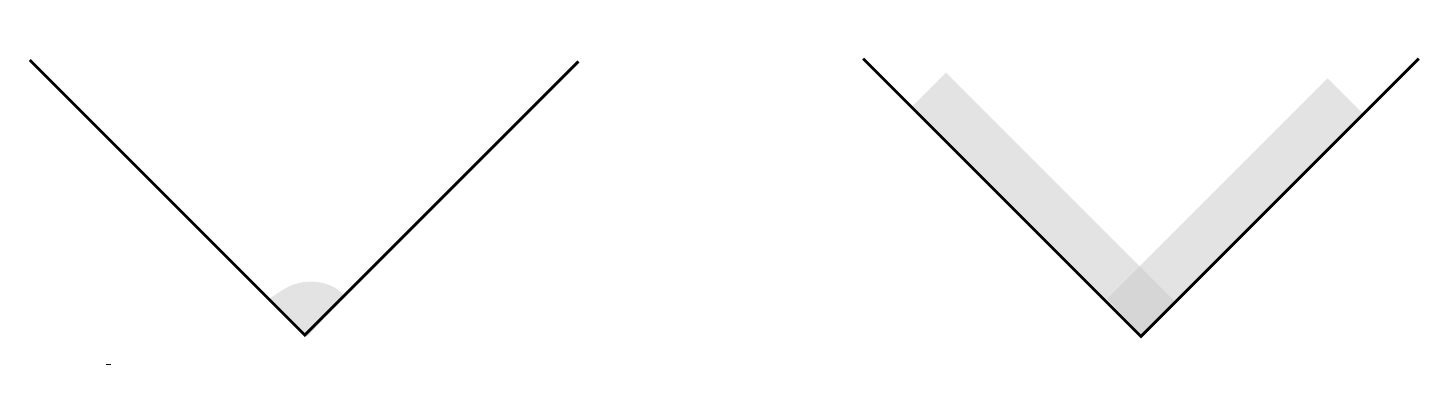}
\put(-415,120){(a)} \put(-180,120){(b)}
\put(-43,50){$\mathcal{N}_\star$} \put(-290,50){$\mathcal{N}_\star$}
\put(-390,50){$\mathcal{N}'_\star$}
\put(-150,50){$\mathcal{N}'_\star$} \put(-338,8){$\mathcal{S}_\star$}
\put(-96,8){$\mathcal{S}_\star$}
\caption{Comparison of the existence domains for the characteristic
  problem: (a) existence domain using Rendall's strategy based on the
  reduction to a standard Cauchy problem; (b) existence domain using
  Luk's strategy ---in principle, the long side of the rectangles
  extends for as much as one has control on the initial data.}
\label{Fig:ComparisonExistenceDomains}
\end{figure}
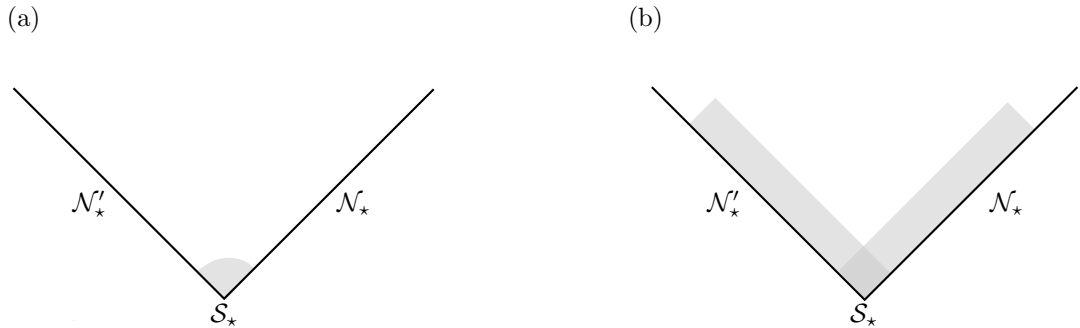

\subsubsection*{Notation and conventions}

We take~$\{ {}_a , {}_b , {}_c , \dots\}$ to denote abstract tensor
indices whereas~$\{_\mu , _\nu , _\lambda , \dots\}$ will be used as
spacetime coordinate indices with the values~${ 0, \dots, 3 }$. Our
conventions for the curvature tensors are fixed by the relation
\begin{align}
(\nabla_a \nabla_b -\nabla_b \nabla_a) v^c = R^c{}_{dab} v^d.
\end{align}
We make systematic use of the NP formalism as described, for example,
in~\cite{Ste91,PenRin84}. In particular, the signature of Lorentzian
metrics is~$(+---)$. Many of our derivations, although
straightforward, are fairly lengthy, so we have included in
Appendix~\ref{App:NP-formalism} a complete summary of the equations of
the NP-formalism, highlighting the simplifications that occur with our
particular gauge. We recommend that the reader keep a copy of the
appendix to hand as they read the paper.

\section{The geometry of the problem}

Let~($\mathcal{M}, \bmg$) denote a vacuum spacetime
satisfying~$R_{ab}=0$, where~$\mathcal{M}$ is a $4$-dimensional
manifold with boundary and an edge. The boundary consists of two null
hypersurface:~$\mathcal{N}_{\star}$, the outgoing null
hypersurface;~$\mathcal{N}_{\star}^{\prime}$, the incoming null
hypersurface with non-empty intersection~$\mathcal{S}_{\star}\equiv
\mathcal{N}_{\star}\cap\mathcal{N}_{\star}^{\prime}$. For concreteness
we will assume that~$\mathcal{S}_{\star}\approx\mathbb{S}^2$.

Given a neighbourhood~$\mathcal{U}$ of~$\mathcal{S}_{\star}$, one can
introduce coordinates~$x=(x^{\mu})$ with~$x^0=v$ and~$x^1=u$ such
that, at least in a neighbourhood of~$\mathcal{S}_{\star}$ one can
write
\begin{subequations}
\begin{align*}
  \mathcal{N}_{\star}=\{p\in\mathcal{U}\mid u(p)=0\},\ \ \
  \mathcal{N}_{\star}^{\prime}=\{p\in\mathcal{U}\mid v(p)=0\}.
\end{align*}
\end{subequations}
Given suitable data on~$(\mathcal{N}_{\star} \cup
\mathcal{N}_{\star}^{\prime}) \cap\mathcal{U}$ we are interested in
making statements about the existence and uniqueness of solutions to
the vacuum Einstein field equations of the aforementioned type on some
open set
\begin{subequations}
\begin{align*}
\mathcal{V}\subset\{p\in\mathcal{U}\mid u(p)\geq0, v(p)\geq0\}
\end{align*}
\end{subequations}
which we identify with a subset of the future domain of
dependence,~$D^{+}(\mathcal{N}_{\star}\cup\mathcal{N}_{\star}^{\prime})$,
of~$\mathcal{N}_{\star}\cup\mathcal{N}_{\star}^{\prime}$.

\subsection{Construction of the gauge: Stewart's approach}
\label{Section:StewartGauge}

We will ultimately be concerned with existence and uniqueness of
solutions, but, as is common in such constructions, it is useful to
start by assuming existence in order to give a concrete PDE
formulation of the problem. In this section we thus briefly review the
gauge choice. In the rest of this article we will call this
construction \emph{Stewart's gauge}.

\subsubsection{Coordinates}

In the following it will be convenient to regard the 2-dimensional
surface~$\mathcal{S}_{\star}$ as a submanifold of a spacelike
hypersurface~$\mathit{S}$. The subsequent discussion will be
restricted to the future of~$\mathit{S}$. As
$\mathcal{S}_{\star}\approx\mathbb{S}^2$, one has that
$\mathcal{S}_{\star}$ divides~$\mathit{S}$ in two regions ---the
interior of~$\mathcal{S}_{\star}$ and the exterior
of~$\mathcal{S}_{\star}$. Now, consider a foliation of $\mathit{S}$ by
2-dimensional surfaces with the topology of~$\mathbb{S}^2$ which
includes~$\mathcal{S}_{\star}$. At each of the 2-dimensional surfaces
we assume there pass two null hypersurfaces. Further, we assume that:
\begin{enumerate}[i).]
\item one of these hypersurfaces has the property that the projection
  of the tangent vectors of their generators at~$\mathcal{S}_{\star}$
  point~$\mathit{outwards}$ ---we call these null
  hypersurfaces~$\mathit{outgoing\ light\ cones}$;

\item one of these hypersurfaces has the property that the projection
  of the tangent vectors of their generators at~$\mathcal{S}_{\star}$
  point $\mathit{inwards}$ ---we call these null
  hypersurfaces~$\mathit{ingoing\ light\ cones}$.
\end{enumerate}

Thus, as least close to~$\mathit{S}$ one obtains a 1-parameter family
of outgoing null hypersurface~$\mathcal{N}_{u}$ and a 1-parameter
family of ingoing null hypersurface~$\mathcal{N}_{v}^{\prime}$. One
can then define scalar fields~$u$ and~$v$ by the requirements,
respectively, that~$u$ is constant on each of the $\mathcal{N}_{u}$
and~$v$ is constant on each~$\mathcal{N}_{v}^{\prime}$. In particular,
we assume that~$\mathcal{N}_{0}=\mathcal{N}_{\star}$
and~$\mathcal{N}_{0}^{\prime}=\mathcal{N}_{\star}^{\prime}$. Following
standard usage, we call~$u$ a~$retarded\ time$ and~$v$ an
$advanced\ time$. We use the notation~$\mathcal{N}_u(v_1,v_2)$ to
denote the part of the hypersurface~$\mathcal{N}_u$ with $v_1\leq
v\leq v_2$. Likewise~$\mathcal{N}'_v(u_1,u_2)$ has a similar
definition. We denote the sphere intersected by~$\mathcal{N}_{u}$
and~$\mathcal{N}_{v}'$ by~$\mathcal{S}_{u,v}$. We define the region
\begin{align}
\bigcup_{0\leq v'\leq v, 0\leq u'\leq u}\mathcal{S}_{u',v'}
\end{align}
as~$\mathcal{D}_{u,v}$. We also define the time function
\begin{align}
t\equiv u+v,\label{eqn:time_coord_defn}
\end{align}
and the {\it truncated causal diamond},
\begin{align}
  \mathcal{D}_{u,v}^{\,\tilde{t}}\equiv\mathcal{D}_{u,v}
  \cap\{t\leq \tilde{t}\},
  \label{eqn:truncated_diamond_defn}
\end{align}
which will be used frequently throughout our arguments.

The scalar fields~$u$ and~$v$ introduced in the previous paragraph
will be used as coordinates in a neighbourhood
of~$\mathcal{S}_{\star}$. To complete the coordinate system, consider
arbitrary coordinates~$(x^{\mathcal{A}})$ on~$\mathcal{S}_{\star}$,
with the index~$\mathcal{A}$ taking the values~$2$, $3$. These
coordinates are then propagated into~$\mathcal{N}_{\star}$ by
requiring them to be constant along the generators
of~$\mathcal{N}_{\star}$. Once coordinates have been defined
on~$\mathcal{N}_{\star}$, one can propagate them into~$\mathcal{V}$ by
requiring them to be constant along the generators of
each~$\mathcal{N}_{v}^{\prime}$. In this manner one obtains a
coordinate system~$(x^{\mu})=(v,\ u,\ x^{\mathcal{A}})$
in~$\mathcal{V}$.

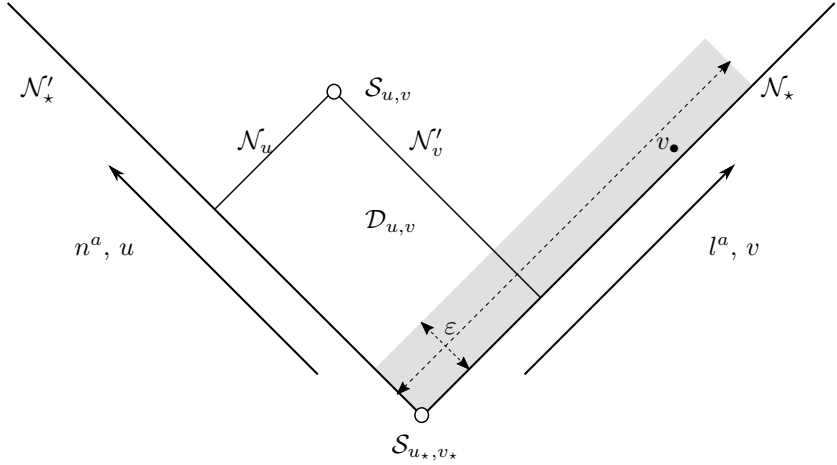
\begin{figure}[t]
\centering
\includegraphics[width=0.8\textwidth]{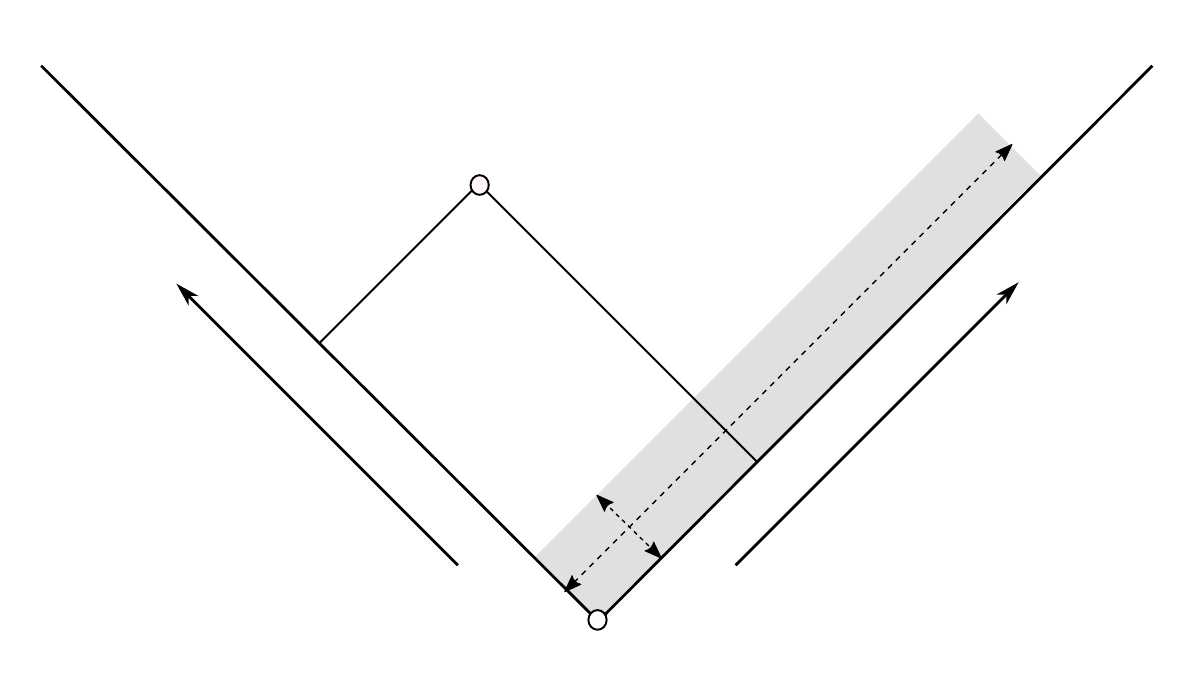}
\put(-60,80){$l^a$, $v$}
\put(-300,80){$n^a$, $u$}
\put(-40,140){$\mathcal{N}_\star$}
\put(-320,140){$\mathcal{N}'_\star$}
\put(-180,5){$\mathcal{S}_{u_\star,v_\star}$}
\put(-190,140){$\mathcal{S}_{u,v}$}
\put(-190,90){$\mathcal{D}_{u,v}$}
\put(-160,50){$\varepsilon$}
\put(-80,120){$v_\bullet$}
\put(-238,120){$\mathcal{N}_u$}
\put(-173,120){$\mathcal{N}'_v$}
\caption{Setup for Stewart's gauge. The construction makes use of a
  double null foliation of the future domain of dependence of the
  initial hypersurface~$\mathcal{N}_\star\cup\mathcal{N}_\star'$. The
  coordinates and NP null tetrad are adapted to this geometric
  setting. The analysis in this article is focused on the arbitrarily
  thin grey rectangular domain along the
  hypersurface~$\mathcal{N}_\star$. The argument can be adapted, in a
  suitable manner, to a similar rectangle
  along~$\mathcal{N}_\star'$. See the main text for the definitions of
  the various regions and objects.}
\label{Fig:CharacteristicSetup}
\end{figure}

\subsubsection{The NP frame}

To construct a null NP tetrad we choose vector fields~$l^a$ and~$n^a$
to be tangent to the generators of~$\mathcal{N}_{u}$
and~$\mathcal{N}_{v}^{\prime}$ respectively. Further we require them
to be normalised according to
\begin{align}
g_{ab}l^an^b=1. \nonumber
\end{align}
The latter normalisation condition is preserved under the boost,
\begin{align*}
l^a\mapsto\varsigma l^a,\ \ n^a\mapsto\varsigma^{-1}n^a,
\ \ \varsigma\in\mathbb{R}.
\end{align*}
This freedom can be used to set
\begin{subequations}
\begin{align}
n_a=\nabla_av. \nonumber
\end{align}
\end{subequations}
This requirement still leaves some freedom left as one can choose a
relabelling of the form~$v\mapsto V(v)$. Next, we choose the complex
vector fields~$m^a$ and~$\bar{m}^a$ so that they are tangent to the
surfaces~$\mathcal{S}_{u,v}$ and satisfy the conditions
\begin{align*}
g_{ab}m^a\bar{m}^b=-1, \ \ g_{ab}m^am^b=0.
\end{align*}
There is still the freedom to perform a spin
\begin{align*}
m^a\mapsto e^{i\theta}m^a, \ \ \theta\in\mathbb{R}
\end{align*}
at each point.

\begin{remark}
\label{Remark_1}
{\em It can be verified that the vectors
  $\{l^a,\ n^a,\ m^a,\ \bar{m}^a \}$ constructed in the previous
  paragraphs satisfy
\begin{align*}
g_{ab}l^am^b=g_{ab}n^am^b=g_{ab}l^a\bar{m}^b=g_{ab}n^a\bar{m}^b=0.
\end{align*}

Now, observing that, by construction, on the generators of each null
hypersurface~$\mathcal{N}_{v}^{\star}$ only the coordinate $u$ varies,
one has that
\begin{align*}
n^\mu\bmpartial_\mu=Q\bmpartial_u,
\end{align*}
where~$Q$ is a real function of the position. Furthermore, since the
vector~$l^a$ is tangent to the generators of each~$\mathcal{N}_{u}$
and~$l^an_a=l^a\nabla_av=1$, one has that
\begin{align*}
l^\mu\bmpartial_\mu =\bmpartial_v +C^{\mathcal{A}}\bmpartial_{\mathcal{A}},
\end{align*}
where, again, the components~$C^{\mathcal{A}}$ are real functions of
the position. By construction, the coordinates~$(x^{\mathcal{A}})$ do
not vary along the generators of~$\mathcal{N}_{\star}$-that is, one
has that~$l^a\nabla_ax^{\mathcal{A}}=0$. Accordingly, one has that
\begin{align*}
C^{\mathcal{A}}=0 \quad \mbox{on}\quad \mathcal{N}_{\star}.
\end{align*}
Finally, since~$m^a$ and~$\bar{m}^a$ span the tangent space of each
surface~$\mathcal{S}_{u,v}$ one has that
\begin{align*}
m^\mu\bmpartial_\mu =P^{\mathcal{A}} \bmpartial_{\mathcal{A}},
\end{align*}
where the coefficients~$P^{\mathcal{A}}$ are complex functions.}
\end{remark}

Summarising, we make the following choice:

\begin{assumption}[\textbf{\em Stewart's choice of the components of the
      frame}]\label{Assumption:Stewarts_Frame}
  {\em On $\mathcal{V}$ one can find a Newman-Penrose frame
  $\{l^a,\ n^a,\ m^a,\ \bar{m}^a \}$ of the form:
\begin{subequations}
\begin{align*}
  \bml=\bmpartial_v +C^{\mathcal{A}}\bmpartial_{\mathcal{A}},
  \qquad \bmn=Q\bmpartial_u,
  \qquad \bmm=P^{\mathcal{A}} \bmpartial_{\mathcal{A}}.
\end{align*}
\end{subequations}}
\end{assumption}

\begin{remark}
{\em In view of the normalisation condition~$g_{ab}m^a\bar{m}^b=-1$,
  there are only 3 real functions involved in
  the~$P^{\mathcal{A}}$'s. Thus,~$Q$, $C^{\mathcal{A}}$ together
  with~$P^{\mathcal{A}}$ give six scalar fields describing the
  metric. Thus the components~$(g^{\mu\nu})$ of the contravariant form
  of the metric~$\bmg$ are of the form
\begin{align*}
(g^{\mu\nu})=\left(\begin{array}{ccc}
0&Q&0 \\
Q&0&QC^{\mathcal{A}} \\
0&QC^{\mathcal{A}}&\sigma^{\mathcal{A}\mathcal{B}} \\
\end{array}\right),
\end{align*}
where
\begin{align*}
\sigma^{\mathcal{A}\mathcal{B}}\equiv
-(P^{\mathcal{A}}\bar{P^{\mathcal{B}}}+\bar{P^{\mathcal{A}}}P^{\mathcal{B}}).
\end{align*}
Here and in what follows~$\bmsigma$ is the induced metric
on~$\mathcal{S}_{u,v}$, and has contravariant
components~$\sigma^{\mathcal{A}\mathcal{B}}$ defined in the standard
manner. Note that care is needed to distinguish~$\sigma$, the NP
connection coefficient, from this quantity. From the expression, we
can compute that $l_{\mu}dx^{\mu}=Q^{-1}du$,
$\sigma_{\mathcal{A}\mathcal{B}}P^{\mathcal{A}}P^{\mathcal{B}}=0$,
$\sigma_{\mathcal{A}\mathcal{B}}P^{\mathcal{A}}\bar
P^{\mathcal{B}}=-1$
and~$-\p_{\mathcal{A}}C^{\mathcal{A}}=\bar{m}_{\mathcal{A}}\delta
C^{\mathcal{A}}+m_{\mathcal{A}}\bar\delta C^{\mathcal{A}}$ directly. }
\end{remark}

\begin{remark}
{\em On~$\mathcal{N}_{\star}^{\prime}$ one has
  that~$\bmn=Q\bmpartial_u$. As the coordinates~$(x^{\mathcal{A}})$
  are constant along the generators of~$\mathcal{N}_{\star}$
  and~$\mathcal{N}_{\star}^{\prime}$, it follows that
  on~$\mathcal{N}_{\star}^{\prime}$ the coefficient~$Q$ is only a
  function of u. Thus, without loss of generality one can
  parameterise~$u$ so as to set~$Q=1$
  on~$\mathcal{N}_{\star}^{\prime}$.}
\end{remark}

\subsection{Analysis of the NP commutators}

In this subsection we analyse some simple consequences of the NP frame
of Assumption~\ref{Assumption:Stewarts_Frame} and the NP commutator
equations~(\ref{NPCommutator1})-(\ref{NPCommutator4}). In particular,
we exploit the fact that given a choice of NP frame, the evaluation of
the NP commutators on the coordinates gives rise to two different
types of equations, namely i). conditions on the spin connection
coefficients, and ii). equations for the coefficients of the frame.
In what follows we analyse these two classes of equations. For future
use observe that from the definition of the NP
frame~$\{l^a,\ n^a,\ m^a,\ \bar{m}^a \}$ in
Assumption~\ref{Assumption:Stewarts_Frame} it readily follows that,
\begin{subequations}
\begin{align}
  Dv&=1, &\quad \Delta v&=0,  &\quad\delta v&=0, &\quad \bar{\delta}v&=0,
  \label{NPcoordinate1}\\
  Du&=0, &\quad \Delta u&=Q, &\quad \delta u&=0,  &\quad \bar{\delta}u&=0,
  \label{NPcoordinate2}\\
 Dx^{\mathcal{A}}&=C^{\mathcal{A}}, &\quad \Delta x^{\mathcal{A}}&=0,
 &\quad \delta x^{\mathcal{A}}&=P^{\mathcal{A}}, &\quad
 \bar{\delta}x^{\mathcal{A}}&=\bar{P}^{\mathcal{A}}.
 \label{NPcoordinate3}
\end{align}
\end{subequations}

\subsubsection{Spin connection coefficients}

Direct inspection of the NP
commutators~(\ref{NPCommutator1})-(\ref{NPCommutator4}) applied to the
coordinates~$(v,\ u,\ x^2,\ x^3)$ taking into account
(\ref{NPcoordinate1})-(\ref{NPcoordinate3}) yields on~$\mathcal{V}$
the conditions,
\begin{align*}
 \kappa=\nu=0, \quad \gamma+\bar{\gamma}=0, \quad
 \rho=\bar{\rho},  \quad \mu=\bar{\mu}, \quad
 \pi=\alpha+\bar{\beta}.
\end{align*}
We will see that these gauge conditions can be refined still further.

\smallskip
\noindent
\textbf{Fixing the rotation freedom.} The set up of frame vectors
under Assumption~\ref{Assumption:Stewarts_Frame} allows the freedom of
a rotation
\begin{align*}
m^a\mapsto m^{\prime a}=e^{i\theta}m^a.
\end{align*}
The latter, in turn, implies the transformation
\begin{align*}
\gamma-\bar{\gamma}\mapsto\gamma^{\prime}-\bar{\gamma}^{\prime}
=\gamma-\bar{\gamma}-i\Delta\theta.
\end{align*}
Accordingly, by requiring~$\theta$ to satisfy the equation
\begin{align}
\Delta\theta=i(\bar{\gamma}-\gamma)
\label{Deltatheta}
\end{align}
it is always possible to assume that~$\bar{\gamma}-\gamma=0$, which,
together with the condition~$\gamma+\bar{\gamma}=0$ allows us to
set~$\gamma=0$ on~$\mathcal{V}$. A similar computation shows that
\begin{align*}
\epsilon-\bar{\epsilon}\mapsto\epsilon^{\prime}
-\bar{\epsilon}^{\prime}=\epsilon-\bar{\epsilon}+i
D\theta.
\end{align*}
This equation can be used to set~$\epsilon-\bar{\epsilon}=0$
on~$\mathcal{N}_{\star}$. Also, after solving this equation, the
result~$\theta$ on~$\mathcal{N}_{\star}$ can be the initial value of
equation \eqref{Deltatheta}. The value of Q on~$\mathcal{N}_{\star}$
can be propagated from~$\mathcal{S}_{\star}$ using the transport
equation,
\begin{align*}
DQ=-(\epsilon+\bar{\epsilon})Q=-2\epsilon Q
\end{align*}
that is,
\begin{align*}
\p_vQ=-2\epsilon Q.
\end{align*}
Summarising, we have the following gauge restriction, which we employ
exclusively in what follows:

\begin{lemma}[\textbf{\em properties of the connection coefficients in
  Stewart's gauge}]\label{Lemma1}
The NP frame of Assumption~\ref{Assumption:Stewarts_Frame} can be
chosen such that
\begin{subequations}
\begin{align}
& \kappa=\nu=\gamma=0, \label{spinconnection1}\\
& \rho=\bar{\rho},\ \ \mu=\bar{\mu}, \label{spinconnection2}\\
& \pi=\alpha+\bar{\beta} \label{spinconnection3}
\end{align}
\end{subequations}
on~$\mathcal{V}$ and, furthermore, with
\begin{align*}
\epsilon-\bar{\epsilon}=0\ \ \ \textrm{on}\ \ \
\mathcal{V}\cap\mathcal{N}_{\star}. 
\end{align*}
\end{lemma}

\subsubsection{Equations for the frame coefficients}

Taking into account the conditions on the spin connection coefficients
given by~(\ref{spinconnection1})-(\ref{spinconnection3}), it follows
that the remaining commutators yield the equations
\begin{subequations}
\begin{align}
  &\Delta C^{\mathcal{A}}=-(\bar{\tau}+\pi)P^{\mathcal{A}}
  -(\tau+\bar{\pi})\bar{P}^{\mathcal{A}}, \label{framecoefficient1} \\
  &\Delta P^{\mathcal{A}}=-\mu P^{\mathcal{A}}
  -\bar{\lambda}\bar{P}^{\mathcal{A}}, \label{framecoefficient2} \\
  &DP^{\mathcal{A}}-\delta C^{\mathcal{A}}=
  (\rho+\epsilon-\bar{\epsilon})P^{\mathcal{A}}
  +\sigma\bar{P}^{\mathcal{A}}, \label{framecoefficient3} \\
&DQ=-(\epsilon+\bar{\epsilon})Q,  \label{framecoefficient4}\\
  &\bar{\delta}P^{\mathcal{A}}-\delta\bar{P}^{\mathcal{A}}=
  (\alpha-\bar{\beta})P^{\mathcal{A}}
  -(\bar{\alpha}-\beta)\bar{P}^{\mathcal{A}}, \label{framecoefficient5} \\
&\delta Q=(\tau-\bar{\pi})Q. \label{framecoefficient6}
\end{align}
\end{subequations}

\begin{remark}
{\em Equations~(\ref{framecoefficient1})-(\ref{framecoefficient2})
  allow us to evolve the frame coefficients~$C^{\mathcal{A}}$
  and~$P^{\mathcal{A}}$ off of the null
  hypersurface~$\mathcal{N}_{\star}^{\prime}$.
  Equations~(\ref{framecoefficient3})-(\ref{framecoefficient4}) allow
  evolution of the coefficients~$Q$ and~$P^{\mathcal{A}}$ along the
  null generators of~$\mathcal{N}_{\star}$.
  Finally~(\ref{framecoefficient5})-(\ref{framecoefficient6}) provide
  constraints for~$Q$ and~$P^{\mathcal{A}}$ on the
  spheres~$\mathcal{S}_{u,v}$.}
\end{remark}

\section{The formulation of the CIVP}

In this section we analyse general aspects of the CIVP for the vacuum
Einstein field equations on the null
hypersurfaces~$\mathcal{N}_{\star}$
and~$\mathcal{N}_{\star}^{\prime}$. The hierarchical structure allows
the identification of the basic reduced initial data set $r_\star$ from
which the full initial data
on~$\mathcal{N}_{\star}\cup\mathcal{N}_{\star}^{\prime}$ can be
computed.

\begin{lemma}[\textbf{\em freely specifiable data for the CIVP}]
  \label{Lemma:FreeDataCIVP} Working in the gauge given by
  Assumption~\ref{Assumption:Stewarts_Frame} and Lemma~\ref{Lemma1},
  initial data for the vacuum Einstein field equations
  on~$\mathcal{N}_{\star}\cup\mathcal{N}_{\star}^{\prime}$ can be
  computed (near~$\mathcal{S}_{\star}$) from the reduced data
  set~$\mathbf{r}_\star$ consisting of:
\begin{align*}
  &\Psi_0, \; \epsilon+\bar\epsilon \quad \mbox{\textrm{on}}
  \quad \mathcal{N}_{\star},\nonumber\\
  &\Psi_4\ \ \mbox{\textrm{on}}\ \ \mathcal{N}_{\star}^{\prime}, \nonumber\\
  &\lambda,\ \ \sigma,\ \ \mu,\ \ \rho, \ \ \pi, \ \
  P^{\mathcal{A}}\ \ \mbox{\textrm{on}}\ \ \mathcal{S}_{\star}. \nonumber
\end{align*}
\end{lemma}

\begin{proof}
The proof follows by inspection of the various intrinsic equations
on~$\mathcal{N}_\star$, $\mathcal{N}_\star'$ and~$\mathcal{S}_\star$.

\smallskip
\noindent
\textbf{Data on $\mathcal{S}_\star$.} Since~$P^{\mathcal{A}}$ are
given, the operators~$\delta$ and~$\bar{\delta}$ are well defined
on~$\mathcal{S}_\star$ and intrinsic to this $2$-dimensional
hypersurface. From the definition of the connection
coefficients~$\alpha$ and~$\beta$ it follows that the inner connection
of~$\mathcal{S}_\star$ is described by the
combination~$\alpha-\bar{\beta}$. This is readily computable from the
data~$P^{\mathcal{A}}$ on~$\mathcal{S}_\star$. Thus,
using~$\alpha+\bar\beta=\pi$, one can compute~$\alpha$
and~$\beta$. Noting that~$Q=1$
on~$\mathcal{S}_\star\subset\mathcal{N}_\star'$, we obtain
that~$\pi=\bar\tau$ from~\eqref{framecoefficient6}. Then we obtain all
the values of connection coefficients on~$\mathcal{S}_\star$. Thus,
the constraint equations~\eqref{structureeq17}, \eqref{structureeq10},
\eqref{structureeq14} of the structure equations can be used to
compute the value of~$\Psi_1$, $\Psi_2$, $\Psi_3$
on~$\mathcal{S}_\star$. With that, all initial data for the connection
coefficients and Weyl curvature on~$\mathcal{S}_\star$ have been
obtained.

\smallskip
\noindent
\textbf{Data on $\mathcal{N}_\star'$.}  On the incoming null
hypersurface $\mathcal{N}_\star'$ we can obtain that~$Q=1$ leads
to~$\tau=\bar{\pi}$ from equation~\eqref{framecoefficient6}
and~$\Delta=\p_u$. Making use of the structure
equations~\eqref{structureeq7} and~\eqref{structureeq15}, which can be
reduced by the gauge condition, namely
\begin{align*}
&\frac{\p\mu}{\p u}=-\lambda\bar{\lambda}-\mu^2, \\
&\frac{\p\lambda}{\p u}=-\Psi_4-2\lambda\mu,
\end{align*}
we can obtain the value of~$\mu$ and~$\lambda$
on~$\mathcal{N}_\star'$. Then the frame coefficients~$P^{\mathcal{A}}$
on~$\mathcal{N}_\star'$ are computed using
equation~\eqref{framecoefficient2} which takes the form
\begin{align*}
\frac{\p P^{\mathcal{A}}}{\p u}=-\mu P^{\mathcal{A}}
-\bar{\lambda}\bar{P}^{\mathcal{A}}.
\end{align*}
Thus we can compute the~$\delta$-direction derivative
on~$\mathcal{N}_\star'$. Solving the structure
equations~\eqref{structureeq4},~\eqref{structureeq11} with the Bianchi
identity equation~\eqref{Bianchi4}, namely
\begin{align*}
&-\frac{\p\alpha}{\p u}=\Psi_3+\beta\lambda+\alpha\bar\mu
+\lambda\tau \\
&-\frac{\p\beta}{\p u}=\alpha\bar\lambda+\beta\mu+\mu\tau, \\
&\frac{\p\Psi_3}{\p u}-P^{\mathcal{A}}\frac{\p\Psi_4}
{\p x^{\mathcal{A}}}=(4\beta-\tau)\Psi_4-4\mu\Psi_3,
\end{align*}
together we can compute the value of~$\alpha$, $\beta$ and~$\Psi_3$
on~$\mathcal{N}_\star'$. Then equation~\eqref{framecoefficient1}
\begin{align*}
\frac{\p C^{\mathcal{A}}}{\p u}
=-(\bar\tau+\pi)P^{\mathcal{A}}-(\tau+\bar\pi)\bar P^{\mathcal{A}}
\end{align*}
reveals the value of the frame coefficients~$C^{\mathcal{A}}$
on~$\mathcal{N}_\star'$. With the above information at hand one can
use equations~\eqref{structureeq1},~\eqref{structureeq9},
\eqref{structureeq18} and~\eqref{Bianchi5}:
\begin{align*}
  &\frac{\p\epsilon}{\p u}=-\Psi_2-\beta\pi-\alpha\bar\pi-\alpha\tau
  -\pi\tau-\beta\bar\tau,\\
  &P^{\mathcal{A}}\frac{\p\tau}{\p x^{\mathcal{A}}}-\frac{\p\sigma}{\p u}
  =\bar\lambda\rho+\mu\sigma-\bar\alpha\tau+\beta\tau+\tau^2,\\
  &\bar P^{\mathcal{A}}\frac{\p\tau}{\p x^{\mathcal{A}}}-\frac{\p\rho}{\p u}
  =\Psi_2+\bar\mu\rho+\lambda\sigma+\alpha\tau-\bar\beta\tau+\tau\bar\tau,\\
  &\frac{\p\Psi_2}{\p u}-P^{\mathcal{A}}\frac{\p\Psi_3}{\p x^{\mathcal{A}}}
  =\sigma\Psi_4+2(\beta-\tau)\Psi_3-3\mu\Psi_2
\end{align*}
to compute the value of~$\epsilon$, $\sigma$, $\rho$ and~$\Psi_2$
on~$\mathcal{N}_\star'$. The Bianchi identity equation~\eqref{Bianchi8}
\begin{align*}
\frac{\p\Psi_1}{\p u}-P^{\mathcal{A}}\frac{\p\Psi_2}
{\p x^{\mathcal{A}}}=-2\mu\Psi_1-3\tau\Psi_2+2\sigma\Psi_3,
\end{align*}
provides the value of~$\Psi_1$ on~$\mathcal{N}_\star'$. With the
results above, we can then compute the value of~$\Psi_0$ from
equation~\eqref{Bianchi2}
\begin{align*}
\frac{\p\Psi_0}{\p u}-P^{\mathcal{A}}\frac{\p\Psi_1}
{\p x^{\mathcal{A}}}=-\mu\Psi_0-2(2\tau+\beta)\Psi_1+3\sigma\Psi_2.
\end{align*}

\smallskip
\noindent
\textbf{Data on~$\mathcal{N}_\star$.} From
equation~\eqref{framecoefficient4} one has
that~$\p_vQ=-(\epsilon+\bar\epsilon)Q$ so that, using the value
of~$Q$ at~$\mathcal{S}_\star$ one can compute the value of~$Q$
on~$\mathcal{N}_\star$. The structure equations~\eqref{structureeq6}
and~\eqref{structureeq13} give 
\begin{align*}
&\frac{\p\sigma}{\p v}=\Psi_0+3\epsilon\sigma-\bar\epsilon\sigma+2\rho\sigma, \\
&\frac{\p\rho}{\p v}=2\epsilon\rho+\rho^2+\sigma\bar\sigma.
\end{align*}
Solving these last equations one can obtain the value of~$\sigma$
and~$\rho$ on $\mathcal{N}_\star$. Then the value of~$P^{\mathcal{A}}$
on~$\mathcal{N}_\star$ can be computed using
equation~\eqref{framecoefficient3} which in the present setting takes
the form
\begin{align*}
\frac{\p P^{\mathcal{A}}}{\p v}=\rho P^{\mathcal{A}}
+\sigma\bar P^{\mathcal{A}}
\end{align*}
Then the structure equations~\eqref{structureeq5},
\eqref{structureeq12} and the Bianchi identity~\eqref{Bianchi1},
namely,
\begin{align*}
  &P^{\mathcal{A}}\frac{\p\epsilon}{\p x^{\mathcal{A}}}-\frac{\p\beta}{\p v}=-\Psi_1
  +\bar\alpha\epsilon+\beta\bar\epsilon-\epsilon\bar\pi-\beta\rho
  -\alpha\sigma-\pi\sigma, \\
  &\bar P^{\mathcal{A}}\frac{\p\epsilon}{\p x^{\mathcal{A}}}-\frac{\p\alpha}{\p v}
  =2\alpha\epsilon+\bar\beta\epsilon-\alpha\bar\epsilon-\epsilon\pi
  -\alpha\rho-\pi\rho-\beta\bar\sigma \\
  &\bar P^{\mathcal{A}}\frac{\p\Psi_0}{\p x^{\mathcal{A}}}-\frac{\p\Psi_1}{\p v}
  =(4\alpha-\pi)\Psi_0-2(2\rho+\epsilon)\Psi_1 .
\end{align*}
provide us the value of~$\alpha$, $\beta$ and~$\Psi_1$
on~$\mathcal{N}_\star$. Next, the structure
equation~\eqref{structureeq2} which takes the form
\begin{align*}
  \frac{\p\tau}{\p v}=\Psi_1+\bar\pi\rho+\pi\sigma
  +\epsilon\tau-\bar\epsilon\tau+\rho\tau+\sigma\bar\tau
\end{align*}
gives us the value of~$\tau$ on~$\mathcal{N}_\star$. Similarly, the
structure equations~\eqref{structureeq8}, \eqref{structureeq16} and
the Bianchi identity equation~\eqref{Bianchi5}
\begin{align*}
  & P^{\mathcal{A}}\frac{\p\pi}{\p x^{\mathcal{A}}}-\frac{\p\mu}{\p v}=-\Psi_2
  +\epsilon\mu+\bar\epsilon\mu+\bar\alpha\pi-\beta\pi-\pi\bar\pi
  -\mu\rho-\lambda\sigma, \\
  &\bar P^{\mathcal{A}}\frac{\p\pi}{\p x^{\mathcal{A}}}-\frac{\p\lambda}{\p v}=
  3\epsilon\lambda-\bar\epsilon\lambda-\alpha\pi+\bar\beta\pi-\pi^2-\lambda\rho
  -\mu\bar\sigma, \\
  &\frac{\p\Psi_2}{\p v}-\bar P^{\mathcal{A}}\frac{\p\Psi_1}{\p x^{\mathcal{A}}}
  =-\lambda\Psi_0+2(\pi-\alpha)\Psi_1+3\rho\Psi_2
\end{align*}
give us the value of~$\mu$, $\lambda$ and~$\Psi_2$
on~$\mathcal{N}_\star$. Next, the Bianchi identity
equations~\eqref{Bianchi7} and~\eqref{Bianchi3}
\begin{align*}
  &\frac{\p\Psi_3}{\p v}-\bar P^{\mathcal{A}}\frac{\p\Psi_2}{\p x^{\mathcal{A}}}
  =2(\rho-\epsilon)\Psi_3+3\pi\Psi_2-2\lambda\Psi_1, \\
  &\bar P^{\mathcal{A}}\frac{\p\Psi_3}{\p x^{\mathcal{A}}}
  -\frac{\p\Psi_4}{\p v}=(4\epsilon-\rho)\Psi_4
  -2(2\pi+\alpha)\Psi_3+3\lambda\Psi_2,
\end{align*}
show us the value of~$\Psi_3$ and~$\Psi_4$ on~$\mathcal{N}_\star$.
Finally, we have obtained all the initial values
on~$\mathcal{N}_{\star}\cup\mathcal{N}_{\star}^{\prime}$ from the
reduced data set~$r_\star$.
\end{proof}

\section{Rendall's local existence theory}

In order to apply the basic local existence theory for the CIVP as
formulated by Rendall~\cite{Ren90} (see also Section 12.5
of~\cite{CFEBook}), one has to extract a suitable symmetric hyperbolic
evolution system from the Einstein field equations. The gauge
introduced in Section~\ref{Section:StewartGauge} allows us to perform
this reduction.

\subsection{Construction of the reduced evolution system}

In the following it will be convenient to group the components of the
frame in the vector valued function
\begin{align*}
\bm{e}^t\equiv (C^{\mathcal{A}},\ P^{\mathcal{A}},\ Q),
\end{align*}
the spin connection coefficients not fixed by the gauge in
\begin{align*}
\bm{\Gamma}^t\equiv
(\epsilon,\ \pi,\ \beta,\ \mu,\ \alpha,\ \lambda,\ \tau,\ \sigma,\ \rho),
\end{align*}
and the independent components of the Weyl spinor as 
\begin{align*}
\bm{\Psi}^t\equiv (\Psi_0,\ \Psi_1,\ \Psi_2, \Psi_3,\ \Psi_4),
\end{align*}
where superscript-$t$ denotes the operation of taking the transpose of
a column vector.

A suitable symmetric hyperbolic system for the the frame components
and the spin coefficients can be obtained from
equations~\eqref{framecoefficient1}, \eqref{framecoefficient2},
\eqref{framecoefficient4} and~\eqref{structureeq1},
\eqref{structureeq2}, \eqref{structureeq3}, \eqref{structureeq4},
\eqref{structureeq6}, \eqref{structureeq7}, \eqref{structureeq11},
\eqref{structureeq13}, \eqref{structureeq15}, respectively. These can
be written in the schematic form
\begin{align*}
& \bm{\mathcal{D}}_1\bm{e}=\bm{B}_1(\bm{\Gamma},\bm{e})\bm{e}, \\
& \bm{\mathcal{D}}_2\bm{\Gamma}=\bm{B}_2(\bm{\Gamma},\bm{\Psi})\bm{\Gamma},
\end{align*}
where~$\bm{\mathcal{D}}_1$ and~$\bm{\mathcal{D}}_2$ are matrix
operators given by,
\begin{align*}
& \bm{\mathcal{D}}_1=\mbox{diag}(\Delta,\ \Delta,\ D), \\
& \bm{\mathcal{D}}_2=\mbox{diag} (\Delta,\ \Delta,\ \Delta,
\ \Delta,\ \Delta,\ \Delta, D,\ D,\ D),
\end{align*}
and~$\bm{B}_1$, $\bm{B}_2$ are smooth matrix-valued functions of their
arguments whose explicit form will not be required in the subsequent
analysis in this section.

The Bianchi identity equations~\eqref{Bianchi1}-\eqref{Bianchi8} can
be reorganised as
\begin{align}
\bm{\mathcal{D}}_3\bm{\Psi} =\bm{B}_3\bm{\Psi}
\label{BianchiEvolutionScheme}
\end{align}
where
\begin{align*}
\mathcal{D}_3=\left(\begin{array}{ccccc}
\Delta &-\delta &0&0&0 \\
-\bar\delta &D+\Delta&-\delta &0&0 \\
0&-\bar{\delta} &D+\Delta &-\delta&0 \\
0&0&-\bar\delta&D+\Delta&-\delta  \\
0&0&0 &-\bar{\delta}&D \\
\end{array}\right)
\end{align*}
and~$\bm{B_3}=\bm{B_3}(\bm{\Gamma})$. Writing
\begin{align*}
\bm{\mathcal{D}}_3 = \bm{A}^\mu_3 \p_\mu
\end{align*}
one has that 
\begin{align*}
&\bm{A}_3^v=\mbox{diag}(0,\ 1,\ 1,\ 1,\ 1),\\
&\bm{A}_3^u=\mbox{diag}(Q,\ Q,\ Q,\ Q,\ 0),
\end{align*}
and 
\begin{align*}
\bm{A}_3^{\mathcal{A}}=\left(\begin{array}{ccccc}
0 &-P^{\mathcal{A}} &0&0&0 \\
-\bar P^{\mathcal{A}} &C^{\mathcal{A}} &-P^{\mathcal{A}}&0&0 \\
0&-\bar{P}^{\mathcal{A}} &C^{\mathcal{A}}&-P^{\mathcal{A}}&0 \\
0&0&-\bar P^{\mathcal{A}} &C^{\mathcal{A}}&-P^{\mathcal{A}}  \\
0&0&0&-\bar{P}^{\mathcal{A}}&C^{\mathcal{A}}  \\
\end{array}\right). 
\end{align*}
The evolution system~\eqref{BianchiEvolutionScheme} for the components
of the Weyl tensor are obtained through the
combinations~\eqref{Bianchi2}, \eqref{Bianchi8}-\eqref{Bianchi1},
\eqref{Bianchi5}$+$\eqref{Bianchi6},
\eqref{Bianchi4}$+$\eqref{Bianchi7} and~$-$\eqref{Bianchi3}
respectively. It can be readily verified that the
matrices~$\bm{A}_3^{\mu}$ are Hermitian. Moreover,
\begin{align*}
\bm{A}_3^{\mu}(l_{\mu}+n_{\mu})=\mbox{diag}(1,\ 2,\ 2,\ 2,\ 1)
\end{align*}
is clearly positive definite. We can summarise the above discussion
with:

\begin{lemma}[\textbf{\em the evolution system}]
\label{Lemma:SHS}
The evolution system
\begin{subequations}
\begin{align}
& \bm{\mathcal{D}}_1\bm{e}=\bm{B}_1 \bm{e}, \label{reducedeq1}\\
& \bm{\mathcal{D}}_2\bm{\Gamma}=\bm{B}_2\bm{\Gamma},\label{reducedeq2} \\
& \bm{\mathcal{D}}_3\bm{\Psi} =\bm{B}_3\bm{\Psi}, \label{reducedeq3}
\end{align}
\end{subequations}
implied by the NP field equations written in Stewart's gauge (see
Section~\ref{Section:StewartGauge}) is symmetric hyperbolic with
respect to the direction given by~$\tau^a = l^a + n^a$.
\end{lemma}

\begin{remark}
{\em In the following, making use of the standard terminology, we call
  the evolution system the reduced Einstein field equations. }
\end{remark}

\begin{remark}
The symmetric hyperbolicity of the reduced
equations~\eqref{reducedeq1}-\eqref{reducedeq3} is the key structural
property which allows us to employ Rendall's local existence strategy
---see the discussion in Section~\ref{Subsection:FormalDerivatives}
below.
\end{remark}

As the hyperbolic reduction leading to the previous result makes use
of a subset of the NP equations, it is also key to have a
\emph{propagation of the constraints} result for the discarded
equations. Making use of analysis similar to the one discussed in
Section 12.5 of~\cite{CFEBook} one obtains the following:

\begin{proposition}[\textbf{\em propagation of the constraints}]
\label{Proposition:PropagationConstraints}
A solution of the reduced vacuum Einstein field
equations~\eqref{reducedeq1}-\eqref{reducedeq3} on a
neighbourhood~$\mathcal{V}$ of~$\mathcal{S}_{\star}$
on~$J^+(\mathcal{S}_{\star})$, the causal future
of~$\mathcal{S}_{\star}$, that coincides with initial data
on~$\mathcal{N}_\star'\cup\mathcal{N}_\star$ satisfying the vacuum
Einstein equations is a solution to the vacuum Einstein field
equations on~$\mathcal{V}$.
\end{proposition}

\begin{remark}
{\em A consequence of the propagation of the constraints, once local
  existence has been established, is that we may use any combination
  of the NP field equations in their gauge simplified form in the
  required subsequent analysis. For example, from this point on we
  have~$\pi=\alpha+\bar{\beta}$, and hence discard~$\pi$ or view it as
  a shorthand in what follows.}
\end{remark}

\subsection{Computation of the formal derivatives
  on~$\mathcal{N}_\star'\cup\mathcal{N}_\star$}
\label{Subsection:FormalDerivatives}

As already mentioned, Rendall's approach to the local existence of
solutions to the characteristic problem for symmetric hyperbolic
systems makes use of an auxiliary Cauchy problem on a spacelike
hypersurface
\begin{align*}
\mathcal{S}_\star \equiv \{p\in\mathbb{R}\times
\mathbb{R}\times\mathbb{S}^2\mid
v(p)+u(p)=0\}.
\end{align*} 
The formulation of this problem crucially depends on Whitney's
extension theorem. To apply this extension theorem it is necessary to
be able to evaluate all derivatives (interior and transverse) of the
initial data on~$\mathcal{N}_\star'\cup\mathcal{N}_\star$. A
discussion of the ideas behind Rendall's approach can be found in
Section 12.5 of~\cite{CFEBook}. For completeness, a formulation of
Rendall's result is given below:

\begin{theorem}[\textbf{\em local existence for the CIVP, Rendall}]
  \label{Thm:Rendall}
  Let~$\mathcal{N}_\star$ and $\mathcal{N}_\star'$ denote two
  characteristic hypersurfaces for the symmetric hyperbolic system
\begin{align*}
\mathbf{A}^\mu(x,\mathbf{u}) \p_\mu \mathbf{u}=\mathbf{B}(x,\mathbf{u})
\end{align*}
with smooth, freely specifiable data on~$\mathcal{N}_\star$
and~$\mathcal{N}_\star′$ such that all (formal) derivatives
of~$\mathbf{u}$ on~$\mathcal{N}_\star \cup\mathcal{N}_\star'$ to any
desired order can be computed in a neighbourhood~$\mathcal{W} \subset
\mathcal{N}_\star \cup\mathcal{N}_\star'$ of~$\mathcal{N}_\star
\cap\mathcal{N}_\star'$. Then there exists a unique
solution~$\mathbf{u}$ to the CIVP in a neighbourhood~$\mathcal{V}$
of~$\mathcal{N}_\star \cap\mathcal{N}_\star'$ with~$u \geq 0$, $v \geq
0$.
\end{theorem}

An important property of the NP equations in Stewart's gauge is that
they allow the computation of the (formal) derivatives of all the
fields to any order from the reduced data~$\mathbf{r}_\star$ provided
in Lemma~\ref{Lemma:FreeDataCIVP}. This property is discussed in the next
paragraphs.

\smallskip
\noindent
\textbf{Computation of formal derivatives on $\mathcal{N}_\star$.} To
compute the formal derivatives on~$\mathcal{N}_\star$ one first
observes that the partial
derivatives~${\p_v,\ \p_2,\ \p_3}$ are interior
whereas~${\p_u}$ is transverse. In this case, direct inspection
shows that except for
\begin{align*}
\p_uQ, \qquad  \p_u\tau, \qquad  \p_u\Psi_4, 
\end{align*}
all~$\p_u$-derivatives of the unknowns in the
vectors~${\bm{e},\ \bm{\Gamma},\ \bm{\Psi}}$ can be computed using the
structure equations~\eqref{framecoefficient1},
\eqref{framecoefficient2}, the NP Ricci identities
\eqref{structureeq1}, \eqref{structureeq3}, \eqref{structureeq4},
\eqref{structureeq7}, \eqref{structureeq9}, \eqref{structureeq11},
\eqref{structureeq15}, \eqref{structureeq18}, and the Bianchi
identities ~\eqref{Bianchi2},~\eqref{Bianchi4},~\eqref{Bianchi6}
and~\eqref{Bianchi8}.

To obtain these exceptional cases one first applies~$Q\bm{\p}_u$ to
both sides of equations~\eqref{framecoefficient4},
\eqref{structureeq2} and~\eqref{Bianchi3} to obtain
\begin{align*}
&Q\p_v(\p_uQ)=-Q^2\p_u(\epsilon+\bar{\epsilon})-Q(\epsilon+\bar{\epsilon})
 \p_uQ, \\
&Q\p_v(\p_u\tau)=L(\p_u\tau), \\
&Q\p_v(\p_u\Psi_4)-Q\p_u\bar{P}^{\mathcal{A}}\p_{\mathcal{A}}\Psi_3
-Q\bar{P}^{\mathcal{A}}\p_u\p_{\mathcal{A}}\Psi_3=M(\p_u\Psi_4), 
\end{align*}
where~$L, M$ are smooth functions of~$\{\bm{e}, \bm{\Gamma},
\bm{\Psi}\}$ and their~$\bmn$-direction derivatives. One can regard
the above equations as first order linear ordinary differential
equations for~$\p_uQ, \p_u\tau, \p_u\Psi_4$ along the generators
of~$\mathcal{N}_\star$. Since we have all the initial values of the
components of~$\{\bm{e}, \bm{\Gamma}, \bm{\Psi}\}$
on~$\mathcal{N}_\star'\cup\mathcal{N}_\star$, we can compute the
initial value of~$\p_uQ,\ \p_u\tau,\ \p_u\Psi_4$
on~$\mathcal{S}_{\star}$. The general results for the existence
theorem of ordinary differential equations ensures that the above
equation system can be solved in a neighbourhood
of~$\mathcal{S}_\star$. \emph{In the following, we assume that the
  initial data provided is such that it yields a uniform existence
  domain for the solutions to the transport equations ---this is a
  major assumption on the initial data in this construction.}
Accordingly, all the first transverse derivatives
on~$\mathcal{N}_\star$ can be explicitly computed. The higher
order~$\p_u$-derivatives can be computed in a similar way. Throughout
it is assumed that the neighbourhood on which this construction can be
done in uniform for any order of the derivatives.

\smallskip
\noindent
\textbf{Computation of formal derivatives on~$\mathcal{N}_\star'$.}
The analysis of the formal derivatives on~$\mathcal{N}_\star'$ is
almost the mirror image of that on~$\mathcal{N}_\star$.  In this
case~${\p_u, \p_2, \p_3}$ are interior while
${\p_v}$ is transverse. Accordingly, except for
\begin{align*}
\p_vC^A,\qquad \p_v\epsilon, \qquad \p_v\Psi_0,
\end{align*}
all~$\p_v$-derivatives of the components of~$\{\bm{e}, \bm{\Gamma},
\bm{\Psi}\}$ can be computed using the structure
equations~(\ref{framecoefficient3})-(\ref{framecoefficient4}), the
Ricci identities and the Bianchi identity. Applying the directional
derivative~$D=\p_v+C^A\p_A$ to both sides of
equations~\eqref{framecoefficient1}, \eqref{structureeq1} and
\eqref{Bianchi2} one obtains equations which can be regarded as first
order linear ordinary differential equations for~$\p_vC^A,
\p_v\epsilon, \p_v\Psi_0$. The solutions to these equations can be
obtained from the initial values prescribed
on~$\mathcal{S}_{\star}$. Thus, all transverse derivatives can be
computed in a neighbourhood of~$\mathcal{S}_\star$
on~$\mathcal{N}_\star'$. A similar procedure applies to higher order
derivatives.

The analysis described in the previous paragraph proves the following
lemma:
\begin{lemma}[\textbf{\em computation of  formal derivatives}] 
\label{Lemma:FormalDrivatives}
Any arbitrary formal derivatives of the unknown functions~$\{\bm{e},
\bm{\Gamma}, \bm{\Psi}\}$ on~$\mathcal{N}_\star'\cup\mathcal{N}_\star$
can be computed from the prescribed initial data~$\bm{r}_\star$ for
the reduced vacuum Einstein field equations
on~$\mathcal{N}_\star'\cap\mathcal{N}_\star$.
\end{lemma}

Combining the analysis above and applying Rendall's reduction strategy
for the CIVP for symmetric hyperbolic systems (see e.g. Section 12.5
of~\cite{CFEBook}) one obtains the following local existence result in
a neighbourhood of~$\mathcal{S}_\star =
\mathcal{N}_\star'\cup\mathcal{N}_\star$:

\begin{theorem}[\textbf{\em existence and uniqueness to the characteristic
      problem}]
\label{Theorem:RendallLocalExistence}
  Given a smooth reduced initial data set~$\bm{r}_\star$ for the
  vacuum Einstein field equations
  on~$\mathcal{N}_\star'\cup\mathcal{N}_\star$, there exists a unique
  smooth solution of the vacuum Einstein field equations in a
  neighbourhood~$\mathcal{V}$ of~$\mathcal{S}_{\star}$
  on~$J^+(\mathcal{S}_{\star})$ which implies the prescribed initial
  data on~$\mathcal{N}_\star'\cup\mathcal{N}_\star$.
\end{theorem}

\begin{remark}
{\em The proof of the above result has two distinct parts. In a first
  stage one uses Rendall's reduction procedure to show the existence
  of a solution in a neighbourhood of $\mathcal{V}$. In a second stage
  one shows that this solution to the reduced equations implies, in
  fact, a solution to the full Einstein field equations. This part of
  the argument relies on the propagation of the constraints as given
  in Proposition~\ref{Proposition:PropagationConstraints}. }
\end{remark}

\section{Setting-up Luk's strategy}

In this section we begin the implementation of Luk's strategy to
obtain an improved existence interval for the solutions to the CIVP
for the NP field equations in Stewart's gauge.

\subsection{Outline and main strategy}

As the argument leading to the improved existence result for the CIVP
is lengthy, we provide here a summary of the role of the various
lemmas and propositions and a discussion of how they fit into the
overall analysis. The whole scheme is based on the use of sequentially
more sophisticated {\it a priori} estimates of an arbitrary solution
that, ultimately, arrives at a contradiction giving us the desired
result.

\smallskip
\noindent
\textbf{Step 0. Estimates for the components of the frame.} The basic
step in the construction is to obtain estimates on the components of
the frame. This can be done by assuming control on the~$L^\infty$-norm
on the spheres~$\mathcal{S}_{u,v}$ of a number of spin connection
coefficients by a constant~$\Delta_\Gamma$. A peculiarity of the
analysis is that one needs to introduce a certain derivative (to be
denoted by~$\chi$) of the components of the frame as an unknown to
quick-start the argument ---this quantity, which is at the level of
the spin connection coefficients, does not arise in the original NP
formalism. The key result in this step is Lemma~\ref{Lemma:QPControl}
in which the frame coefficients~$Q$ and~$P^{\mathcal{A}}$ are
controlled by their initial data and Lemma~\ref{Lemma:CControl} in
which the frame coefficients~$C^{\mathcal{A}}$ are controlled along
the \emph{short direction}.

The bounds on the components of the frame allow us to control in a
systematic and streamlined manner the solutions to transport equations
along null directions in terms of integral quantities over the
spheres~$\mathcal{S}_{u,v}$. The technical results required to this
end are presented in Lemmas~\ref{Lemma:DerivativeIntegrals}
and~\ref{Lemma:IntegralIdentities}. From these, more specific results
valid for~$L^p$ and~$L^\infty$ norms are given in
Propositions~\ref{Proposition:TransportLpEstimates},
\ref{Proposition:SupremumNormTransportEquations},
\ref{Proposition:Sobolev}
and~\ref{Proposition:EstimateInfinityNorm}. Within our geometric setup
and gauge these results are fairly general and are used repeatedly in
the subsequent steps of the procedure.

\smallskip
\noindent
\textbf{Step 1. Estimates for the connection coefficients.} With the
general technology to study transport equations along the generators
of light cones has been established, one can proceed to control the
spin connection coefficients. The key idea of this analysis is the
integration of the transport equations implied by the Ricci
identities. In a first step, in
Proposition~\ref{Proposition:FirstEstimateConnection}, assuming
control on the supremum norm of the third angular derivatives of the
NP connection coefficient~$\tau$ and on the components of the
curvature one obtains control on the supremum norm of the various
connection coefficients and $\tau$ itself. This result is used in turn
in Proposition~\ref{PropositionSecondEstimateConnection} to obtain
control on the~$L^4$-norms of the connection coefficients and
the~$L^2$-norm of their derivatives in
Proposition~\ref{Proposition:ThirdEstimateConnection}.

\smallskip
\noindent
\textbf{Step 2. First estimate for the curvature.} A first estimate
for the components of the Weyl tensor is given in Proposition
\ref{Proposition:FirstEstimateCurvature}. In this result one assumes
control of the components of the Weyl tensor along the light cones and
of the $L^2$-norm of the third angular derivatives of the connection
coefficient $\tau$ on the spheres to obtain control of the components
of the Weyl tensor on the spheres.

The results of the steps~$1$ and~$2$ are conveniently summarised in
Proposition~\ref{Proposition:SummaryBasicEstmatesConnectionCurvature}
in which an assumed control on the components of the curvature along
light cones and of the~$L^2$-norm of the third angular derivatives
of~$\tau$ is used to obtain control on the spheres~$\mathcal{S}_{u,v}$
of various norms of the connection and its derivatives and of the
components of the curvature.

\smallskip
\noindent
\textbf{Step 3. Improved estimate for the connection.} In the next
step one obtains an improved estimate for the connection in which the
third angular derivatives of the connection, including~$\tau$, are
controlled assuming control only on the curvature along the light
cones. This result is given in
Proposition~\ref{Proposition:ImprovedEstimates}.

\smallskip
\noindent
\textbf{Step 4. Main estimates for the curvature.} At this point we
are in a position to run the central part of the argument, which
depends crucially on the particular structure of the Bianchi
identities. General inequalities for integrals of the various
components of the Weyl tensor implied by the Bianchi identities are
given in Propositions \ref{Proposition:FirstMainEstimateCurvature},
\ref{Proposition:SecondMainEstimateCurvature} and
\ref{Proposition:EstimatesDerivativesWeyl0123} and
\ref{Proposition:EstimatesDerivativesWeyl34}. The whole argument is
wrapped up in Proposition~\ref{Proposition:FinalEstimateCurvature} in
which, under the boundedness of the connection and the curvature on
the initial null hypersurfaces one obtains control of the curvature on
later null hypersurfaces. This is the crucial estimate which allows us
to close the lengthy boostrap argument.

\smallskip
\noindent
\textbf{Final step. Last slice argument.} The control of various
norms of the connection and curvature obtained in the previous steps
do not provide, by themselves, the improved existence result. For
this, we make use of a \emph{last slice argument} in which one argues
by contradiction under the assumption that the solution to the
evolution equations breaks down at some point. The estimates of the
previous steps show that this assumption leads to a contradiction.

\subsection{Definitions and conventions}

In this section we set up the conventions for the various norms that
will be used in the subsequent analysis.

\smallskip
\noindent
\textbf{Integration.} In the following let~$\phi$ denote a scalar
field. For conciseness, we will often use the notation
\begin{align*}
\int_{\mathcal{S}_{u,v}}\phi\equiv \int_{\mathcal{S}_{u,v}}\phi
\mathrm{d}{\bm \sigma} 
\end{align*}
to denote integration on the spheres~$\mathcal{S}_{u,v}$ of
constant~$u$ and~$v$. In the previous expression~$\mathrm{d}
{\bm\sigma}\equiv\sqrt{|\det{\bm \sigma}|} \mathrm{d}x^2\mathrm{d}x^3$
denotes the volume element of the induced metric $\bm\sigma$ on
$\mathcal{S}_{u,v}$. On the truncated causal
diamonds~$\mathcal{D}_{u,v}^{\,t}$ we define integration using the
volume form of the spacetime metric,
\begin{align*}
\int_{\mathcal{D}_{u,v}^{\,t}}\phi&\equiv
\int_0^u\int_0^{\tilde{v}}\int_{\mathcal{S}_{u',v'}}\phi\sqrt{|\det {\bm g}|}
\mathrm{d}x^2\mathrm{d}x^3\mathrm{d}v'\mathrm{d}u'\\
&=\int_0^u\int_0^{\tilde{v}}\int_{\mathcal{S}_{u',v'}}
Q^{-1}\phi\sqrt{|\det{\bm\sigma}|} \mathrm{d}x^2
\mathrm{d}x^3\mathrm{d}v'\mathrm{d}u',
\end{align*}
with~$\tilde{v}\equiv\min(v,t-u)$. We will denote integration over the
complete causal diamond in the obvious manner by the natural omission
of the superscript~$t$ on~$\mathcal{D}_{u,v}^{\,t}$. As there are no
canonical volume forms on the null hypersurfaces~$\mathcal{N}_u$
and~$\mathcal{N}'_v$ we define, for convenience the following:
\begin{align*}
&\int_{\mathcal{N}_u(0,v)}\phi \equiv \int_0^v\int_{S_{u,v'}}
\phi\sqrt{|\det{\bm\sigma}|} \mathrm{d}x^2\mathrm{d}x^3\mathrm{d}v',\\
&\int_{\mathcal{N}'_v(0,u)}\phi \equiv
\int_0^u\int_{S_{u',v}}\phi\sqrt{|\det{\bm \sigma}|}
\mathrm{d}x^2\mathrm{d}x^3\mathrm{d}u'.
\end{align*}
We will often use the notation
\begin{align*}
\int_{\mathcal{N}_u^{\,t}}\phi \equiv \int_{\mathcal{N}_u(I^{\,t})}\phi , \qquad 
\int_{\mathcal{N}^{'t}_v}\phi \equiv \int_{\mathcal{N}'_v[0,\varepsilon]^t}\phi 
\end{align*}
where~$I^{\,t}\equiv[0,\min(v_\bullet,t-u)]$, with $v_\bullet\in
\mathbb{R}^+$, denotes the {\it truncated long integration interval}.
Similarly, the interval~$[0,\varepsilon]^t\equiv
[0,\min(\varepsilon,t-v)]$ will be called the \emph{truncated short
  integration interval}. Dropping the superscript~$t$ we define the
full long and short integration intervals,~$I$ and~$[0,\varepsilon]$
respectively, and the norms on the full outgoing and incoming slices
in the natural way.

\smallskip
\noindent
\textbf{Norms.} Keeping the above conventions for integration in mind,
we can now define the various norms to be used in our analysis. As
before, let~$\phi$ define a scalar field. For~$1\leq p<\infty$ we
define the~$L^p$-norms
\begin{align*}
  ||\phi||_{L^p(\mathcal{S}_{u,v})}\equiv
  \left(\int_{\mathcal{S}_{u,v}}|\phi|^{p}\right)^{1/p},\qquad
  ||\phi||_{L^p(\mathcal{N}_u^{\,t})}\equiv
  \displaystyle{\left(\int_{\mathcal{N}_u^{\,t}}|\phi|^{p}\right)^{1/p}},\qquad
  ||\phi||_{L^p(\mathcal{N}^{'t}_v)}\equiv
  \left(\int_{\mathcal{N}^{'t}_v}|\phi|^{p}\right)^{1/p}.
\end{align*}
The~$L^{\infty}$-norm is defined by
\begin{align*}
||\phi||_{L^{\infty}(\mathcal{S}_{u,v})}\equiv\sup_{\mathcal{S}_{u,v}}|\phi|.
\end{align*}
For a tensor field~$\phi_{a_1\ldots a_p}$ on the~$2$-sphere, we define
\begin{align*}
  ||\phi||_{L^p(\mathcal{S}_{u,v})}\equiv
  \left(\int_{\mathcal{S}_{u,v}}\langle\phi,\phi\rangle^{p/2}_{\bmsigma}\right)^{1/p},\,
  ||\phi||_{L^p(\mathcal{N}_u^{\,t})}\equiv
  \left(\int_{\mathcal{N}_u^{\,t}}\langle\phi,\phi\rangle^{p/2}_{\bmsigma}\right)^{1/p},\,
  ||\phi||_{L^p(\mathcal{N}^{'t}_v)}\equiv
  \left(\int_{\mathcal{N}^{'t}_v}\langle\phi,\phi\rangle^{p/2}_{\bmsigma}\right)^{1/p},
\end{align*}
where~$\langle\phi,\phi\rangle_{\bmsigma}\equiv\sigma^{a_1b_1}...
\sigma^{a_pb_p}\bar\phi_{a_1,...,a_p}\phi_{b_1,...,b_p}$. As in the
definition of the integrals, suppresion of the label~$t$ denotes
taking the norms over the full long and short integration intervals.

\smallskip
\noindent
\textbf{Integration by parts.} In the following we denote
by~$\nablasl$ the covariant derivative of the induced metric~${\bm
  \sigma}$ on the spheres $\mathcal{S}_{u,v}$ of constant~$u$
and~$v$. Similarly,~$\Deltasl$ will denote the associated
Laplacian. As these spheres have no boundary we have 
\begin{align*}
    ||\nablasl
    \phi||^2_{L^2(\mathcal{S}_{u,v})}&=\int_{\mathcal{S}_{u,v}}\sigma^{ab}\nablasl_a
    \phi\nablasl_b\bar{\phi}=\int_{\mathcal{S}_{u,v}}\nablasl_a
    (\sigma^{ab}\phi\nablasl_b\bar{\phi})-\int_{\mathcal{S}_{u,v}}\phi
    \Deltasl\bar{\phi},\\
    &=-\int_{\mathcal{S}_{u,v}}\phi\Deltasl
    \bar{\phi}\leq 2\left(\int_{\mathcal{S}_{u,v}}|\phi|^2\right)^{1/2}
    \left(\int_{\mathcal{S}_{u,v}}|\nablasl^2\phi|^2\right)^{1/2}
\end{align*}
where in the last step inequality~\eqref{InequalityLaplacian} in
Appendix~\ref{nablaf} has been used. Integrating
over~$\langle\phi,\pi\rangle_{\bmsigma}$ over two-spheres naturally
defines an inner product, so we similarly obtain,
\begin{align*}
||\nablasl \phi||_{L^2(\mathcal{S}_{u,v})}&\leq
||\phi||_{L^2(\mathcal{S}_{u,v})}+||\nablasl^2\phi||_{L^2(\mathcal{S}_{u,v})},\\
||\nablasl^2\phi||_{L^2(\mathcal{S}_{u,v})}&\leq ||\nablasl
   \phi||_{L^2(\mathcal{S}_{u,v})}+||\nablasl^3\phi||_{L^2(\mathcal{S}_{u,v})}.
\end{align*}

\subsection{Estimates for the components of the frame}

As a preliminary step we now show that, assuming the components of the
connection coefficients are controlled by a basic boostrap assumption,
it is possible to estimate the components of the NP frame in terms of
the size of its initial data on~$\mathcal{N}_\star \cup
\mathcal{N}'_\star$. The key observation in the argument is that the
structure equations provide~$\Delta$-equations for all the components
of the frame. Given our particular choice of gauge, these equations
are essentially ordinary differential equations with respect to the
coordinate~$u$. In fact as the structure equations form a neat
hierarchy, they can be integrated sequentially. The quantity,
\begin{align}
\Delta_{e_\star}\equiv\sup_{\mathcal{N}_\star,\mathcal{N}'_\star}
\left(|Q|,|Q^{-1}|,|C^{\mathcal{A}}|,|P^{\mathcal{A}}|\right)
\label{Definition:DeltaEstar}
\end{align}
will be used to measure of the size of the initial data of the
coefficients of the frame. \emph{Throughout, given that the procedure
  has only a finite number of steps we denote all constants depending
  on the initial data generically by~$C(\Delta_{e_\star})$ ---the
  latter corresponds to the largest constant arising in the various
  steps.} For convenience in the subsequent discussion let
\begin{align*}
\chi\equiv \Delta\log Q.
\end{align*}
The scalar~$\chi$, being a derivative of a component of the frame is
at the same level of the connection coefficients. It provides a
component of the connection which does not arise in the original NP
formalism, but is needed to obtain a complete set of~$\Delta$
equations for the frame. A direct computation using the definition
of~$\chi=\Delta\log Q$ and the NP Ricci identities yields
\begin{align}
\label{EqDchi}
D\chi=\Psi_2+\bar\Psi_2+2\alpha\tau+2\bar\beta\tau+2\bar\alpha\bar\tau
+2\beta\bar\tau+2\tau\bar\tau-(\epsilon+\bar\epsilon)\chi.
\end{align}
The initial data of~$\chi$ on~$\mathcal{N}'_{\star}$ is~$0$ due to the
gauge choice that~$Q=1$
on~$\mathcal{N}'_{\star}$. On~$\mathcal{N}_{\star}$, making use of the
information of~$\alpha$, $\beta$, $\tau$, $\epsilon$ and~$\Psi_2$
obtained in Lemma~\ref{Lemma:FreeDataCIVP}, one can compute the value
of~$\chi$ with equation \eqref{EqDchi}. It will also be convenient to
define,
\begin{align*}
\varpi \equiv \beta-\bar\alpha
\end{align*}
corresponding to the only independent component of the connection on
the spheres~$\mathcal{S}_{u,v}$. As mentioned above, the proof is
based on demonstrating a priori estimates for an arbitrary solution
and consequently demonstrating that any such solution must extend to a
neighborhood of~$\mathcal{N}_\star\cup\mathcal{N}_\star'$. We
therefore now introduce the following, which will be initially
guaranteed on a sufficiently small diamond by
Theorem~\ref{Thm:Rendall}, and will be employed in most of what
follows:

\begin{assumption}[\textbf{\em assumption to control the coefficients
      of the frame}]\label{Assumption:Frame} Assume that we have a
  solution to the vacuum EFEs in Stewart's gauge satisfying,
\begin{align*}
  ||\{\mu, \lambda, \alpha, \beta, \tau, \chi\}||_{L^{\infty}(\mathcal{S}_{u,v})}
  \leq\Delta_{\Gamma}, 
\end{align*}
on a truncated causal diamond~$\mathcal{D}_{u,v_\bullet}^{\,t}$,
where~$\Delta_{\Gamma}$ is some constant.
\end{assumption}

\smallskip
\noindent
\textbf{Step 1.} Work under Assumption~\ref{Assumption:Frame}.
Integrating the definition of~$\chi=\Delta\log Q$ in the short
(i.e.~$u$) direction along an incoming null geodesic one readily finds
that,
\begin{align*}
|Q-Q_\star|=\left|\,\int_0^{\varepsilon}\chi
\mathrm{d}u\,\right|\leq\int_0^{\varepsilon}|\chi|\mathrm{d}u
\leq\int_0^{\varepsilon}
\Delta_{\Gamma}\mathrm{d}u=\Delta_{\Gamma}\varepsilon
\end{align*}
for any~$v$. It follows that
\begin{align*}
||Q-Q_\star||_{L^{\infty}(\mathcal{S}_{u,v})}\leq\Delta_{\Gamma}\varepsilon .
\end{align*}
Hence, one can find a constant~$C$ depending on the initial data such
that
\begin{align*}
Q^{-1}, Q\leq C(\Delta_{e_\star}).
\end{align*}

\smallskip
\noindent
\textbf{Step 2.} We now integrate the components~$P^{\mathcal{A}}$ in
the short direction using equation~\eqref{framecoefficient2}. It
follows then that
\begin{align*}
  &\p_u|P^{\mathcal{A}}|^2=\p_u(P^{\mathcal{A}}\bar P^{\mathcal{A}})
  =P^{\mathcal{A}}\p_u\bar P^{\mathcal{A}}+\bar P^{\mathcal{A}}\p_uP^{\mathcal{A}} \\
  &\phantom{\p_u|P^{\mathcal{A}}|^2}=-Q^{-1}\left(P^{\mathcal{A}}
  (\bar\mu\bar P^{\mathcal{A}}+\lambda P^{\mathcal{A}})+\bar P^{\mathcal{A}}
  (\mu P^{\mathcal{A}}+\bar\lambda\bar P^{\mathcal{A}})\right) \\
  &\phantom{\p_u|P^{\mathcal{A}}|^2}=-Q^{-1}
  \left(\bar\mu |P^{\mathcal{A}}|^2+\lambda(P^{\mathcal{A}})^2
  +\mu |P^{\mathcal{A}}|^2+\bar\lambda(\bar P^{\mathcal{A}})^2 \right) \\
  &\phantom{\p_u|P^{\mathcal{A}}|^2}\leq Q^{-1}
  (\mu+\bar\mu+\lambda+\bar\lambda)|P^{\mathcal{A}}|^2.
\end{align*}
In the previous chain of inequalities it is understood that there is
no summation on the repeated indices~$\mathcal{A}$. From the last
inequality one readily concludes that
\begin{align*}
\p_u\ln|P^{\mathcal{A}}|^2\leq 4Q^{-1}\Delta_{\Gamma} 
\end{align*}
so that
\begin{align*}
  |P^{\mathcal{A}}|^2\leq|P^{\mathcal{A}}_\star|^2\exp(4C(\Delta_{e_\star})
  \Delta_{\Gamma}\varepsilon).
\end{align*}
As~$\varepsilon$ is arbitrary, we can choose it so that
\begin{align*}
  |P^{\mathcal{A}}|\leq C(\Delta_{e_\star}), \ \mbox{for any}\ u\
  \mbox{and fixed} \ v.
\end{align*}
The analysis of~\emph{Steps 1} and~\emph{2} can be summarised in the
following:

\begin{lemma}[\textbf{\em control on the components of the frame. I}]
\label{Lemma:QPControl} 
Under Assumption~\ref{Assumption:Frame}, if~$\varepsilon>0$ is
sufficiently small, there exists a constant~$C$ depending on the size
of the initial data such that
\begin{align*}
  Q^{-1}, Q\leq C(\Delta_{e_\star}), \qquad
  ||P^{\mathcal{A}}||_{L^{\infty}(\mathcal{S}_{u,v})}\leq C(\Delta_{e_\star}),
\end{align*}
on~$\mathcal{D}_{u,v_\bullet}^{\,t}$.
\end{lemma}

A direct consequence of this result is that one can control the
components of the induced metric on the spheres~$S_{u,v}$ and
associated concomitants. This follows from the relation
\begin{align*}
\sigma^{\mathcal{AB}}=-P^{\mathcal{A}}\bar
P^{\mathcal{B}}-P^{\mathcal{B}}\bar P^{\mathcal{A}}.
\end{align*}

\begin{corollary}[\textbf{\em control on the metric of~$\mathcal{S}_{u,v}$}]
\label{Corrollary:Area}
If~$\varepsilon>0$ is sufficiently small there exist non-negative
constants~$c(\Delta_{e_\star})$ and~$C(\Delta_{e_\star})$ such that,
\begin{align*}
  |\sigma^{\mathcal{AB}}|,\,|\sigma_{\mathcal{AB}}|\leq C(\Delta_{e_\star}),
  \qquad c(\Delta_{e_\star})\leq|\det\bmsigma|\leq C(\Delta_{e_\star}).
\end{align*}  
Moreover, one also has that
\begin{align*}
\sup_{u,v}|\mbox{\em Area}(\mathcal{S}_{u,v})
-\mbox{\em Area}(\mathcal{S}_{0,v})|\leq
C(\Delta_{e_\star})\Delta_{\Gamma}\varepsilon,
\end{align*}
on~$\mathcal{D}_{u,v_\bullet}^{\,t}$. Consequently the area
of~$\mathcal{S}_{u,v}$ is bounded above by a constant depending in
initial data in the same region, for~$\varepsilon$ sufficiently small.
\end{corollary}

\smallskip
\noindent
\textbf{Step 3.} One can now use equation~\eqref{framecoefficient1} to
integrate the coefficients~$C^{\mathcal{A}}$. By a procedure similar
to that used in the previous steps one has,
\begin{align*}
|C^{\mathcal{A}}-C^{\mathcal{A}}_\star|&=\left|\,\int_0^{\epsilon}
Q^{-1}\left((\bar\tau+\pi)P^{\mathcal{A}}+(\tau+\bar\pi)
\bar P^{\mathcal{A}} \right)\mathrm{d}u\,\right| \\
&\leq C(\Delta_{e_\star})
\int_0^{\epsilon}|(\bar\tau+\pi)
P^{\mathcal{A}}+(\tau+\bar\pi)\bar P^{\mathcal{A}}|\mathrm{d}u \\
&\leq2C(\Delta_{e_\star}) \int_0^{\epsilon}|\bar\tau+\pi||P^{\mathcal{A}}|
\mathrm{d}u\leq2C(\Delta_{e_\star})^2\Delta_{\Gamma}\varepsilon.
\end{align*}
Here~$\pi$ should be viewed as a shorthand for
~$\pi=\alpha+\bar{\beta}$. Since~$C^{\mathcal{A}}_\star=0$
on~$\mathcal{N}_\star$, we arrive at:

\begin{lemma}[\textbf{\em control on the components of the frame. II}]
\label{Lemma:CControl}
Under Assumption \ref{Assumption:Frame}, if~$\varepsilon>0$ is
sufficiently small, then there is a constant~$C(\Delta_{e_\star})$
depending only on the initial data such that choosing~$\varepsilon$
suitably, one
has~$||C^{\mathcal{A}}||_{L^{\infty}(\mathcal{S}_{u,v})}\leq
C(\Delta_{e_\star})$ on~$\mathcal{D}_{u,v_\bullet}^{\,t}$.
\end{lemma}

\subsection{General estimates for transport equations}
\label{Section:EstimatesTransportEquations}

The purpose of this section is to develop a general set of tools that
allow us to obtain estimates from the transport equations on
hypersurfaces of constant~$u$ or~$v$. The prototype of these transport
equations are the NP Ricci
identities~\eqref{structureeq1}-\eqref{structureeq18}. The results of
this section do not depend on Assumption~\ref{Assumption:Frame} unless
explicitly stated.

\smallskip
\noindent
\textbf{Derivatives of integrals over~$\mathcal{S}_{u,v}$.} We are
mostly interested on integral estimates over the
spheres~$\mathcal{S}_{u,v}$ and how they evolve along null
directions. In the following we will systematically need to compute
derivatives of integrals over~$\mathcal{S}_{u,v}$ with respect to the
advanced and retarded null coordinates. The key observation in this
respect is the following:

\begin{lemma}[\textbf{\em computing derivatives of integrals
      over~$\mathcal{S}_{u,v}$}]\label{Lemma:DerivativeIntegrals}
  Given a scalar~$\phi$ one has that
\begin{subequations}
\begin{align}
&\frac{\mathrm{d}}{\mathrm{d}v}\int_{\mathcal{S}_{u,v}}
\phi=\int_{\mathcal{S}_{u,v}}\left(D\phi-2\rho\phi\right),
\label{DerivativeTwoSphere1}\\
&\frac{\mathrm{d}}{\mathrm{d}u}\int_{\mathcal{S}_{u,v}}
\phi=\int_{\mathcal{S}_{u,v}}Q^{-1}\left(\Delta\phi+2\mu\phi\right),
\label{DerivativeTwoSphere2}
\end{align}
\end{subequations}
along the outgoing and incoming null geodesics that rule~$\mathcal{N}_{v}'$
and~$\mathcal{N}_{u}$.
\end{lemma}

\begin{proof}
The proof follows a direct computation. More precisely, one has that
\begin{align*}
  \frac{\mathrm{d}}{\mathrm{d}v}\int_{\mathcal{S}_{u,v}}\phi
  &=\int_{\mathcal{S}_{u,v}}\frac{\p}{\p v}(\phi\sqrt{|\det\bmsigma|})
  \mathrm{d}x^2\mathrm{d}x^3\\
  &=\int_{\mathcal{S}_{u,v}}\left(D(\phi\sqrt{|\det\bmsigma|})
  -C^{\mathcal{A}}\p_{\mathcal{A}}(\phi\sqrt{\det\bmsigma}\right))
  \mathrm{d}x^2\mathrm{d}x^3\\
  &=\int_{\mathcal{S}_{u,v}}\left(D\phi\sqrt{|\det\bmsigma|}
  +\phi D\sqrt{|\det\bmsigma|}-C^{\mathcal{A}}
  \p_{\mathcal{A}}(\phi\sqrt{|\det\bmsigma|})\right)
  \mathrm{d}x^2\mathrm{d}x^3.
\end{align*}
For the second term in the integrand,~$\phi D\sqrt{|\det\bmsigma|}$,
we find that
\begin{align*}
  D\sqrt{|\det\bmsigma|}&=\frac{1}{2\sqrt{|\det\bmsigma|}}D\det\bmsigma
  =\frac{|\det\bmsigma|}{2\sqrt{|\det\bmsigma|}}\sigma^{\mathcal{A}\mathcal{B}}
  D\sigma_{\mathcal{A}\mathcal{B}}=-\frac{\sqrt{|\det\bmsigma|}}{2}
  \sigma_{\mathcal{A}\mathcal{B}}D\sigma^{\mathcal{A}\mathcal{B}}\\
  &=\sqrt{|\det\bmsigma|}
  \sigma_{\mathcal{A}\mathcal{B}}\left(\bar P^{\mathcal{B}}DP^\mathcal{A}
  +P^\mathcal{A}D\bar P^{\mathcal{B}}\right)\\
  &=\sqrt{|\det\bmsigma|}
  \left(\sigma_{\mathcal{A}\mathcal{B}}\bar P^{\mathcal{B}}\delta C^{\mathcal{A}}
  +\sigma_{\mathcal{A}\mathcal{B}}P^{\mathcal{A}}\bar\delta C^{\mathcal{B}}-2\rho
  +\sigma\sigma_{\mathcal{A}\mathcal{B}}\bar P^{\mathcal{A}}\bar P^{\mathcal{B}}
  +\bar\sigma\sigma_{\mathcal{A}\mathcal{B}}P^{\mathcal{A}}P^{\mathcal{B}}\right) \\
  &=\sqrt{|\det\bmsigma|}\left(\bar{m}_{\mathcal{A}}
  \delta C^{\mathcal{A}}+m_{\mathcal{A}}\bar\delta C^{\mathcal{A}}-2\rho\right)
  =-\sqrt{|\det\bmsigma|}\left(\p_{\mathcal{A}}C^{\mathcal{A}}+2\rho\right),
\end{align*}
where we have used Remark~\ref{Remark_1} and the structure
equation~\eqref{framecoefficient3}. For the third term in the integral
one has that
\begin{align*}
  \int_{\mathcal{S}_{u,v}}C^{\mathcal{A}}\p_{\mathcal{A}}
  \left(\phi\sqrt{|\det\bmsigma|}\right)
  \mathrm{d}x^2\mathrm{d}x^3&=\int_{\mathcal{S}_{u,v}}\p_{\mathcal{A}}
  \left(C^{\mathcal{A}}\phi\sqrt{|\det\bmsigma|}\right)
  \mathrm{d}x^2\mathrm{d}x^3
  -\int_{\mathcal{S}_{u,v}}\phi\p_{\mathcal{A}}C^{\mathcal{A}}
  \sqrt{|\det\bmsigma|}\mathrm{d}x^2\mathrm{d}x^3\\
  &=-\int_{\mathcal{S}_{u,v}}\phi\p_{\mathcal{A}}C^{\mathcal{A}}
  \sqrt{|\det\bmsigma|}\mathrm{d}x^2\mathrm{d}x^3
  +\int_{\mathcal{S}_{u,v}}\nabla_\mathcal{A}
  \left(C^{\mathcal{A}}\phi\sqrt{|\det\bmsigma|}\right)
  \mathrm{d}x^2\mathrm{d}x^3\\
  &=-\int_{\mathcal{S}_{u,v}}\phi\p_{\mathcal{A}}C^{\mathcal{A}}
  \sqrt{|\det\bmsigma|}\mathrm{d}x^2\mathrm{d}x^3,
 \end{align*}
where for the last equality we have use Stokes' theorem and the fact
that sphere has no boundary. Combining the above observations one
finds that
\begin{align*}
  \frac{\mathrm{d}}{\mathrm{d}v}\int_{\mathcal{S}_{u,v}}\phi
  =\int_{\mathcal{S}_{u,v}}\left(D\phi-2\rho\phi\right)\sqrt{|\det\bmsigma|}
  \mathrm{d}x^2\mathrm{d}x^3.
\end{align*}
To compute the derivative with respect to~$u$, we first consider
\begin{align*}
  \Delta\sqrt{|\det\bmsigma|}&=-\frac{1}{2}\sqrt{|\det\bmsigma|}
  \sigma_{\mathcal{A}\mathcal{B}}\Delta\sigma^{\mathcal{A}\mathcal{B}}
  =\frac{1}{2}\sqrt{|\det\bmsigma|}\sigma_{\mathcal{A}\mathcal{B}}
  \left(\bar P^{\mathcal{B}}\Delta P^\mathcal{A}+P^\mathcal{A}
  \Delta\bar P^{\mathcal{B}} \right) \\
  &=\frac{1}{2}\sqrt{|\det\bmsigma|}\sigma_{\mathcal{A}\mathcal{B}}
  \left(\bar P^{\mathcal{B}}(-\mu P^\mathcal{A}-\bar\lambda\bar
  P^\mathcal{A})+P^\mathcal{A}(-\bar\mu\bar P^\mathcal{B}
  -\lambda P^\mathcal{B} ) \right)\\
  &=\frac{1}{2}\left(\mu+\bar\mu\right)
  \sqrt{|\det\bmsigma|}=\mu\sqrt{|\det\bmsigma|}.
\end{align*}
From the above identity one readily obtains
\begin{align*}
  \frac{\mathrm{d}}{\mathrm{d}u}\int_{\mathcal{S}_{u,v}}\phi
  &=\int_{\mathcal{S}_{u,v}}\frac{\p}{\p u}
  \left(\phi\sqrt{|\det\bmsigma|}\right)
  \mathrm{d}x^2\mathrm{d}x^3\\
  &=\int_{\mathcal{S}_{u,v}}Q^{-1}\left(\sqrt{|\det\bmsigma|}
  \Delta\phi+\phi\Delta\sqrt{|\det\bmsigma|}\right)
  \mathrm{d}x^2\mathrm{d}x^3\\
  &=\int_{\mathcal{S}_{u,v}}Q^{-1}\left(\Delta\phi+2\mu\phi\right)
  \sqrt{|\det\bmsigma|}\mathrm{d}x^2\mathrm{d}x^3,
\end{align*}
as required.
\end{proof}

\smallskip
\noindent
\textbf{Integrals over~$\mathcal{D}_{u,v}$.} The construction of
energy-type estimates for the components of the Weyl tensor require
further integral identities. These integrals allow us to write the
integral over the diamond~$\mathcal{D}_{u,v}$ of the~$D$
and~$\Delta$-derivatives of the components of the Weyl tensor in terms
of integrals on the light cones and an integral over the bulk diamond
of the (undifferentiated) components.

\begin{lemma}[\textbf{\em integral over causal diamonds of derivatives of
      a scalar}]\label{Lemma:IntegralIdentities} Let~$f$ be a scalar
  field in the causal diamond~$\mathcal{D}_{u,v}$. One has then that
  \begin{align*}
  &\int_{\mathcal{D}_{u,v}}Df=\int_{\mathcal{N}'_v(0,u)}Q^{-1}f
  -\int_{\mathcal{N}'_0(0,u)}Q^{-1}f
 +\int_{\mathcal{D}_{u,v}}(2\rho+\epsilon+\bar\epsilon)f ,\\
 &\int_{\mathcal{D}_{u,v}}\Delta f=\int_{\mathcal{N}_u(0,v)}f
 -\int_{\mathcal{N}_0(0,v)}f-\int_{\mathcal{D}_{u,v}}2\mu f .
\end{align*}
\end{lemma}

\begin{proof}
The proof of the identities follows by integration by parts. For the
long direction we have, by definition, that
\begin{align*}
  \int_{\mathcal{D}_{u,v}}Df=\int_0^u\int_0^v\int_{\mathcal{S}_{u',v'}}Q^{-1}
  (\p_vf+C^{\mathcal{A}}\p_{\mathcal{A}}f)\sqrt{|\det\bmsigma|}\mathrm{d}x^2
  \mathrm{d}x^3\mathrm{d}u'\mathrm{d}v'.
\end{align*}
Now, on the one hand, integrating by parts with respect to $v$ one has
that,
\begin{align*}
&\int_0^u\int_0^v\int_{\mathcal{S}_{u',v'}}Q^{-1}
 \p_vf\sqrt{|\det\bmsigma|}\mathrm{d}x^2
 \mathrm{d}x^3\mathrm{d}v'\mathrm{d}u'\\
 &\quad\quad=\int_0^u\int_0^v\int_{\mathcal{S}_{u',v'}}
 \p_v(Q^{-1}f\sqrt{|\det\bmsigma|})
 \mathrm{d}x^2\mathrm{d}x^3\mathrm{d}v'\mathrm{d}u'\\
&\quad\quad\quad-\int_0^u\int_0^v\int_{\mathcal{S}_{u',v'}}f
 \p_v(Q^{-1}\sqrt{|\det\bmsigma|})
 \mathrm{d}x^2\mathrm{d}x^3\mathrm{d}v'\mathrm{d}u',\\
&\quad\quad=\int_{\mathcal{N}'_v(0,u)}Q^{-1}f-\int_{\mathcal{N}'_0(0,u)}Q^{-1}f\\
&\quad\quad\quad-\int_0^u\int_0^v\int_{\mathcal{S}_{u',v'}}
 (f\p_vQ^{-1}\sqrt{|\det\bmsigma|}+Q^{-1}f\p_v\sqrt{|\det\bmsigma|})
 \mathrm{d}x^2\mathrm{d}x^3\mathrm{d}v'\mathrm{d}u' .
\end{align*}
On the other hand, integration by parts respect to the angular
coordinates gives
\begin{align*}
  &\int_0^u\int_0^v\int_{\mathcal{S}_{u',v'}}Q^{-1}C^{\mathcal{A}}
  \p_{\mathcal{A}}f\sqrt{|\det\bmsigma|}
  \mathrm{d}x^2\mathrm{d}x^3\mathrm{d}v'\mathrm{d}u'\\
  &=-\int_0^u\int_0^v\int_{\mathcal{S}_{u',v'}}f
  \p_{\mathcal{A}}(Q^{-1}C^{\mathcal{A}}\sqrt{|\det\bmsigma|})
  \mathrm{d}x^2\mathrm{d}x^3\mathrm{d}v'\mathrm{d}u',\\
  &=-\int_0^u\int_0^v\int_{\mathcal{S}_{u',v'}}(f\sqrt{|\det\bmsigma|}
  C^{\mathcal{A}}\p_{\mathcal{A}}Q^{-1}+Q^{-1}f\sqrt{|\det\bmsigma|}
  \p_{\mathcal{A}}C^{\mathcal{A}}+Q^{-1}f C^{\mathcal{A}}\p_{\mathcal{A}}
  \sqrt{|\det\bmsigma|}) \mathrm{d}x^2\mathrm{d}x^3\mathrm{d}v'\mathrm{d}u'.
\end{align*}
Thus, we have
\begin{align*}
  &\int_{\mathcal{D}_{u,v}}Df=\int_{\mathcal{N}'_v(0,u)}Q^{-1}f
  -\int_{\mathcal{N}'_0(0,u)}Q^{-1}f\\
  &-\int_0^u\int_0^v\int_{S_{u',v'}}(\sqrt{|\det\bmsigma|}f
  Q^{-2}DQ+Q^{-1}fD\sqrt{|\det\bmsigma|}+Q^{-1}f
  \sqrt{|\det\bmsigma|}\p_{\mathcal{A}}C^{\mathcal{A}})
  \mathrm{d}x^2\mathrm{d}x^3\mathrm{d}v'\mathrm{d}u'.
\end{align*}
Finally, making use of the expressions for~$DQ$ from
equation~\eqref{framecoefficient4} and~$D\sqrt{|\det\bmsigma|}$ from
Proposition~\ref{Lemma:DerivativeIntegrals}, respectively, one obtains
the desired identity.

To demonstrate the identity along the short direction one proceeds in
a similar fashion.
\end{proof}

\begin{corollary}
  If~$f=f_1f_2$, then
\begin{align*}
  & \int_{\mathcal{D}_{u,v}}f_1Df_2+\int_{\mathcal{D}_{u,v}}f_2Df_1
  =\int_{\mathcal{N}'_v(0,u)}Q^{-1}f_1f_2-\int_{\mathcal{N}'_0(0,u)}Q^{-1}f_1f_2
  +\int_{\mathcal{D}_{u,v}}(2\rho+\epsilon+\bar\epsilon)f_1f_2 ,\\
  & \int_{\mathcal{D}_{u,v}}f_1\Delta f_2+\int_{\mathcal{D}_{u,v}}f_2
  \Delta f_1=\int_{\mathcal{N}_u(0,v)}f_1f_2-\int_{\mathcal{N}_0(0,v)}f_1f_2
  -\int_{\mathcal{D}_{u,v}}2\mu f_1f_2 .
\end{align*}
\end{corollary}

\smallskip
\noindent
\textbf{Basic $L^p$ estimates.} The first step in the analysis is the
construction of $L^p$ estimates. These estimates require \emph{a
  priori} control of the NP spin connection coefficients~$\rho$
and~$\mu$. The reason for their special treatment can be traced back
to their appearance in
Lemma~\ref{Lemma:DerivativeIntegrals}. Proceeding in this way we
obtain the following:

\begin{proposition}[\textbf{\em control of the~$L^p$-norm with transport
      equations}]\label{Proposition:TransportLpEstimates} Work under
  Assumption~\ref{Assumption:Frame}. Assume furthermore
  on~$\mathcal{D}_{u,v_\bullet}^{\,t}$ that
\begin{align*}
\sup_{u,v}||\{\rho,\mu\}||_{L^{\infty}(\mathcal{S}_{u,v})}\leq\mathcal{O}.
\end{align*}
Then there
exists~$\varepsilon_\star=\varepsilon_\star(\Delta_{e_\star},\mathcal{O})$
such that for all~$\varepsilon\leq\varepsilon_\star$ and for
every~$1\leq p<\infty$, we have the estimates:
\begin{align*}
  &||\phi||_{L^p(\mathcal{S}_{u,v})}\leq C(I,\mathcal{O})
  \left(||\phi||_{L^p(\mathcal{S}_{u,0})}+\int_0^v
  ||D\phi||_{L^p(\mathcal{S}_{u,v'})}\mathrm{d}v' \right), \\
  &||\phi||_{L^p(\mathcal{S}_{u,v})}\leq
  2\left(||\phi||_{L^p(\mathcal{S}_{0,v})}+C(\Delta_{e_\star},\mathcal{O})
  \int_0^u||\Delta\phi||_{L^p(\mathcal{S}_{u',v})}\mathrm{d}u'\right),
\end{align*}
where, as elsewhere,~$I$ denotes the long direction interval.
\end{proposition}

\begin{proof}
Making use of the definition of~$||\phi||_{L^p(\mathcal{S}_{u,v})}$ and the
identity in Lemma~\ref{Lemma:DerivativeIntegrals}, we have
\begin{align*}
  ||\phi||^p_{L^p(\mathcal{S}_{u,v})}&=||\phi||^p_{L^p(\mathcal{S}_{u,0})}
  +\int_0^v\frac{\mathrm{d}}{\mathrm{d} v}||\phi||^p_{L^p(\mathcal{S}_{u,v'})}
  \mathrm{d}v'\\
  &=||\phi||^p_{L^p(\mathcal{S}_{u,0})}
  +\int_0^v\left(\frac{\mathrm{d}}{\mathrm{d} v}\int_{\mathcal{S}_{u,v'}}|\phi|^p
  \right)
  \mathrm{d}v'\\
  &=||\phi||^p_{L^p(\mathcal{S}_{u,0})}
  +\int_0^v\left(\int_{\mathcal{S}_{u,v'}}(D|\phi|^p-2\rho|\phi|^p)\right)
  \mathrm{d}v'. 
\end{align*}
Now, Young's inequality gives
\begin{align*}
  D|\phi|^p=p|\phi|^{p-1}D|\phi|\leq p\left(
  \frac{\left(|\phi|^{p-1}\right)^{\frac{p}{p-1}}}{p/\left(p-1\right)}
  +\frac{\left(D|\phi|\right)^p}{p} \right)=(p-1)|\phi|^p
  +\left(D|\phi|\right)^p.
\end{align*}
Thus, we have that
\begin{align*}
  ||\phi||^p_{L^p(\mathcal{S}_{u,v})}&\leq ||\phi||^p_{L^p(\mathcal{S}_{u,0})}
  +\int_0^v\left(\int_{\mathcal{S}_{u,v'}}\left(D|\phi|\right)^p
  +(p-1-2\rho)|\phi|^p \right)\mathrm{d}v' \\
  &\leq ||\phi||^p_{L^p(\mathcal{S}_{u,0})}
  +\int_0^v\left(\int_{\mathcal{S}_{u,v'}}\left(D|\phi|\right)^p+C_1(\mathcal{O})
  |\phi|^p \right)\mathrm{d}v'\\
  &\leq ||\phi||^p_{L^p(\mathcal{S}_{u,0})}
  +\int_0^v||D\phi||^p_{L^p(\mathcal{S}_{u,v'})}\mathrm{d}v'+C_1(\mathcal{O})
  \int_0^v||\phi||^p_{L^p(\mathcal{S}_{u,v'})}\mathrm{d}v'.
\end{align*}
Now, making use of Gr\"onwall's inequality, we obtain
\begin{align*}
  ||\phi||^p_{L^p(\mathcal{S}_{u,v})}&\leq C(I, \mathcal{O})
  \left(||\phi||^p_{L^p(\mathcal{S}_{u,0})}+\int_0^v
  ||D\phi||^p_{L^p(\mathcal{S}_{u,v'})}
  \mathrm{d}v' \right) \\
  &\leq C(I, \mathcal{O})
  \left(||\phi||^p_{L^p(\mathcal{S}_{u,0})}+\left(\int_0^v
  ||D\phi||_{L^p(\mathcal{S}_{u,v'})}
  \mathrm{d}v' \right)^p\right) \\
  &\leq C(I, \mathcal{O})\left(||\phi||_{L^p(\mathcal{S}_{u,0})}
  +\int_0^v||D\phi||_{L^p(\mathcal{S}_{u,v'})}
  \mathrm{d}v' \right)^p,
\end{align*}
so that, in fact, one has
\begin{align*}
  ||\phi||_{L^p(\mathcal{S}_{u,v})}\leq C(I, \mathcal{O})\left(
  ||\phi||_{L^p(\mathcal{S}_{u,0})}
  +\int_0^v||D\phi||_{L^p(\mathcal{S}_{u,v'})}\mathrm{d}v'\right).
\end{align*}
Now, for the integration in the short direction~$0\leq
u\leq\varepsilon$, using the assumption
that~$\sup_{u,v}||\mu||_{L^{\infty}(\mathcal{S}_{u,v})}\leq\mathcal{O}$,
a similar argument as before, and now using
Lemma~\ref{Lemma:QPControl}, allows us to show that
\begin{align*}
  ||\phi||^p_{L^p(\mathcal{S}_{u,v})}\leq ||\phi||^p_{L^p(\mathcal{S}_{0,v})}
  +C(\Delta_{e_\star})\left(C(\mathcal{O})\int_0^u||\phi||^p_{L^p(\mathcal{S}_{u',v})}
  \mathrm{d}u'+\int_0^u||\Delta\phi||^p_{L^p(\mathcal{S}_{u',v})}
  \mathrm{d}u'\right),
\end{align*}
so that one has
\begin{align*}
  ||\phi||_{L^p(\mathcal{S}_{u,v})}\leq ||\phi||_{L^p(\mathcal{S}_{0,v})}
  +C(\Delta_{e_\star},\mathcal{O})\left(\int_0^u||\phi||_{L^p(\mathcal{S}_{u',v})}
  \mathrm{d}u'+\int_0^u||\Delta\phi||_{L^p(\mathcal{S}_{u',v})}\mathrm{d}u'\right).
\end{align*}
Then, using Gr\"onwall's inequality one is led to
\begin{align*}
||\phi||_{L^p(\mathcal{S}_{u,v})}\leq
\exp(C(\Delta_{e_\star},\mathcal{O})\varepsilon)\left(||\phi||_{L^p(\mathcal{S}_{0,v})}
+C(\Delta_{e_\star},\mathcal{O})\int_0^u||\Delta\phi||_{L^p(\mathcal{S}_{u',v})}
\mathrm{d}u'\right).
\end{align*}
From, the latter choosing~$\varepsilon>0$ small enough one concludes
that
\begin{align*}
  ||\phi||_{L^p(\mathcal{S}_{u,v})}\leq 2\left(||\phi||_{L^p(\mathcal{S}_{0,v})}
  +C(\Delta_{e_\star},\mathcal{O})\int_0^u||\Delta\phi||_{L^p(\mathcal{S}_{u',v})}
  \mathrm{d}u'\right).
\end{align*}
\end{proof}

As a particular example of the previous discussion
consider~$\phi=\delta f$, with~$p=2$. In this case one has
\begin{align*}
||\delta f||_{L^2(\mathcal{S}_{u,v})}\leq C(I, \mathcal{O})\left(||\delta
f||_{L^2(\mathcal{S}_{u,0})}+\int_0^v\left(\int_{\mathcal{S}_{u,v'}}
D|\delta f|^2\right)^{1/2} \mathrm{d}v'\right).
\end{align*} 
If~$p=4$ one has that
\begin{align*}
  ||\delta f||_{L^4(\mathcal{S}_{u,v})}\leq C(I, \mathcal{O})
  \left(||\delta f||_{L^4(\mathcal{S}_{u,0})}+\int_0^v
  \left(\int_{\mathcal{S}_{u,v'}}D|\delta f|^2\right)^{1/4} \mathrm{d}v'\right).
\end{align*} 
For the short direction one readily obtains analogous expressions.

\smallskip
\noindent
\textbf{Basic $L^\infty$ estimates.} Our analysis will also require
estimates on the~$L^\infty$ norm of various scalars. The first result
in this direction is the following:

\begin{proposition}[\textbf{\em supremum norm of solutions to transport
      equations}]
  \label{Proposition:SupremumNormTransportEquations} Work under
  Assumption~\ref{Assumption:Frame}. There exists~$\varepsilon_\star$ such
  that for all~$\varepsilon\leq\varepsilon_\star$, we have
\begin{align*}
  &||\phi||_{L^{\infty}(\mathcal{S}_{u,v})}\leq ||\phi||_{L^{\infty}(\mathcal{S}_{u,0})}
  +\int_0^v||D\phi||_{L^{\infty}(\mathcal{S}_{u,v'})}\mathrm{d}v',\\
  &||\phi||_{L^{\infty}(\mathcal{S}_{u,v})}\leq ||\phi||_{L^{\infty}(\mathcal{S}_{0,v})}
  +C(\Delta_{e_\star})\int_0^u||\Delta\phi||_{L^{\infty}(\mathcal{S}_{u',v})}\mathrm{d}u',
\end{align*}
on~$\mathcal{D}_{u,v_\bullet}^{\,t}$.
\end{proposition}

\begin{proof}

Given a fixed point~$(u,0,x^{\mathcal{A}})$ on~$\mathcal{N}_\star'$,
and then integrating out along integral curves of~$l^a$, conveniently
parametrizing with~$v$, gives
\begin{align*}
  \phi_{\mathcal{S}_{u,v}}-\phi_{\mathcal{S}_{u,0}}
  =\int_0^v\frac{d\phi}{dv}\mathrm{d}v'=\int_0^vD\phi \mathrm{d}v'.
\end{align*}
Fixing~$u$, varying the angular point~$x^A$ on~$\mathcal{N}_\star'$
arbitrarily, and taking the supremum we obtain the inequality of the
of the proposition. The proof of the second inequality is similar.
\end{proof}

\smallskip
\noindent
\textbf{More advanced $L^p$-estimates.} Finally, we discuss the
construction of more refined $L^p$-estimates. As in the case of the
basic~$L^p$-estimates, these estimates require some \emph{a priori}
control on the~$L^\infty$-norm of the the NP spin connection
coefficients~$\rho$ and~$\mu$. More precisely, one has the following:

\begin{proposition}[\textbf{\em $L^4$-norm of solutions to transport
      equations}]
\label{PropositionL4Estimates} 
Work under Assumption~\ref{Assumption:Frame}. Assume, as in
Proposition~\ref{Proposition:TransportLpEstimates}, furthermore that
\begin{align*}
\sup_{u,v}||\{\rho,\mu\}||_{L^{\infty}(\mathcal{S}_{u,v})}\leq\mathcal{O}.
\end{align*}
on~$\mathcal{D}_{u,v_\bullet}^{\,t}$. Then there
exists~$\varepsilon_\star=\varepsilon_\star(\Delta_{e_\star},\mathcal{O})$
such that for all~$\varepsilon\leq\varepsilon_\star$ we have the
estimates:
\begin{align*}
  &||\phi||_{L^4(\mathcal{S}_{u,v})}\leq
  C(\Delta_{e_\star},\mathcal{O})\left(||\phi||_{L^4(\mathcal{S}_{u,0})}
  +||D\phi||_{L^2(\mathcal{N}_u(0,v))}^{1/2}
  \left(||\phi||^2_{L^2(\mathcal{N}_u(0,v))}
  +||\nablasl\phi||^2_{L^2(\mathcal{N}_u(0,v))}\right)^{1/4}\right),
  \\
  &||\phi||_{L^4(\mathcal{S}_{u,v})}\leq
  2\left(||\phi||_{L^4(\mathcal{S}_{0,v})}
  +C(\Delta_{e_\star})||\Delta\phi||_{L^2(\mathcal{N}'_v(0,u))}^{1/2}
  \left(||\phi||^2_{L^2(\mathcal{N}'_v(0,u))}
  +||\nablasl\phi||^2_{L^2(\mathcal{N}'_v(0,u))}\right)^{1/4}\right),
\end{align*}
on~$\mathcal{D}_{u,v_\bullet}^{\,t}$.
\end{proposition}

\begin{proof}
The proof proceeds by direct computation. We first obtain the estimate
on the long direction. Following arguments similar to those used in
Proposition~\ref{Proposition:TransportLpEstimates}, we find that
\begin{align}
  ||\phi||^4_{L^4(\mathcal{S}_{u,v})}&=||\phi||^4_{L^4(\mathcal{S}_{u,0})}
  +\int_0^v\left(\int_{S_{u,v'}}D|\phi|^4-2\rho|\phi|^4\right)\mathrm{d}v'
  \nonumber\\
  & \leq||\phi||^4_{L^4(\mathcal{S}_{u,0})}+2\,\mathcal{O}\int_0^v
  ||\phi||^4_{L^4(\mathcal{S}_{u,v'})}\mathrm{d}v'
  +4\left(\int_{\mathcal{N}_u(0,v)}|\phi|^6
  \right)^{1/2}\left(\int_{\mathcal{N}_u(0,v)}|D\phi|^2\right)^{1/2}\,.
  \label{Prop6FirstInequality}
\end{align}
Now, for small enough~$\varepsilon$, using the Nirenberg-Sobolev
inequality (see Appendix~\ref{App:Inequalities}) we estimate:
\begin{align*}
  \int_{\mathcal{N}_u(0,v)}|\phi|^6&=\int_0^v\int_{\mathcal{S}_{u,v'}}|\phi|^6
  \mathrm{d}v'=\int_0^v|||\phi|^3||^2_{L^2(\mathcal{S}_{u,v'})}\mathrm{d}v'\\
  &\leq C(\Delta_{e_\star})\int_0^v
  \left(|||\phi|^3||_{L^1(\mathcal{S}_{u,v'})}
  +||\nablasl|\phi|^3||_{L^1(\mathcal{S}_{u,v'})}
  \right)^2\mathrm{d}v'\\ &\leq
  C(\Delta_{e_\star})\int_0^v
  \left(|||\phi|^2||_{L^2(\mathcal{S}_{u,v'})}||\phi||_{L^2(\mathcal{S}_{u,v'})}
  +|||\phi|^2||_{L^2(\mathcal{S}_{u,v'})}||\nablasl\phi||_{L^2(\mathcal{S}_{u,v'})}
  \right)^2 \mathrm{d}v'\\ &\leq
  C(\Delta_{e_\star})\int_0^v||\phi||^4_{L^4(\mathcal{S}_{u,v'})}
  \left(||\phi||_{L^2(\mathcal{S}_{u,v'})}
  +||\nablasl\phi||_{L^2(\mathcal{S}_{u,v'})}
  \right)^2\mathrm{d}v'\\ &\leq2C(\Delta_{e_\star})
  \left(\sup_{u,v}||\phi||^4_{L^4(\mathcal{S}_{u,v})}\right)
  \int_0^v\left(||\phi||^2_{L^2(\mathcal{S}_{u,v'})}
  +||\nablasl\phi||^2_{L^2(\mathcal{S}_{u,v'})}\right) \mathrm{d}v'
  \\
  &\leq C(\Delta_{e_\star})
  \left(\sup_{u,v}||\phi||^4_{L^4(\mathcal{S}_{u,v})}\right)
  \left(||\phi||^2_{L^2(\mathcal{N}_u(0,v))}
  +||\nablasl\phi||^2_{L^2(\mathcal{N}_u(0,v))}\right),
\end{align*}
where to pass from the second to the third line we have made use of
H\"older's inequality and, to pass from the third to fourth we have
extracted common factors. Making use of the above estimate in
inequality~\eqref{Prop6FirstInequality}, we have that
\begin{align*}
  ||\phi||^4_{L^4(\mathcal{S}_{u,v})}&\leq||\phi||^4_{L^4(\mathcal{S}_{u,0})}
  +2\,\mathcal{O}\int_0^v||\phi||^4_{L^4(\mathcal{S}_{u,v'})}\mathrm{d}v'\\
  &\quad+C(\Delta_{e_\star})\left(\sup_{u,v}||\phi||^2_{L^4(\mathcal{S}_{u,v})}\right)
  \left(||\phi||^2_{L^2(\mathcal{N}_u(0,v))}
  +||\nablasl\phi||^2_{L^2(\mathcal{N}_u(0,v))}\right)^{1/2}
  ||D\phi||_{L^2(\mathcal{N}_u(0,v))}\\
  &\leq||\phi||^4_{L^4(\mathcal{S}_{u,0})}+2\,\mathcal{O}\int_0^v
  ||\phi||^4_{L^4(\mathcal{S}_{u,v'})}\mathrm{d}v'
  +C(\Delta_{e_\star})\delta
  \left(\sup_{u,v}||\phi||^4_{L^4(\mathcal{S}_{u,v})}\right)\\
  &\quad
  +\frac{C(\Delta_{e_\star})}{4\delta}
  \left(||\phi||^2_{L^2(\mathcal{N}_u(0,v))}
  +||\nablasl\phi||^2_{L^2(\mathcal{N}_u(0,v))}\right)
  ||D\phi||^2_{L^2(\mathcal{N}_u(0,v))},
\end{align*}
for some~$\delta>0$. Now, choosing~$\delta$ sufficiently small and
making use of Gr\"onwall's inequality, one finally obtains that
\begin{align*}
  ||\phi||^4_{L^4(\mathcal{S}_{u,v})}\leq C(\Delta_{e_\star},\mathcal{O})
  \left(||\phi||^4_{L^4(\mathcal{S}_{u,0})}
  +||D\phi||_{L^2(\mathcal{N}_u(0,v))}^2
  \left(||\phi||^2_{L^2(\mathcal{N}_u(0,v))}
  +||\nablasl\phi||^2_{L^2(\mathcal{N}_u(0,v))}\right)\right).
\end{align*}
The proof of the estimate along the short direction is similar. In
this case we can choose~$\varepsilon>0$ sufficiently small to make the
overall constant equal to, say,~$2$.
\end{proof}

\subsection{Sobolev inequalities}

In the last step in our preparatory work, we now obtain Sobolev-type
inequalities on the spheres~$\mathcal{S}_{u,v}$ ---i.e. estimates of
the~$L^p$-norms of a scalar in terms of its $L^2$-norms and those of
its derivatives. The key tool in this analysis is the
\emph{isoperimetric Sobolev inequality on~$\mathcal{S}_{u,v}$}
---see~\cite{ChrKla93}:

\begin{theorem}[\textbf{\em isoperimetric Sobolev inequality
      on~$\mathcal{S}_{u,v}$}] Let~$\phi$ denote an integrable
  function and with integrable first derivatives
  on~$\mathcal{S}_{u,v}$. Then we have that
\begin{align}
\label{isoperimetricinequality}
\int_{\mathcal{S}_{u,v}}|\phi-\bar \phi|^2\leq
\mathcal{I}(\mathcal{S}_{u,v})\left(\int_{\mathcal{S}_{u,v}}|\nablasl
\phi|\right)^2,
\end{align}
where~$\bar{\phi}$ denotes the average of~$\phi$
over~$\mathcal{S}_{u,v}$ and~$\mathcal{I}(\mathcal{S}_{u,v})$ is the
\emph{isoperimetric constant}.
\end{theorem}

\begin{remark}
{\em The isoperimetric inequality can be shown to be controlled by the
  area of the 2-dimensional surfaces~$\mathcal{S}_{u,v}$ ---see
  e.g.~\cite{ChrKla93}. Thus, if one has control over the area of the
  surface (as it is, in principle, in our setup), one has also control
  over the isoperimetric constant.}
\end{remark}

Using this we can prove the following result concerning Sobolev-type
inequalities:

\begin{proposition}[\textbf{\em Sobolev-type inequality. I}]
\label{Proposition:Sobolev} Work under
  Assumption~\ref{Assumption:Frame}. Let~$\phi$ be a scalar field
  on~$\mathcal{S}_{u,v}$ which is square-integrable with
  square-integrable first covariant derivatives. Then for
  each~$2<p<\infty$, $\phi\in L^p(\mathcal{S}_{u,v})$, there
  exists~$\varepsilon_\star=\varepsilon_\star(\Delta_{e_\star},\Delta_{\Gamma})$
  such that as long as~$\varepsilon\leq\varepsilon_\star$, we have
  \begin{align*}
||\phi||_{L^p(\mathcal{S}_{u,v})}\leq
G_p(\bmsigma)\left(||\phi||_{L^2(\mathcal{S}_{u,v})}
+||\nablasl\phi||_{L^2(\mathcal{S}_{u,v})}\right)
\end{align*} 
where~$G_p(\bmsigma)$ is a constant also depends on the isoperimetric
constant~$\mathcal{I}(\mathcal{S}_{u,v})$ and~$p$, but is controlled
by some~$C(\Delta_{e_\star})$, $\nablasl$ is the induced connection
on~$\mathcal{S}_{u,v}$ which is associated with the
metric~$\bm\sigma$.
\end{proposition}

\begin{proof}
We make use of the following result which can be found in Lemma~$5.1$
in Chapter~$5.2$ of~\cite{Chr08}:
\begin{align}
\label{SobolevLp1}
\left(\mbox{Area}(\mathcal{S}_{u,v})\right)^{-1/p}
||\phi||_{L^p(\mathcal{S}_{u,v})}\leq
C_p\sqrt{\mathcal{I}'(\mathcal{S}_{u,v})}
\left(\left(\mbox{Area}(\mathcal{S}_{u,v})\right)^{-1/2}
||\phi||_{L^2(\mathcal{S}_{u,v})}+||\nablasl\phi||_{L^2(\mathcal{S}_{u,v})}\right),
\end{align} 
where~$C_p$ is a numerical constant depending only on~$p$,
\begin{align*}
\mathcal{I}'(\mathcal{S}_{u,v})=\max\{1,\mathcal{I}(\mathcal{S}_{u,v})\},
\end{align*}
where as above~$\mathcal{I}(\mathcal{S}_{u,v})$ is the isoperimetric
constant of~$\mathcal{S}_{u,v}$. Now, under
Assumption~\ref{Assumption:Frame} we have that the area
of~$\mathcal{S}_{u,v}$ is finite in the tilted rectangle. Accordingly,
inequality~\eqref{SobolevLp1} can be adapted to our particular
setting.
\end{proof}

Consequently we have the following two results:

\begin{proposition}[\textbf{\em Sobolev-type inequality. II}]
  \label{Proposition:EstimateInfinityNorm} Work under
  Assumption~\ref{Assumption:Frame}. There
  exists~$\varepsilon_\star=\varepsilon_\star(\Delta_{e_\star},\Delta_{\Gamma})$
  such that as long as~$\varepsilon\leq\varepsilon_\star$, we have
\begin{align*}
  ||\phi||_{L^{\infty}(\mathcal{S}_{u,v})}\leq G_{p}(\bmsigma)
  \left(||\phi||_{L^p(\mathcal{S}_{u,v})}+||
  \nablasl\phi||_{L^p(\mathcal{S}_{u,v})}\right),
\end{align*}
with~$2<p<\infty$ and~$G_{p}(\bmsigma)\leq C(\Delta_{e_\star})$ as above.
\end{proposition}

\begin{corollary}[\textbf{\em Sobolev-type inequality. III}]
\label{Corollary:SobolevEmbedding}
  Work under Assumption~\ref{Assumption:Frame}. There
  exists~$\varepsilon_\star=\varepsilon_\star(\Delta_{e_\star},\Delta_{\Gamma})$
  such that as long as~$\varepsilon\leq\varepsilon_\star$, we have
\begin{align*}
  &||\phi||_{L^4(\mathcal{S}_{u,v})}\leq G(\bmsigma)
  \left(||\phi||_{L^2(\mathcal{S}_{u,v})}
  +||\nablasl\phi||_{L^2(\mathcal{S}_{u,v})}\right), \\
  &||\phi||_{L^{\infty}(\mathcal{S}_{u,v})}\leq G(\bmsigma)
  \left(||\phi||_{L^2(\mathcal{S}_{u,v})}
  +||\nablasl\phi||_{L^2(\mathcal{S}_{u,v})}
  +||\nablasl^2\phi||_{L^2(\mathcal{S}_{u,v})}\right),
\end{align*} 
again with~$G(\bmsigma)\leq C(\Delta_{e_\star})$.
\end{corollary}

\section{Main estimates}

In this section we provide a discussion of the construction of the
main estimates required to obtain the improved existence result for
the CIVP. The arguments rely heavily on the preparatory work carried
out in the previous section.

\subsection{Norms for the initial data}

The boostrap argument requires assumptions on the size of the initial
data. Following Luk~\cite{Luk12}, we define the following:

\begin{enumerate}[i).]
\item Norm for the initial value of the connection coefficients, given
  by
\begin{align*}
  \Delta_{\Gamma_\star} \equiv \sup_{\mathcal{S}_{u,v}
    \subset\mathcal{N}_\star,\mathcal{N}'_\star}
  \sup_{\Gamma\in\{\mu,\lambda,\rho,\sigma,\alpha,\beta,\tau,\epsilon\}}
  \max\{1,||\Gamma||_{L^{\infty}(\mathcal{S}_{u,v})},
  \sum_{i=0}^1||\nablasl^i\Gamma||_{L^4(\mathcal{S}_{u,v})},
  \sum_{i=0}^2||\nablasl^i\Gamma||_{L^2(\mathcal{S}_{u,v})}\}.
\end{align*}

\item Norm for the initial value of the components of the Weyl tensor,
  given by
\begin{align*}
  &\Delta_{\Psi_\star}\equiv \sup_{\mathcal{S}_{u,v}\subset\mathcal{N}_\star,\mathcal{N}'_\star}
  \sup_{\Psi\in\{\Psi_0,\Psi_1,\Psi_2,\Psi_3,\Psi_4\}}
  \max\{1,\sum_{i=0}^1||\nablasl^i\Psi||_{L^4(\mathcal{S}_{u,v})},
  \sum_{i=0}^2||\nablasl^i\Psi||_{L^2(\mathcal{S}_{u,v})}\} \\
  &+\sum_{i=0}^3\sup_{\Psi\in\{\Psi_0,\Psi_1,\Psi_2,\Psi_3\}}
  ||\nablasl^i\Psi||_{L^2(\mathcal{N}_\star)}
  +\sup_{\Psi\in\{\Psi_1,\Psi_2,\Psi_3,\Psi_4\}}||\nablasl^i\Psi||_{L^2(\mathcal{N}'_\star)}.
\end{align*}

\item Norm for the components of the Weyl tensor at later null
  hypersurfaces, given by
\begin{align*}
  \Delta_{\Psi}\equiv  \sum_{i=0}^3\sup_{\Psi\in\{\Psi_0,\Psi_1,\Psi_2,\Psi_3\}}\sup_{u}
  ||\nablasl^i\Psi||_{L^2(\mathcal{N}_{u}^{\,t})}
  +\sup_{\Psi\in\{\Psi_1,\Psi_2,\Psi_3,\Psi_4\}}\sup_{v}
  ||\nablasl^i\Psi||_{L^2(\mathcal{N}^{'t}_{v})}
\end{align*}
where the suprema in~$u$ and~$v$ are taken
over~$\mathcal{D}_{u,v_\bullet}^{\,t}$.

\item Sup over the~$L^2$-norm of the components of the Weyl tensor at
  spheres of constant~$u,v$, given by,
\begin{align*}
  \Delta_{\Psi}(\mathcal{S})=\sum_{i=0}^2\sup_{u,v}
    ||\nablasl^i(\Psi_0,\Psi_1,\Psi_2,\Psi_3)||_{L^2(\mathcal{S}_{u,v})},
\end{align*}
with the supremum taken over~$\mathcal{D}_{u,v_\bullet}^{\,t}$, and in
which~$u$ will be taken sufficiently small to apply our estimates.
\end{enumerate}

\begin{remark}
{\em There is no appearance of~$\chi$ in~$\Delta_{\Gamma_{\star}}$
  because initial data for~$\chi$ used in the following calculations
  are required only on~$\mathcal{N}'_{\star}$ where~$\chi$ is zero.}
\end{remark}

\begin{remark}
{\em In addition to the above norms, we recall that the
  norm~$\Delta_{e_\star}$, as defined in
  equation~\eqref{Definition:DeltaEstar} has been used to control the
  initial value of the components of the frame.}
\end{remark}

\begin{remark}
{\em Observe that the above expressions do not include any norm for
  the components of the connection coefficients away from the initial
  null hypersurfaces. Instead such norms will be controlled by local
  bootstrap arguments within the proof.}
\end{remark}

\begin{remark}
{\em Throughout the proof besides keeping track
  of~$\Delta_{\Psi_\star}$ and~$\Delta_{\Psi_\star}(\mathcal{S})$, to
  assist in future generalization, we trace also the dependence of our
  various constants
  on~$I,\Delta_{e_\star},\Delta_{\Gamma_\star},\Delta_{\Psi_\star}$.
  Note that because of the way that we setup our frame none of the
  constants so far depend upon~$I$.}
\end{remark}

\subsection{Estimates for the connection coefficients}

In this section we show how to construct estimates on the coefficients
of the connection. The strategy is an application of the tools
developed in Section~\ref{Section:EstimatesTransportEquations} to
estimate the solutions of generic transport equations along null
hypersurfaces. In this approach, as a bootstrap, control is assumed of
the curvature (components of the Weyl tensor) on the double foliation
of null hypersurfaces and on the 2-spheres of constant~$u$ and~$v$
through the norms~$\Delta_\Psi$ and~$\Delta_\Psi(\mathcal{S})$.

In a first step we obtain basic control of the~$L^\infty$-norm of the
connection coefficients by assuming finiteness of~$\Delta_\Psi$
and~$\Delta_\Psi(\mathcal{S})$ and of third derivatives of the NP
coefficient~$\tau$ in terms of the~$L^2$-norm on the
2-spheres~$\mathcal{S}_{u,v}$.

\begin{proposition}[\textbf{\em control on the supremum norm of the
      connection coefficients}]
  \label{Proposition:FirstEstimateConnection}
  Assume that we have a solution of the vacuum EFEs in Stewart's gauge
  in a region~$\mathcal{D}_{u,v_\bullet}^{\,t}$ with
\begin{align*}
  \sup_{u,v}||\{\mu, \lambda, \alpha, \beta, \epsilon, \rho, \sigma,
  \tau, \chi\}||_{L^\infty(S_{u,v})}\leq \Delta_\Gamma\,,
\end{align*}
for some positive~$\Delta_\Gamma$. Assume also
\begin{align*}
  \sup_{u,v}||\nablasl^2\tau||_{L^2(S_{u,v})}<\infty, \qquad
  \sup_{u,v}||\nablasl^3\tau||_{L^2(S_{u,v})}<\infty, \qquad
  \Delta_{\Psi}(\mathcal{S})<\infty, \qquad \Delta_{\Psi}<\infty,
\end{align*}
on the same domain. Then there exists
\begin{align*}
  \varepsilon_\star =\varepsilon_\star(I,\Delta_{e_\star},\Delta_{\Gamma_\star},
  \sup_{u,v}||\nablasl^2\tau||_{L^2(\mathcal{S}_{u,v})},
  \sup_{u,v}||\nablasl^3\tau||_{L^2(\mathcal{S}_{u,v})},\Delta_{\Psi}),
\end{align*} such that
when~$\varepsilon\leq\varepsilon_\star$, we have
\begin{align*}
  & \sup_{u,v}||\{\tau,\chi\}||_{L^{\infty}(\mathcal{S}_{u,v})}\leq
  C(I,\Delta_{e_\star},\Delta_{\Gamma_\star},\Delta_{\Psi}(\mathcal{S})),\\
  &\sup_{u,v}||\{\mu, \lambda, \alpha, \beta, \epsilon, \rho, \sigma\}
  ||_{L^{\infty}(\mathcal{S}_{u,v})}\leq3\Delta_{\Gamma_\star}, 
\end{align*}
on~$\mathcal{D}_{u,v_\bullet}^{\,t}$.
\end{proposition}

\begin{remark}
{\em Observe that in the above proposition, as well as in several of
  the following ones, the NP spin connection coefficient~$\tau$ is
  singled out as it requires additional hypotheses.}
\end{remark}

\begin{remark}
  {\em The first assumption here covers
      Assumption~\ref{Assumption:Frame}, which allows us to employ
      Lemma~\ref{Lemma:QPControl}, Corollary~\ref{Corrollary:Area},
      Lemma~\ref{Lemma:CControl},
      Proposition~\ref{Proposition:SupremumNormTransportEquations} and
      the Sobolev inequalities of
      Propositions~\ref{Proposition:Sobolev},
      \ref{Proposition:EstimateInfinityNorm} and
      Corollary~\ref{Corollary:SobolevEmbedding}. It also permits the
      use of Propositions~\ref{Proposition:TransportLpEstimates}
      and~\ref{PropositionL4Estimates}.}
\end{remark}

\begin{proof}
$\phantom{X}$

\smallskip
\noindent
\textbf{Basic bootstrap assumption.} We start by making the bootstrap
assumption
\begin{align*}
  \sup_{u,v}||(\{\mu, \lambda, \alpha, \beta,\epsilon, \rho, \sigma
  \}||_{L^{\infty} (\mathcal{S}_{u,v})}\leq4\Delta_{\Gamma_\star}.
\end{align*}

\smallskip
\noindent
\textbf{Estimate for~$\tau$.} As first step we prove that
\begin{align*}
||\tau||_{L^{\infty}(\mathcal{S}_{u,v})}\leq
C(I,\Delta_{e_\star},\Delta_{\Gamma_\star},\Delta_{\Psi}(\mathcal{S})).
\end{align*}
For this, we make use of the~$D$-equation~\eqref{structureeq2} for the
NP coefficient~$\tau$:
\begin{align*}
  D\tau=(\epsilon-\bar\epsilon+\rho)\tau+\sigma\bar\tau+\bar\pi\rho
  +\pi\sigma+\Psi_1.
\end{align*}
Making use of the Sobolev inequality in
Proposition~\ref{Proposition:EstimateInfinityNorm}, we readily obtain
from our assumptions that for~$\varepsilon$ sufficiently small,
\begin{align*}
  ||\Psi_0, \, \Psi_1,\, \Psi_2,\, \Psi_3,\, \Psi_4||_{L^{\infty}(\mathcal{S}_{u,v})}
  \leq\Delta_{\Psi}(\mathcal{S})<\infty.
\end{align*}  
Moreover, the inequalities in
Proposition~\ref{Proposition:SupremumNormTransportEquations} show that
\begin{align*}
  ||\tau||_{L^{\infty}(\mathcal{S}_{u,v})}&\leq||\tau||_{L^{\infty}(\mathcal{S}_{u,0})}
  +\int_0^v||D\tau||_{L^{\infty}(\mathcal{S}_{u,v'})}\mathrm{d}v' \\
  & \leq||\tau||_{L^{\infty}(\mathcal{S}_{u,0})}+\int_0^v||\bar\pi\rho+\pi\sigma
  +\Psi_1||_{L^{\infty}(\mathcal{S}_{u,v'})}\mathrm{d}v' \\
  &\quad\quad+\int_0^v|\epsilon-\bar\epsilon
  +\rho|\,||\tau||_{L^{\infty}(\mathcal{S}_{u,v'})}
  \mathrm{d}v'+\int_0^v|\sigma|\,
  ||\bar\tau||_{L^{\infty}(\mathcal{S}_{u,v'})}\mathrm{d}v'\\
  &\leq\Delta_{\Gamma_\star}+(32\Delta_{\Gamma_\star}^2
  +\Delta_{\Psi}(\mathcal{S}))v_\bullet
  +16\Delta_{\Gamma_\star}\int_0^v
  ||\tau||_{L^{\infty}(\mathcal{S}_{u,v'})}\mathrm{d}v'.
\end{align*}  
Using Gr\"onwall's inequality in the previous expression one then
concludes that
\begin{align*}
  ||\tau||_{L^{\infty}(\mathcal{S}_{u,v})}\leq
  C(I,\Delta_{e_\star},\Delta_{\Gamma_\star},
  \Delta_{\Psi}(\mathcal{S})).
\end{align*}

\smallskip
\noindent
\textbf{Estimate for~$\chi$.} To obtain the estimate for~$\chi$ we
proceed in a similar manner. We use the $D$-transport equation
equation~\eqref{EqDchi} for~$\chi$ to obtain 
\begin{align*}
  &||\chi||_{L^{\infty}(\mathcal{S}_{u,v})}\leq ||\chi||_{L^{\infty}(\mathcal{S}_{u,0})}
  +\int_0^v||D\chi||_{L^{\infty}(\mathcal{S}_{u,v'})}\mathrm{d}v'\\
  &\leq(2\Delta_{\Psi}(\mathcal{S})
  +c\Delta_{\Gamma_\star}+C)v_\bullet+2\Delta_{\Gamma_\star}\int_0^v
  ||\chi||_{L^{\infty}(\mathcal{S}_{u,v'})}\mathrm{d}v',
\end{align*} 
where~$c$ is a positive constant and the constant~$C$ is related to
the constant appearing in the estimate for~$\tau$. From the latter,
Gr\"onwall's inequality readily yields
\begin{align*}
  ||\chi||_{L^{\infty}(\mathcal{S}_{u,v})}\leq C(I,\Delta_{e_\star},
  \Delta_{\Gamma_\star},\Delta_{\Psi}(\mathcal{S})).
\end{align*}

\smallskip
\noindent
\textbf{Estimates for~$\mu$ and~$\lambda$.} To obtain estimates of the
NP coefficients~$\mu$ and~$\lambda$ we make use of
the~$\Delta$-transport equations~\eqref{structureeq7}
and~\eqref{structureeq15}:
\begin{align*}
&\Delta\mu=-\mu^2-\lambda\bar\lambda, \\
& \Delta\lambda=-2\mu\lambda-\Psi_4.
\end{align*}  
These are Riccati-type equations and, thus, they can only be naively
integrated for a small distance in the~$u$ direction
---i.e. $u\in[0,\varepsilon]$. Now, making use of the inequalities in
Proposition~\ref{Proposition:SupremumNormTransportEquations} we find
that
\begin{align*}
  ||\mu||_{L^{\infty}(\mathcal{S}_{u,v})}\leq||\mu||_{L^{\infty}(\mathcal{S}_{0,v})}
  +C(\Delta_{e_\star})\int_0^{\varepsilon}
  ||\Delta\mu||_{L^{\infty}(\mathcal{S}_{u',v})}\mathrm{d}u'. 
\end{align*}
Accordingly, one concludes that
\begin{align*}
  ||\mu||_{L^{\infty}(\mathcal{S}_{u,v})}&\leq||\mu||_{L^{\infty}(\mathcal{S}_{0,v})}
  +C(\Delta_{e_\star})\int_0^{\varepsilon}
  ||\mu^2+\lambda\bar\lambda||_{L^{\infty}(\mathcal{S}_{u',v})}\mathrm{d}u'\\
  &\leq||\mu||_{L^{\infty}(\mathcal{S}_{0,v})}+32C(\Delta_{e_\star})
  \int_0^{\varepsilon}\Delta_{\Gamma_0}^2\mathrm{d}u'\\
  &\leq\Delta_{\Gamma_\star}+32C(\Delta_{e_\star})\Delta_{\Gamma_\star}^2\varepsilon. 
\end{align*}
For~$\lambda$ one obtains that
\begin{align*}
  ||\lambda||_{L^{\infty}(\mathcal{S}_{u,v})}&\leq\Delta_{\Gamma_\star}
  +32C(\Delta_{e_\star})\Delta_{\Gamma_\star}^2\varepsilon
  +C(\Delta_{e_\star})\int_0^u
  ||\Psi_4||_{L^{\infty}(\mathcal{S}_{u',v})}du'\\
  &\leq\Delta_{\Gamma_\star}+C(\Delta_{e_\star},\Delta_{\Gamma_\star})\varepsilon
  +C(\Delta_{e_\star})\int_0^u\sum_{i=0}^2
  ||\nablasl^i\Psi_4||_{L^{2}(\mathcal{S}_{u',v})}\mathrm{d}u',
\end{align*}
where in the second inequality we have made use of the Sobolev
embedding property ---see
corollary~\ref{Corollary:SobolevEmbedding}. Now, using H\"older's
inequality, we can transform the estimate of~$\Psi_4$ from one on
sphere~$\mathcal{S}_{u,v}$ to one on a null hypersurface. More
precisely, one has that
\begin{align*}
  \int_0^u||\nablasl^i\Psi_4||_{L^{2}(\mathcal{S}_{u',v})}\mathrm{d}u'
  &=\int_0^u\left(\int_{\mathcal{S}_{u',v}}|\nablasl^i\Psi_4|^2\right)^{1/2}
  \mathrm{d}u'\leq\left(\int_0^u\int_{\mathcal{S}_{u',v}}|\nablasl^i\Psi_4|^2
  \mathrm{d}u'\right)^{1/2}\left(\int_0^u 1 \mathrm{d}u'\right)^{1/2}\\
  &\leq C\varepsilon^{1/2}||\nablasl^i\Psi_4||_{L^{2}(\mathcal{N}'_v(0,u))}.
\end{align*}
Hence, we conclude that
\begin{align*}
||\lambda||_{L^{\infty}(\mathcal{S}_{u,v})}\leq\Delta_{\Gamma_\star}
+C(\Delta_{e_\star},\Delta_{\Gamma_\star})\varepsilon
+C\Delta_{\Psi}\varepsilon^{1/2}.
\end{align*}
Together, the estimates for~$\mu$ and~$\lambda$ show that the maximum
of these functions will not be too far away from their initial value
for~$\varepsilon$ sufficiently small.

\smallskip
\noindent
\textbf{Estimates for~$\alpha$, $\beta$ and~$\epsilon$.}
Estimates~$\alpha$, $\beta$ and~$\epsilon$ can be obtained by a
similar method ---i.e. integration along the \emph{short}
direction. In this case the relevant~$\Delta$-transport equations are
given by the structure equations~\eqref{structureeq11},
\eqref{structureeq4} and~\eqref{structureeq1},
\begin{align*}
&\Delta\alpha=-\mu\alpha-\lambda\beta-\lambda\tau-\Psi_3, \\
&\Delta\beta=-\bar\lambda\alpha-\mu\beta-\tau\mu, \\
& \Delta\epsilon=-\alpha\bar\pi-\beta\pi-\alpha\tau-\beta\bar\tau
-\pi\tau-\Psi_2,
\end{align*}
where it is recalled that in the present gauge one has
that~$\pi=\alpha+\bar{\beta}$ ---see Lemma~\ref{Lemma1},
equation~\eqref{spinconnection3}. The details are omitted.

\smallskip
\noindent
\textbf{Estimates for~$\rho$ and~$\sigma$.} In this case the
relevant~$\Delta$-transport equations are the structure
equations~\eqref{structureeq9} and~\eqref{structureeq18}:
\begin{align*}
  \Delta\rho&=\bar\delta\tau-\mu\rho-\lambda\sigma-\alpha\tau+\bar\beta\tau
  -\tau\bar\tau-\Psi_2,\\
  \Delta\sigma&=\delta\tau-\bar\lambda\rho-\mu\sigma+\bar\alpha\tau
  -\beta\tau-\tau^2.
\end{align*}
Observe that these equations contain the derivatives~$\delta \tau$
and~$\bar\delta\tau$. To control these terms from our hypotheses, we
make use of the Sobolev inequalities in
corollary~\ref{Corollary:SobolevEmbedding} which, together with
integration by parts on~$\mathcal{S}_{u,v}$ allows us to show that,
\begin{align*}
  ||\nablasl\tau||_{L^{\infty}(\mathcal{S}_{u,v})}\leq C(\Delta_{e_\star})
  \sum_{i=1}^{3}||\nablasl^i\tau||_{L^2(\mathcal{S}_{u,v})}\leq
  C(\Delta_{e_\star})\big(||\tau||_{L^2(\mathcal{S}_{u,v})}
  +||\nablasl^2\tau||_{L^2(\mathcal{S}_{u,v})}
  +||\nablasl^3\tau||_{L^2(\mathcal{S}_{u,v})}\big).
\end{align*}
It follows then from the H\"older inequality 
\begin{align*}
||\tau||_{L^2(\mathcal{S}_{u,v})}\leq||\tau||_{L^{\infty}(\mathcal{S}_{u,v})}
\,\mbox{Area}(\mathcal{S}_{u,v})^{1/2}
\end{align*}
and the boundedness assumptions
on~$||\nablasl^i\tau||_{L^2(\mathcal{S}_{u,v})}$ for~$i=2,3$, that
\begin{align*}
||\nablasl\tau||_{L^{\infty}(\mathcal{S}_{u,v})} <\infty.
\end{align*}
From this observation, an argument similar to the one used for~$\mu$
and~$\lambda$ yields the required estimates.

\smallskip
\noindent
\textbf{Concluding the argument.} From the estimates for the NP
connection coefficients constructed above it follows that one can
choose
\begin{align*}
  \varepsilon=\varepsilon(I,\Delta_{e_{\star}},\Delta_{\Gamma_\star},
  \sup_{u,v}||\nablasl^2\tau||_{L^2(\mathcal{S}_{u,v})},
  \sup_{u,v}||\nablasl^3\tau||_{L^2(\mathcal{S}_{u,v})},
  \Delta_{\Psi}(\mathcal{S}),\Delta_{\Psi})
\end{align*}
sufficiently small so that
\begin{align*}
  \sup_{u,v}||\{\mu,\lambda,\alpha,\beta,\epsilon,\rho,\sigma\}
  ||_{L^{\infty}(\mathcal{S}_{u,v})} \leq3\Delta_{\Gamma_\star}.
\end{align*}
Accordingly, we have improved our initial bootstrap assumption. As
this is our first such improvement we give an overview of the
technique. Recall that to complete a bootstrap argument we need first,
to verify that the hypothesis, in our case
that~$\sup_{u,v}||\{\mu,\lambda,\alpha,\beta,\epsilon,\rho,\sigma\}
||_{L^{\infty}(\mathcal{S}_{u,v})}\leq4\Delta_{\Gamma_\star}$ holds
over the region of interest, is satisfied. We then need to
demonstrate, as in the previous argument, that the hypothesis can be
improved for~$\varepsilon$ sufficiently small. Obviously if the
conclusion~$\sup_{u,v}||\{\mu,\lambda,\alpha,\beta,\epsilon,\rho,\sigma\}
||_{L^{\infty}(\mathcal{S}_{u,v})} \leq3\Delta_{\Gamma_\star}$ holds
at some point then our hypothesis holds in a neighborhood of that
point. Since the interval~$[0,\varepsilon]$ is connected and the set
on which our desired conclusion holds is open, closed and non-empty,
it follows that the desired conclusion holds
for~$u\in[0,\varepsilon]$. In the argument above we have shown that we
can improve the hypothesis from a bound~$4\Delta_{\Gamma_\star}$
to~$3\Delta_{\Gamma_\star}$. Evidently the same arguments could be
used to improve from~$(N+1)\Delta_{\Gamma_\star}$
to~$N\Delta_{\Gamma_\star}$ for any natural number~$N\geq3$. Given our
initial assumption
that~$||\{\mu,\lambda,\alpha,\beta,\epsilon,\rho,\sigma\}
||_{L^{\infty}(\mathcal{S}_{u,v})}\leq\Delta_{\Gamma}$ we can
therefore choose~$N$ so that~$\Delta_\Gamma\leq
N\Delta_{\Gamma_\star}$ and iterate from~$N$ down to~$4$ to guarantee
that our hypothesis is indeed satisfied in some truncated diamond,
demonstrating the statement.
\end{proof}

The existence proof also requires control over the~$L^4$-norms of
the~$\delta$ and~$\bar\delta$ derivatives of the NP spin connection
coefficients. This is provided by the following:

\begin{proposition}[\textbf{\em control on the~$L^4$-norm of the
    connection coefficients}]
\label{PropositionSecondEstimateConnection}
Make the same assumptions as in
Proposition~\ref{Proposition:FirstEstimateConnection}, and
additionally assume that,
\begin{align*}
  \sup_{u,v}||\nablasl\{\mu, \lambda, \alpha, \beta, \epsilon, \rho,
  \sigma\} ||_{L^4(\mathcal{S}_{u,v})}\leq\Delta_{\Gamma},
\end{align*}
in the truncated diamond~$\mathcal{D}_{u,v_\bullet}^{\,t}$.
Then there exists,
\begin{align*}
  \varepsilon_\star=\varepsilon_\star(I,\Delta_{e_\star},\Delta_{\Gamma_\star},
  \sup_{u,v}||\nablasl^2\tau||_{L^2(\mathcal{S}_{u,v})},
  \sup_{u,v}||\nablasl^3\tau||_{L^2(\mathcal{S}_{u,v})},
  \Delta_{\Psi}(\mathcal{S}),\Delta_{\Psi}),
\end{align*}
such that when~$\varepsilon\leq\varepsilon_\star$, we have,
\begin{align*}
  &\sup_{u,v}||\nablasl\{\tau,\chi\}||_{L^4(\mathcal{S}_{u,v})}\leq
  C(I,\Delta_{e_{\star}},\Delta_{\Gamma_\star},\Delta_{\Psi}(\mathcal{S})), \\
  &\sup_{u,v}||
  \nablasl\{\mu, \lambda, \alpha, \beta, \epsilon, \rho, \sigma\}
  ||_{L^4(\mathcal{S}_{u,v})}\leq3\Delta_{\Gamma_\star},
\end{align*}
on~$\mathcal{D}_{u,v_\bullet}^{\,t}$. 
\end{proposition}

\begin{proof}
$\phantom{X}$
  
\smallskip
\noindent
\textbf{Basic bootstrap assumption.} In order to run the argument we
make the following bootstrap assumption:
\begin{align*}
  \sup_{u,v}||\nablasl\{\mu, \lambda, \alpha, \beta, \epsilon, \rho, \sigma\}
  ||_{L^4(\mathcal{S}_{u,v})}\leq4\Delta_{\Gamma_\star}.
\end{align*}

\smallskip
\noindent
\textbf{Estimates for $\nablasl\tau$.} First we make use of the
boundedness of the~$L^2$-norm of~$\tau$ and its angular derivatives up
to third order to estimate the~$L^4$-norm of the first order angular
derivatives of~$\tau$. For this, we apply~$\delta$ to the
$D$-transport equation for~$\tau$
---equation~\eqref{structureeq2}. After making use of the commutators
of directional covariant derivatives one arrives at the equations
\begin{subequations}
\begin{align}
  D\delta\tau&=(\rho+\bar\rho+2\epsilon-2\bar\epsilon)\delta\tau
  +\sigma\bar\delta\tau+\sigma\delta\bar\tau
  +\delta(\epsilon-\bar\epsilon+\rho)\tau\nonumber\\
  & \quad
  +\bar\tau\delta\sigma+\rho\delta\bar\pi+\bar\pi\delta\rho
  +\sigma\delta\pi+\pi\delta\sigma+\delta\Psi_1, \label{Ddeltau1}\\
  D\bar\delta\tau&=2\rho\bar\delta\tau+\sigma\bar\delta\bar\tau
  +\bar\sigma\delta\tau+\tau\bar\delta(\epsilon-\bar\epsilon+\rho)
  +\bar\tau\bar\delta\sigma\nonumber  \\
  &\quad+\rho\bar\delta\bar\pi+\bar\pi\delta\sigma\rho
  +\sigma\bar\delta\pi+\pi\bar\delta\sigma+\bar\delta\Psi_1.\label{Ddeltau2}
\end{align}  
\end{subequations}
The above equation contains terms of the form~$\Gamma\nablasl\Gamma$
---i.e. products of connection coefficients and their derivatives. In
the following the~$L^4$-norm of these products will be split using the
H\"older inequality as follows:
\begin{align*}
  ||\Gamma\nablasl\Gamma||_{L^4(\mathcal{S}_{u,v})}\leq
  ||\Gamma||_{L^{\infty}(\mathcal{S}_{u,v})}||\nablasl\Gamma||_{L^4(\mathcal{S}_{u,v})}.
\end{align*}
Observe that from
Proposition~\ref{Proposition:FirstEstimateConnection} it follows that
terms of the type~$||\Gamma||_{L^{\infty}(\mathcal{S}_{u,v})}$ are
bounded.

Now, making use of the Sobolev inequality in
Proposition~\ref{Proposition:Sobolev}, we obtain that
\begin{align*}
\sum_{j=0}^1||\nablasl^j\Psi_i||_{L^4(\mathcal{S}_{u,v})}\leq\Delta_{\Psi}(\mathcal{S})<\infty,
\qquad \ i=0,1,2,3.
\end{align*}
Combining this with the inequality in the long direction shown in
Proposition~\ref{Proposition:TransportLpEstimates} we find that
\begin{align*}
  &||\delta\tau||_{L^4(\mathcal{S}_{u,v})}
  +||\bar\delta\tau||_{L^4(\mathcal{S}_{u,v})}\\
  &\quad\quad\leq
  C(I,\Delta_{\Gamma_\star})\Big(||\delta\tau||_{L^4(\mathcal{S}_{u,0})}
  +||\bar\delta\tau||_{L^4(\mathcal{S}_{u,0})}
  +\int_0^v||D\delta\tau||_{L^4(\mathcal{S}_{u,v'})}
  +||D\bar\delta\tau||_{L^4(\mathcal{S}_{u,v'})}\mathrm{d}v'\Big).
\end{align*}
Substituting the expressions for~$D\delta\tau$ and~$D\bar\delta\tau$
given by equations~\eqref{Ddeltau1}-\eqref{Ddeltau2} one concludes
that
\begin{align*}
  &||\delta\tau||_{L^4(\mathcal{S}_{u,v})}+||\bar\delta\tau||_{L^4(\mathcal{S}_{u,v})}\\
  &\quad\quad \leq C_1(I,\Delta_{\Gamma_\star},\Delta_{\Psi}(\mathcal{S}))
  +C_2(I,\Delta_{\Gamma_\star})\int_0^v(||\delta\tau||_{L^4(\mathcal{S}_{u,v'})}
  +||\bar\delta\tau||_{L^4(\mathcal{S}_{u,v'})})\mathrm{d}v'.
\end{align*}
Thus, using Gr\"onwall's inequality it follows that
\begin{align*}
  ||\delta\tau||_{L^4(\mathcal{S}_{u,v})}+||\bar\delta\tau||_{L^4(\mathcal{S}_{u,v})}
  \leq C(I,\Delta_{\Gamma_\star},\Delta_{\Psi}(\mathcal{S})) .
\end{align*}
Consequently, one has 
\begin{align*}
||\nablasl\tau||_{L^4(\mathcal{S}_{u,v})}\leq
C(I,\Delta_{\Gamma_\star},\Delta_{\Psi}(\mathcal{S}))
\end{align*}
as required. 

\smallskip
\noindent
\textbf{Estimates for $\nablasl\chi$.} From
equation~\eqref{EqDchi} one can readily compute that
\begin{align*}
  D\delta\chi=(\bar\rho-2\bar\epsilon)\delta\chi+\sigma\bar\delta\chi
  +\delta(\Psi_2+\bar\Psi_2)+\Gamma\delta\Gamma-\chi\delta
  (\epsilon+\bar\epsilon),
\end{align*}
where~$\Gamma$ represents a combination of connection coefficients
whose particular form is not essential. A similar equation
for~$D\bar\delta\chi$ can be computed. Using the same strategy used
for~$\nablasl\chi$ one concludes from the above equations
that,
\begin{align*}
  ||\delta\chi||_{L^4(\mathcal{S}_{u,v})}+||\bar\delta\chi||_{L^4(\mathcal{S}_{u,v})}
  \leq C(I,\Delta_{\Gamma_\star},\Delta_{\Psi}(\mathcal{S})).
\end{align*}
In other words, we find that
\begin{align*}
||\nablasl\chi||_{L^4(\mathcal{S}_{u,v})}\leq
C(I,\Delta_{\Gamma_\star},\Delta_{\Psi}(\mathcal{S})).
\end{align*}

\smallskip
\noindent
\textbf{Estimates for the remaining connection coefficients.} In order
to obtain equations for~$\delta\mu$ and~$\delta\lambda$, we apply
the~$\Delta$-directional derivative on both sides of
equations~\eqref{structureeq7} and~\eqref{structureeq15}. This gives,
\begin{align*}
  &\Delta\delta\mu=(\tau-\bar\alpha-\beta)(\mu^2+\lambda\bar\lambda)
  -3\mu\delta\mu-\bar\lambda\bar\delta\mu-\lambda\delta\bar\lambda
  -\bar\lambda\delta\lambda, \\
  & \Delta\delta\lambda=(\tau-\bar\alpha-\beta)(2\mu\lambda+\Psi_4)
  -3\mu\delta\lambda-\bar\lambda\bar\delta\lambda-2\lambda\delta\mu
  -\delta\Psi_4.
\end{align*}
A direct computation using
Proposition~\ref{Proposition:TransportLpEstimates} shows that there
exists an~$\varepsilon_\star$ such that
\begin{align*}
||\nablasl\{\mu,\lambda\}||_{L^4(\mathcal{S}_{u,v})}\leq3\Delta_{\Gamma_\star}
\end{align*}
if~$\varepsilon\leq\varepsilon_\star$. The details of this computation
can be found in Appendix~\ref{Appendix:Propositions8-9}. We can
estimate~$\delta\alpha$, $\delta\beta$ and~$\delta\epsilon$ by the
same method. Since, by our bootstrap
assumption~$\sup_{u,v}||\nablasl^3\tau||_{L^2(\mathcal{S}_{u,v})}<\infty$,
it follows from the Sobolev inequalities in
Corollary~\ref{Corollary:SobolevEmbedding}
that~$||\nablasl^i\tau||_{L^4(\mathcal{S}_{u,v})}$ for~$i\leq 2$ are
finite. Using this information we can estimate~$\delta\rho$
and~$\delta\sigma$ applying the~$\delta$-directional derivative to
equations~\eqref{structureeq9} and~\eqref{structureeq18}.

\smallskip
\noindent
\textbf{Concluding the argument.} From the previous estimates it
follows that we can find an~$\varepsilon_\star$ depending on~$I$,
$\Delta_{e_{\star}}$,$\Delta_{\Gamma_\star}$,
$\sup_{u,v}||\nablasl^2\tau||_{L^2(\mathcal{S}_{u,v})}$,
$\sup_{u,v}||\nablasl^3\tau||_{L^2(\mathcal{S}_{u,v})}$,
$\Delta_{\Psi}(\mathcal{S})$, and~$\Delta_{\Psi}$, such that
\begin{align*}
\sup_{u,v}||\nablasl\{\mu, \lambda, \alpha, \beta,\epsilon, \rho,
\sigma\}||_{L^4(\mathcal{S}_{u,v})}\leq3\Delta_{\Gamma_\star}.
\end{align*}
The bootstrap can hence be closed as in
Proposition~\ref{Proposition:FirstEstimateConnection}.
\end{proof}

In a similar vein, the next proposition shows how to obtain control on
the~$L^2$-norms of the NP connection coefficients and their first and
second derivatives.

\begin{proposition}[\textbf{\em control on the~$L^2$-norm of the
    connection coefficients}]
\label{Proposition:ThirdEstimateConnection}
Assume that we have a solution of the vacuum EFEs in Stewart's
  gauge in a region~$\mathcal{D}_{u,v_\bullet}^{\,t}$ with
\begin{align*}
  \sup_{u,v}||\{\mu, \lambda, \alpha, \beta, \epsilon, \rho, \sigma,
  \tau, \chi\}||_{L^\infty(S_{u,v})}&\leq \Delta_\Gamma\,,\\
    \sup_{u,v}||\nablasl\{\mu, \lambda, \alpha, \beta, \epsilon, \rho,
    \sigma\} ||_{L^4(\mathcal{S}_{u,v})}&\leq\Delta_{\Gamma},\\
    \sup_{u,v}||\nablasl^2\{\mu, \lambda, \alpha, \beta, \epsilon, \rho,
    \sigma, \tau\}||_{L^2(\mathcal{S}_{u,v})}&\leq\Delta_{\Gamma},
\end{align*}
for some positive~$\Delta_\Gamma$. Assume also
\begin{align*}
  \sup_{u,v}||\nablasl^3\tau||_{L^2(S_{u,v})}<\infty, \qquad
  \Delta_{\Psi}(\mathcal{S})<\infty, \qquad \Delta_{\Psi}<\infty,
\end{align*}
on the same domain. We have that there exists
\begin{align*}
  \varepsilon_\star=\varepsilon_\star(I, \Delta_{e_\star},\Delta_{\Gamma_\star},
  \sup_{u,v}||\nablasl^3\tau||_{L^2(\mathcal{S}_{u,v})},\Delta_{\Psi}(\mathcal{S}),
  \Delta_{\Psi}),
\end{align*}
such that when~$\varepsilon\leq\varepsilon_\star$, we have that
\begin{align*}
  & \sup_{u,v}||\nablasl^2\{\tau,\chi\}||_{L^2(\mathcal{S}_{u,v})}\leq
  C(I,\Delta_{e_\star},\Delta_{\Gamma_\star},\Delta_{\Psi}(\mathcal{S})),\\
  &\sup_{u,v}||\nablasl^2\{\mu, \lambda, \alpha, \beta,\epsilon,
  \rho, \sigma\}||_{L^2(\mathcal{S}_{u,v})}\leq3\Delta_{\Gamma_\star}.
\end{align*}
\end{proposition}

\begin{proof}
$\phantom{X}$

\smallskip
\noindent
\textbf{Basic bootstrap assumption.} Examining the above hypotheses we
first observe that both
Propositions~\ref{Proposition:FirstEstimateConnection}
and~\ref{PropositionSecondEstimateConnection} are applicable. We start
then with the following basic bootstrap assumption:
\begin{align*}
\sup_{u,v}||\nablasl^2\{\mu, \lambda, \alpha, \beta, \epsilon, \rho,
\sigma\}||_{L^2(\mathcal{S}_{u,v})}\leq4\Delta_{\Gamma_\star}.
\end{align*}

\smallskip
\noindent
\textbf{Estimates for~$||\nablasl^2\tau||_{L^2(\mathcal{S}_{u,v})}$
  and~$||\nablasl^2\chi||_{L^2(\mathcal{S}_{u,v})}$.} Starting from
equation~\eqref{Ddeltau1}, applying the~$\delta$-directional
derivative and using the commutators one obtains a~$D$-transport
equation of the form
\begin{align*}
  D\delta^2\tau=\Gamma\delta^2\tau+\Gamma\delta^2\bar\tau
  +\Gamma\bar\delta\delta\tau+\Gamma\delta\bar\delta\tau
  +\delta^2\Psi_1+\Gamma_1\delta^2\Gamma_1
  +\delta\Gamma_1\delta\Gamma_1,
\end{align*}
where~$\Gamma$ depends linearly on~$ \epsilon,\, \rho,\, \sigma,$
while~$\Gamma_1$ depends linearly on~$\tau,\, \alpha,\, \beta,\,
\epsilon,\, \rho,\, \sigma$. Similar computations lead to equations
for~$D\bar\delta\tau$ and~$D\delta\bar\delta\tau$. The
term~$\delta\Gamma_1\delta\Gamma_1$ is dealt with using the H\"older
inequality to obtain
\begin{align*}
  ||\delta\Gamma_1\delta\Gamma_1||_{L^2(\mathcal{S}_{u,v})}
  \leq||\delta\Gamma_1||_{L^4(\mathcal{S}_{u,v})}
  ||\delta\Gamma_1||_{L^4(\mathcal{S}_{u,v})}.
\end{align*}
Using Proposition~\ref{PropositionSecondEstimateConnection}, it
follows then that the left-hand side of the inequality is finite.

Now, the inequality in the long direction of
Proposition~\ref{Proposition:TransportLpEstimates} and the equation
for~$D\delta\tau$ show that,
\begin{align*}
  &||\delta^2\tau||_{L^2(\mathcal{S}_{u,v})}\leq
  C(I,\Delta_{\Gamma_\star})\left(||\delta^2\tau||_{L^2(\mathcal{S}_{u,0})}
  +\int_0^v||D\delta^2\tau||_{L^2(\mathcal{S}_{u,v'})}\mathrm{d}v' \right), \\
  &\phantom{||\delta^2\tau||_{L^2(\mathcal{S}_{u,v})}}
  \leq C(I,\Delta_{e_\star},\Delta_{\Gamma_\star},
  \Delta_{\Psi}(\mathcal{S}))+C(I,\Delta_{e_\star},\Delta_{\Gamma_\star})
  \int_0^v||\nablasl^2\tau||_{L^2(\mathcal{S}_{u,v'})}\mathrm{d}v'.
\end{align*}
Similar estimates can be obtained for~$\bar\delta^2\tau$,
$\delta\bar\delta\tau$ and~$\bar\delta\delta\tau$. 

Recalling the result in Corollary~\ref{Corrollary:Area} that the area
of~$\mathcal{S}_{u,v}$ is bounded one can estimate the
norm~$||\delta\tau||_{L^2(\mathcal{S}_{u,v})}$ by observing that
\begin{align*}
  ||\delta\tau||_{L^2(\mathcal{S}_{u,v})}\leq
  C(\Delta_{e_\star},\Delta_{\Gamma_\star})
  ||\delta\tau||_{L^4(\mathcal{S}_{u,v})}.
\end{align*}
Hence, using Proposition~\ref{PropositionSecondEstimateConnection} it
follows that~$||\delta\tau||_{L^2(\mathcal{S}_{u,v})}$ is also
finite. Now, from inequality~\eqref{nabla2f} we then obtain that
\begin{align*}
||\nablasl^2\tau||_{L^2(\mathcal{S}_{u,v})}\leq
C(I,\Delta_{e_\star},\Delta_{\Gamma_\star},
\Delta_{\Psi}(\mathcal{S}))
+C(I,\Delta_{e_\star},\Delta_{\Gamma_\star})\int_0^v
||\nablasl^2\tau||_{L^2(\mathcal{S}_{u,v'})}\mathrm{d}v'.
\end{align*}
so that using Gr\"onwall's inequality one concludes that
\begin{align*}
||\nablasl^2\tau||_{L^2(\mathcal{S}_{u,v})}\leq
C(I,\Delta_{e_\star},\Delta_{\Gamma_\star},\Delta_{\Psi}(\mathcal{S})).
\end{align*}

\smallskip
\noindent
\textbf{Estimates for~$||\nablasl^2\chi||_{L^2(\mathcal{S}_{u,v})}$.}
An analysis analogous to that for~$\tau$, readily shows
that~$||\nablasl^2\chi||_{L^2(\mathcal{S}_{u,v})}$ is bounded.

\smallskip
\noindent
\textbf{Estimates for the the remaining spin connection coefficients.}
The remaining connection coefficients can be estimated using the same
ideas as in Proposition~\ref{Proposition:FirstEstimateConnection}
---namely, we first compute equations for~$\Delta\delta^2\Gamma$
and~$\Delta\bar\delta\delta\Gamma$ using the NP Ricci identities and
the commutators for covariant directional derivatives.  In a second
step we make use of the short direction inequality of
Proposition~\ref{Proposition:TransportLpEstimates}. It then follows
that one can choose~$\varepsilon$ small enough so that,
\begin{align*}
\sup_{u,v}||\nablasl^2\{\mu,\lambda,\alpha,\beta,\epsilon,\rho,\sigma\}
||_{L^2(\mathcal{S}_{u,v})}\leq3\Delta_{\Gamma_\star}\,,
\end{align*}
for,
\begin{align*}
  \varepsilon\leq\varepsilon_\star(I, \Delta_{e_\star},\Delta_{\Gamma_\star},
  \sup_{u,v}||\nablasl^3\tau||_{L^2(\mathcal{S}_{u,v})},
  \Delta_{\Psi}(\mathcal{S}),\Delta_{\Psi})\,.
\end{align*}
Details of the generic calculations involved in these last steps are
discussed in Appendix~\ref{nablaf}.
\end{proof}

\subsection{A first estimate for the curvature}

Having obtained estimates for the NP spin connection coefficients, we
are now in the position to obtain a first estimate for the
curvature. The proposition of this section provides for bounds the
components of the Weyl tensor of the spheres~$\mathcal{S}_{u,v}$
assuming, as a bootstrap, their boundedness on the null hypersurfaces
and boundedness on~$\tau$ and its derivatives.

\begin{proposition}[\textbf{\em basic control of the curvature}]
  \label{Proposition:FirstEstimateCurvature}
Assume that we are given a solution to the vacuum EFEs in
  Stewart's gauge satisfying the assumptions of
  Proposition~\ref{Proposition:ThirdEstimateConnection}. Then there
  exists
\begin{align*}
\varepsilon_\star=\varepsilon_\star
(\Delta_{e_\star},\Delta_{\Gamma_\star},\Delta_{\Psi_\star},
\sup_{u,v}||\nablasl^3\tau||_{L^2(\mathcal{S}_{u,v})},\Delta_{\Psi})
\end{align*}
such that for~$\varepsilon\leq\varepsilon_\star$, one has 
\begin{align*}
\Delta_{\Psi}(\mathcal{S})\leq C(\Delta_{\Psi_\star})\,,
\end{align*}
on~$\mathcal{D}_{u,v_\bullet}^{\,t}$.
\end{proposition}

\begin{proof}
$\phantom{}$

\smallskip
\noindent
\textbf{Boostrap assumption.} In this proof we start with the
following bootstrap assumption:
\begin{align*}
\sup_{u,v}||\nablasl^i\{\Psi_0,\Psi_1,\Psi_2,\Psi_3\}
||_{L^2(\mathcal{S}_{u,v})}\leq 4\Delta_{\Psi_\star},\qquad i=0,...,2,
\end{align*}
which we then aim to improve.

\smallskip
\noindent
\textbf{$L^2$-norm of the
  components~$\{\Psi_0,\Psi_1,\Psi_2,\Psi_3\}$.} Estimates for
the~$L^2$-norms of the components~$\{\Psi_0,\Psi_1,\Psi_2,\Psi_3\}$
can be obtained from the~$\Delta$-Bianchi identity
equations~\eqref{structureeq2}, \eqref{structureeq8},
\eqref{structureeq6} and~\eqref{structureeq4} which are then
integrated along the short direction. As an example of the procedure
we consider here the coefficient~$\Psi_2$. From
Proposition~\ref{Proposition:TransportLpEstimates} it follows that
\begin{align*}
  ||\Psi_2||_{L^2(\mathcal{S}_{u,v})}&\leq2\left(||\Psi_2||_{L^2(\mathcal{S}_{0,v})}
  +C(\Delta_{e_\star},\Delta_{\Gamma_\star})
  \int_0^u||\Delta\Psi_2||_{L^2(\mathcal{S}_{u',v})}
  \mathrm{d}u' \right) \\
  & \leq2\Big(\Delta_{\Psi_\star}+C(\Delta_{e_\star},\Delta_{\Gamma_\star})
  \int_0^u||\nablasl\Psi_3||_{L^2(\mathcal{S}_{u',v})}
  +||3\mu\Psi_2||_{L^2(\mathcal{S}_{u',v})} \\
  &\quad\quad
  +||2(\beta-\tau)\Psi_3||_{L^2(\mathcal{S}_{u',v})}
  +||\sigma\Psi_4||_{L^2(\mathcal{S}_{u',v})} \mathrm{d}u' \Big)  \\
  &\leq2\left(\Delta_{\Psi_\star}+C(\Delta_{e_\star},\Delta_{\Gamma_\star},
  \Delta_{\Psi_\star})\varepsilon  +C(\Delta_{e_\star},\Delta_{\Gamma_\star})
  \Delta_{\Psi}\varepsilon^{1/2} \right.\\
  &\quad\quad \left.
+C(\Delta_{e_\star},\Delta_{\Gamma_\star})
  ||\Psi_4||_{L^2(\mathcal{N}'_v(0,u))}\varepsilon^{1/2}
 \right)\\
  &\leq2\Delta_{\Psi_\star}
  +C(\Delta_{e_\star},\Delta_{\Gamma_\star},\Delta_{\Psi_\star})\varepsilon
  +C(\Delta_{e_\star},\Delta_{\Gamma_\star},\Delta_{\Psi_\star})
  \Delta_{\Psi}\varepsilon^{1/2},
\end{align*}
In passing from the second to the third inequality we have used that
the term
\begin{align*}
\int_0^u||\nablasl\Psi_3||_{L^2(\mathcal{S}_{u',v})}\mathrm{du'}
\end{align*}
is, in fact, an statement on the light cone and, hence, it is
controlled by the definition of~$\Delta_{\Psi}$. Moreover, we have
also used H\"older's inequality in the form
\begin{align*}
\int_0^u||\Psi_4||_{L^2(\mathcal{S}_{u',v})}\mathrm{d}u'\leq
C\varepsilon^{1/2}||\Psi_4||_{L^2(\mathcal{N}_{v'}(0,u))}.
\end{align*}
The analysis for the coefficients~$\Psi_0, \Psi_1, \Psi_3$ is
similar. Consequently, we can find~$\varepsilon_\star$ depending on
the initial data,~$\Delta_{\Psi}$ and~$I$ such that
for~$\varepsilon\leq\varepsilon_\star$, we have
\begin{align*}
\sup_{u,v}||\{\Psi_0,\Psi_1,\Psi_2,\Psi_3\}
||_{L^2(\mathcal{S}_{u,v})}\leq3\Delta_{\Psi_\star}.
\end{align*}

\smallskip
\noindent
\textbf{Estimates
  for~$||\nablasl\{\Psi_0,\Psi_1,\Psi_2,\Psi_3\}||_{L^2(\mathcal{S}_{u,v})}$.}
Again, we focus our discussion on the analysis of the
coefficient~$\Psi_2$. From
Proposition~\ref{Proposition:TransportLpEstimates} we find that
\begin{align*}
  ||\nablasl\Psi_2||_{L^2(\mathcal{S}_{u,v})}&\leq2
  \left(||\nablasl\Psi_2||_{L^2(\mathcal{S}_{0,v})}
  +C(\Delta_{e_\star},\Delta_{\Gamma_\star})\int_0^u
  \left(\int_{\mathcal{S}_{u',v}}\Delta\left\langle\nablasl\Psi_2,\nablasl\Psi_2
  \right\rangle_{\bmsigma}\right)^{1/2}\mathrm{d}u' \right) \\
  &\leq 2\Delta_{\Psi_{\star}}+C(\Delta_{e_\star},\Delta_{\Gamma_\star})
  \int_0^u\left(\int_{\mathcal{S}_{u',v}}|\nablasl\Psi_2|
  (|\Delta\delta\Psi_2|+|\Delta\bar\delta\Psi_2|) \right)^{1/2}\mathrm{d}u' .
\end{align*}
Now, using the expression for~$\Delta\delta\Psi_2$
and~$\Delta\bar\delta\Psi_2$ obtained from using the commutators on
the~$\Delta$-Bianchi equation for~$\Psi_2$, and schematically denoting
arbitrary connection coefficients by~$\Gamma$, one obtains that
\begin{align}
  & \int_0^u\left(\int_{\mathcal{S}_{u',v}}
  |\nablasl\Psi_2|(|\Delta\delta\Psi_2|+|\Delta\bar\delta\Psi_2|)\right)^{1/2}
  \mathrm{d}u' \nonumber \\
  &\quad \leq \int_0^u\Big(\int_{\mathcal{S}_{u',v}}
  |\nablasl\Psi_2||\Gamma|^2|\Psi_{2,3}|\Big)^{1/2}  \mathrm{d}u'
  +\int_0^u\Big(\int_{\mathcal{S}_{u',v}}
  |\nablasl\Psi_2||\Gamma|^2|\Psi_4|\Big)^{1/2}\mathrm{d}u'\nonumber\\
  &\quad+\int_0^u\Big(\int_{\mathcal{S}_{u',v}}
  |\nablasl\Psi_2||\Gamma||\nablasl\Psi_{2,3}|\Big)^{1/2}\mathrm{d}u'
  +\int_0^u\Big(\int_{\mathcal{S}_{u',v}}|\nablasl\Psi_2||\Gamma||
  \nablasl\Psi_4|\Big)^{1/2}\mathrm{d}u'\nonumber \\
  &\quad+\int_0^u\Big(\int_{\mathcal{S}_{u',v}}
  |\nablasl\Psi_2||\nablasl\Gamma||\Psi_{2,3}|\Big)^{1/2}\mathrm{d}u'
  +\int_0^u\Big(\int_{\mathcal{S}_{u',v}}|\nablasl\Psi_2||
  \nablasl\Gamma||\Psi_4|\Big)^{1/2}\mathrm{d}u'\nonumber\\
  &\quad+\int_0^u\Big(\int_{\mathcal{S}_{u',v}}|\nablasl
  \Psi_2||\nablasl^2\Psi_3|\Big)^{1/2}\mathrm{d}u'.
  \label{Proposition10IntermediateInequality}
\end{align}
In the first and third terms of the right-hand side or the above
inequality we can separate the~$L^{\infty}$-norm of the connection
coefficients. Thus, using the bootstrap assumption with
Proposition~\ref{Proposition:FirstEstimateConnection}, we find that
\begin{align*}
  &\int_0^u\Big(\int_{\mathcal{S}_{u',v}}|\nablasl\Psi_2||\Gamma|^{2-i}|
  \nablasl^i\Psi_{2,3}|\Big)^{1/2}\mathrm{d}u'
  \leq C(I,\Delta_{e_\star},\Delta_{\Gamma_\star},\Delta_{\Psi_\star})
  \int_0^u||\nablasl\Psi_2||^{1/2}_{L^2(\mathcal{S}_{u',v})}
  ||\nablasl^i\Psi_{2,3}||^{1/2}_{L^2(\mathcal{S}_{u',v})}
  \mathrm{d}u', 
\end{align*}
for~$i=0,1$. Accordingly, using the bootstrap assumption once again,
we conclude that,
\begin{align*}
  \quad \int_0^u\Big(\int_{\mathcal{S}_{u',v}}|\nablasl^i\Psi_2
  ||\Gamma|^{2-i}|\nablasl\Psi_{2,3}|\Big)^{1/2}\mathrm{d}u'
  \leq C(I,\Delta_{e_\star},\Delta_{\Gamma_\star},\Delta_{\Psi_\star}) \varepsilon,
\end{align*}
for~$i=0,1$. The second and fourth term in the right-hand side of
inequality~\eqref{Proposition10IntermediateInequality} can be handled
in an analogous manner. Since we do not have control on the
the~$L^2(\mathcal{S}_{u,v})$ norm of~$\Psi_4$, we transform
the~$L^2(\mathcal{S}_{u,v})$ norm to a norm over the light cone. More
precisely, one has that using H\"older's inequality
\begin{align*}
  \int_0^u
  \Big(\int_{\mathcal{S}_{u',v}}|\nablasl\Psi_2||\Gamma|^{2-i}
  |\nablasl^i\Psi_4|\Big)^{1/2}\mathrm{d}u'
  &\leq\int_0^u||\nablasl\Psi_2||^{1/2}_{L^2(\mathcal{S}_{u',v})}
  ||\nablasl^i\Psi_4||^{1/2}_{L^2(\mathcal{S}_{u',v})}\mathrm{d}u'\\
  &\leq
  C(\Delta_{\Psi_\star})
  ||\nablasl^i\Psi_4||^{1/2}_{L^2(\mathcal{N}'_v(0,u))}
  \varepsilon^{3/4}, \ i=0,1.
\end{align*}
Hence, we conclude that
\begin{align*}
  \int_0^u\Big(\int_{\mathcal{S}_{u',v}}|\nablasl\Psi_2|
  |\Gamma|^2|\Psi_4|\Big)^{1/2}\mathrm{d}u',
  \quad \int_0^u\Big(\int_{\mathcal{S}_{u',v}}|\nablasl\Psi_2
  ||\Gamma||\nablasl\Psi_4|\Big)^{1/2}\mathrm{d}u' \leq
  C(\Delta_{\Psi_\star},\Delta_{\Psi})\varepsilon^{3/4}.
\end{align*}
Now, for the fifth term in
inequality~\eqref{Proposition10IntermediateInequality} one has that
\begin{align*}
  \int_0^u\Big(\int_{\mathcal{S}_{u',v}}|\nablasl\Psi_2|
  |\nablasl\Gamma||\Psi_{2,3}|\Big)^{1/2}\mathrm{d}u'
  \leq \int_0^u\left(||\Psi_{2,3}||_{L^{\infty}(\mathcal{S}_{u,v})}
  ||\nablasl\Psi_2||_{L^2(\mathcal{S}_{u,v})}
  ||\nablasl\Gamma||_{L^2(\mathcal{S}_{u,v})}
  \right)^{1/2}\mathrm{d}u',
\end{align*}
where the first term in the integral in the right-hand side can be
controlled by the bootstrap assumption and Sobolev embedding
(Corollary~\ref{Corollary:SobolevEmbedding}). The third term can be
controlled by the~$L^4(\mathcal{S}_{u,v})$ norm as given by
Proposition~\ref{PropositionSecondEstimateConnection}, again in
combination with the bootstrap assumption. One then concludes that,
\begin{align*}
  \int_0^u\Big(\int_{\mathcal{S}_{u',v}}|\nablasl\Psi_2
  ||\nablasl\Gamma||\Psi_{2,3}|\Big)^{1/2}\mathrm{d}u'
  &\leq C(I,\Delta_{e_\star},\Delta_{\Gamma_\star},\Delta_{\Psi_\star})
  \sum_{i=0}^2\int_0^u
  ||\nablasl^i\Psi_{2,3}||^{1/2}_{L^2(\mathcal{S}_{u',v})}
  \mathrm{d}u'\\
  &\leq C(I,\Delta_{e_\star},\Delta_{\Gamma_\star},\Delta_{\Psi_\star},\Delta_{\Psi})
  \varepsilon^{3/4}.
\end{align*}
The sixth term in
inequality~\eqref{Proposition10IntermediateInequality} can also be
dealt with by transforming the norms of the coefficients of the Weyl
tensor on~$\mathcal{S}_{u,v}$ to norms on the light cone. More
precisely, one has that
\begin{align*}
  & \int_0^u\Big(\int_{\mathcal{S}_{u',v}}|\nablasl\Psi_2||
  \nablasl\Gamma||\Psi_4|\Big)^{1/2}\mathrm{d}u'
  \leq \int_0^u\left(||\Psi_4||_{L^{\infty}(\mathcal{S}_{u,v})}
  ||\nablasl\Psi_2||_{L^2(\mathcal{S}_{u,v})}||\nablasl\Gamma||_{L^2(\mathcal{S}_{u,v})}
  \right)^{1/2}\mathrm{d}u' \\
  &\qquad\qquad\leq C(I,\Delta_{e_\star},\Delta_{\Gamma_\star},\Delta_{\Psi_\star})
  \left(\int_0^u\sum_{i=0}^2
  ||\nablasl^i\Psi_4||_{L^2(\mathcal{S}_{u,v})}\mathrm{d}u'\right)^{1/2}\\
  &\qquad\qquad\leq C(I,\Delta_{e_\star},\Delta_{\Gamma_\star},\Delta_{\Psi_\star})
  \Big(\sum_{i=0}^2||\nablasl^i\Psi_4||_{L^2(\mathcal{N}'_v(0,u))}\Big)
  \leq C(I,\Delta_{e_\star},\Delta_{\Gamma_\star},
  \Delta_{\Psi_\star},\Delta_{\Psi})\varepsilon^{3/4}. 
\end{align*}
Finally, the last integral in the right-hand side of
inequality~\eqref{Proposition10IntermediateInequality} can be
separated into two~$L^2$-norms. The estimate of~$\nablasl^2\Psi_3$
can, in turn, be transformed to an estimate on the light cone and,
hence, it can be controlled by the definition of~$\Delta_{\Psi}$.

Collecting all the estimates for the various terms in
inequality~\eqref{Proposition10IntermediateInequality} we conclude
that,
\begin{align*}
||\nablasl\Psi_2||_{L^2(\mathcal{S}_{u,v})}\leq 2\Delta_{\Psi_\star}
+C(I,\Delta_{e_\star},\Delta_{\Gamma_\star},\Delta_{\Psi_\star})\varepsilon
+C(I,\Delta_{e_\star},\Delta_{\Gamma_\star},\Delta_{\Psi_\star},\Delta_{\Psi})
\Delta_{\Psi}\varepsilon^{3/4}.
\end{align*}
The latter inequality implies that we can improve the
bootstrap assumption by choosing~$\varepsilon$ small enough. A similar
strategy allows us to
estimate~$\nablasl\{\Psi_0,\Psi_1,\Psi_3\}$. Therefore we have that
\begin{align*}
  \sup_{u,v}||\nablasl\{\Psi_0,\Psi_1,\Psi_2,\Psi_3\}||_{L^2(S_{u,v})}
  \leq3\Delta_{\Psi_\star}.
\end{align*}

\smallskip
\noindent
\textbf{Estimates for~$||\nablasl^2\{\Psi_0,\Psi_1,\Psi_2,\Psi_3\}
  ||_{L^2(\mathcal{S}_{u,v})}$.}  As before, we focus the discussion
on~$||\nablasl^2\Psi_2||_{L^2(\mathcal{S}_{u,v})}$. The estimate along
the short direction in
Proposition~\ref{Proposition:TransportLpEstimates} shows that
\begin{align}
  ||\nablasl^2\Psi_2||_{L^2(\mathcal{S}_{u,v})}&\leq
  2\left(||\nablasl^2\Psi_2||_{L^2(\mathcal{S}_{0,v})}
  +C(\Delta_{e_\star},\Delta_{\Gamma_\star})\int_0^u
  \left(\int_{\mathcal{S}_{u',v}}\Delta
  \left\langle\nablasl^2\Psi_2,\nablasl^2\Psi_2
  \right\rangle_{\bmsigma}\right)^{1/2} \mathrm{d}u'\right) \nonumber\\
  &\leq2\Delta_{\Psi_\star}+C(\Delta_{e_\star},\Delta_{\Gamma_\star})
  \int_0^u\left(\int_{\mathcal{S}_{u',v}}|\nablasl^2\Psi_2
  |(|\Delta T|) \right)^{1/2}\mathrm{d}u',\label{Proposition10Intermediate2}
\end{align}
where~$T$ denotes an expression involving products of connection
coefficients, their derivatives and components of the Weyl tensor and
their derivatives. In particular, one has that
\begin{align*}
  &\int_{\mathcal{S}_{u',v}}|\nablasl^2\Psi_2|(|\Delta T|) \\ &
  \leq\int_{S_{u',v}}|\nablasl^2\Psi_2||\Psi\nablasl^2
  \Gamma+\Gamma\nablasl^2\Psi+\nablasl\Psi\nablasl\Gamma+\Gamma^2\nablasl\Psi
  +\Gamma\Psi\nablasl\Psi+\Gamma^3\Psi+\Psi_3\nablasl\Psi_2+\nablasl^3\Psi_3|.
\end{align*}
We can then proceed with a strategy similar to that used in the
analysis of the estimates for the first order derivatives of the
components of the Weyl tensor. In particular, we use H\"older's
inequality to split products and then apply the Sobolev embedding
theorem as necessary. The estimates on the sphere for the
terms~$\nablasl^i\Psi_4$ and~$\nablasl^3\Psi_3$ are transformed into
estimates on the light cone. Hence the integral on the right-hand-side
of inequality~\eqref{Proposition10Intermediate2} can be made as small
as necessary by choosing a suitable~$\varepsilon$. Ultimately, we
conclude that
\begin{align*}
\sup_{u,v}||\nablasl^2\{\Psi_0,\Psi_1,\Psi_2,\Psi_3\}
||_{L^2(\mathcal{S}_{u,v})}\leq3\Delta_{\Psi_\star}.
\end{align*}

\smallskip
\noindent
\textbf{Concluding the argument.} Collecting all the estimates in the
previous steps one obtains the statement
\begin{align*}
\sup_{u,v}||\nablasl^i\{\Psi_0,\Psi_1,\Psi_2,\Psi_3\}
||_{L^2(\mathcal{S}_{u,v})}\leq3\Delta_{\Psi_\star},\qquad i=0,...,2,
\end{align*}
which improves the starting bootstrap assumption.
\end{proof}

Applying the standard embedding of~$L^p$ into~$L^q$ for~$p\leq q$, we
can summarise the results of
Propositions~\ref{Proposition:FirstEstimateConnection},
\ref{PropositionSecondEstimateConnection},
\ref{Proposition:ThirdEstimateConnection}
and~\ref{Proposition:FirstEstimateCurvature} in the following
proposition:

\begin{proposition}[\textbf{\em summary of the basic estimates for the
 NP quantities}]
\label{Proposition:SummaryBasicEstmatesConnectionCurvature}
Suppose we are given a solution to the vacuum EFE's in Stewart's gauge
emanating from data for the CIVP as prepared in
Lemma~\ref{Lemma:FreeDataCIVP}, satisfying
  \begin{align*}
  \sup_{u,v}||\{\mu, \lambda, \alpha, \beta, \epsilon, \rho, \sigma,
  \tau, \chi\}||_{L^\infty(\mathcal{S}_{u,v})}&< \infty\,,\quad
    \sup_{u,v}||\nablasl\{\mu, \lambda, \alpha, \beta, \epsilon, \rho,
    \sigma\} ||_{L^4(\mathcal{S}_{u,v})}<\infty\,,\\
    \sup_{u,v}||\nablasl^2\{\mu, \lambda, \alpha, \beta, \epsilon, \rho,
    \sigma, \tau\}||_{L^2(\mathcal{S}_{u,v})}&<\infty\,,\quad
  \sup_{u,v}||\nablasl^3\tau||_{L^2(\mathcal{S}_{u,v})}<\infty\,, \quad
  \Delta_{\Psi}(\mathcal{S})<\infty\,, \quad \Delta_{\Psi}<\infty\,,
\end{align*}
on some truncated causal
diamond~$\mathcal{D}_{u,v_\bullet}^{\,t}$. Then there exists,
\begin{align*}
  \varepsilon_\star
  =\varepsilon_\star(I,\Delta_{e_\star},\Delta_{\Gamma_\star},\Delta_{\Psi_\star},
  \sup_{u,v}||\nablasl^3\tau||_{L^2(\mathcal{S}_{u,v})},\Delta_{\Psi})\,,
\end{align*}
such that for~$\varepsilon\leq\varepsilon_\star$, we have
\begin{align*}
  &||\Gamma||_{L^{\infty}(\mathcal{S}_{u,v})}\leq
  C(I,\Delta_{e_\star},\Delta_{\Gamma_\star},\Delta_{\Psi_\star}), &\qquad
  &\sum_{i=0}^1||\nablasl^i\Gamma||_{L^4(\mathcal{S}_{u,v})}\leq
  C(I,\Delta_{e_\star},\Delta_{\Gamma_\star},\Delta_{\Psi_\star}), \\
  &\sum_{i=0}^2||\nablasl^i\Gamma||_{L^2(\mathcal{S}_{u,v})}\leq
  C(I,\Delta_{e_\star},\Delta_{\Gamma_\star},\Delta_{\Psi_\star}), &\qquad
  &\Delta_{\Psi}(\mathcal{S})\leq
  C(\Delta_{\Psi_\star}),
\end{align*}
on~$\mathcal{D}_{u,v_\bullet}^{\,t}$, with~$\Gamma$ standing for an
arbitrary connection coefficient.
\end{proposition}

\subsection{Estimates on the third derivatives of connection coefficients}

We are now in the position to obtain estimates for the NP spin
connection coefficients which only require assumptions on the
curvature on the light cone. More precisely, one has the following:

\begin{proposition}[\textbf{\em further control on the $L^2$-norm of
    the connection coefficients}]
\label{Proposition:ImprovedEstimates}
Assume, as in the previous proposition, that we are given a solution
to the vacuum EFE's in Stewart's gauge emanating from data for the
CIVP as prepared in Lemma~\ref{Lemma:FreeDataCIVP}. Suppose that,
  \begin{align*}
  \sup_{u,v}||\{\mu, \lambda, \alpha, \beta, \epsilon, \rho, \sigma,
  \tau, \chi\}||_{L^\infty(\mathcal{S}_{u,v})}&< \infty\,,\quad
    \sup_{u,v}||\nablasl\{\mu, \lambda, \alpha, \beta, \epsilon, \rho,
    \sigma\} ||_{L^4(\mathcal{S}_{u,v})}<\infty\,,\\
    \sup_{u,v}||\nablasl^2\{\mu, \lambda, \alpha, \beta, \epsilon, \rho,
    \sigma, \tau\}||_{L^2(\mathcal{S}_{u,v})}&<\infty\,,\quad
    \Delta_{\Psi}(\mathcal{S})<\infty\,, \quad \Delta_{\Psi}<\infty\,,
\end{align*}
and furthermore that,
\begin{align*}
    \sup_{u,v}||\nablasl^3\{\mu,\lambda,\alpha,\beta,\epsilon,\tau\}
    ||_{L^2(\mathcal{S}_{u,v})}<\infty\,,
\end{align*}
on~$\mathcal{D}_{u,v_\bullet}^{\,t}$. Then there
exists~$\varepsilon_\star=\varepsilon_\star(I,\Delta_{e_\star},\Delta_{\Gamma_\star},
\Delta_{\Psi_\star},\Delta_{\Psi})$ such that
for~$\varepsilon\leq\varepsilon_\star$, we have
\begin{align*}
  &\sup_{u,v}||\nablasl^3\{\mu,\lambda,\alpha,\beta,\epsilon\}
  ||_{L^2(\mathcal{S}_{u,v})}\leq3\Delta_{\Gamma_\star}, \\
  &\sup_{u,v}||\nablasl^3\{\rho,\sigma\}||_{L^2(\mathcal{S}_{u,v})}
  \leq C(I,\Delta_{e_\star},\Delta_{\Gamma_\star},\Delta_{\Psi_\star}),\\
  &\sup_{u,v}||\nablasl^3\{\tau,\chi\}||_{L^2(\mathcal{S}_{u,v})}
  \leq C(I,\Delta_{e_\star}, \Delta_{\Gamma_\star},\Delta_{\Psi_\star},\Delta_{\Psi}) .
\end{align*}
\end{proposition}

\begin{proof}
$\phantom{}$

\smallskip
\noindent
\textbf{Bootstrap assumption.} In order to start the proof we place
bootstrap assumptions on~$\mu,\lambda,\alpha,\beta$ and~$\epsilon$,
and name the bound on~$\tau$ as follows,
\begin{align*}
  \sup_{u,v}||\nablasl^3\{\mu,\lambda,\alpha,\beta,\epsilon\}
  ||_{L^2(\mathcal{S}_{u,v})}\leq4\Delta_{\Gamma_\star}, \qquad
  \sup_{u,v}||\nablasl^3\tau||_{L^2(\mathcal{S}_{u,v})}\leq\Delta_\tau.
\end{align*}

\smallskip
\noindent
\textbf{Estimates for~$\rho$ and~$\sigma$.} We first estimate the spin
connection coefficients~$\rho$ and~$\sigma$ using the long direction
transport equations~\eqref{structureeq13} and~\eqref{structureeq6} as
this allows to avoid higher derivatives on the sphere that arise in
the short direction equations.  Using the expression
for~$||\nablasl^3f||_{L^2(S_{u,v})}$ for an arbitrary scalar~$f$ given
in Appendix~\ref{nablaf}, we will discuss four typical terms. The
first is~$\delta^3\rho$. Making use of the commutators of directional
covariant derivatives, we can compute the long direction derivative of
any third derivatives of~$\rho$ on the sphere ---for example, one has
that,
\begin{align*}
  D\delta^3\rho&=\Gamma^5+\Gamma^3\delta\Gamma+\Gamma(\delta\Gamma)^2
  +\Gamma^2\delta^2\Gamma+\delta\Gamma\delta^2\Gamma
  +\rho\delta^3(\epsilon+\bar\epsilon) \\
  & \quad+(4\epsilon-2\bar\epsilon+5\rho)\delta^3\rho
  +\sigma\delta^3\bar\sigma
  +\bar\sigma\delta^3\sigma+\sigma\delta^2\bar\delta\rho,
\end{align*}
where here~$\Gamma$ represents linear combinations of the
coefficients~$\epsilon$, $\rho$ and~$\sigma$, whose precise form is
not crucial for the discussion. The~$L^2$-norm of the
term~$\delta\Gamma\delta^2\Gamma$ can be split as
\begin{align*}
  ||\delta\Gamma\delta^2\Gamma||_{L^2(\mathcal{S}_{u,v})}
  \leq||\nablasl\Gamma||_{L^4(\mathcal{S}_{u,v})}
  ||\nablasl^2\Gamma||_{L^4(\mathcal{S}_{u,v})}.
\end{align*}
The first term on the right-hand side of the inequality can be
controlled using the results of
Proposition~\ref{PropositionSecondEstimateConnection}. The second term
can be controlled using the Sobolev inequality,
\begin{align*}
||\nablasl^2\Gamma||_{L^4(\mathcal{S}_{u,v})}\leq
C(\Delta_{e_\star})\left(||\nablasl^2\Gamma||_{L^2(\mathcal{S}_{u,v})}
+||\nablasl^3\Gamma||_{L^2(\mathcal{S}_{u,v})} \right).
\end{align*}
Proceeding in a similar way with the other terms in the equation
for~$D\delta^3\rho$ and the using the long direction inequality in
Proposition~\ref{Proposition:TransportLpEstimates} leads to
\begin{align*}
||\delta^3\rho||_{L^2(\mathcal{S}_{u,v})}\leq
C(I,\Delta_{e_\star},\Delta_{\Gamma_\star},\Delta_{\Psi_\star})
+C(I,\Delta_{e_\star},\Delta_{\Gamma_\star},\Delta_{\Psi_\star})\int_0^v
\left(||\nablasl^3\rho||_{L^2(\mathcal{S}_{u,v'})}
+||\nablasl^3\sigma||_{L^2(\mathcal{S}_{u,v'})}\right)
\mathrm{d}v'.
\end{align*}
The second representative term in the expansion
of~$||\nablasl^3\rho||_{L^2(\mathcal{S}_{u,v})}$
is~$||\varpi\delta^2\rho||_{L^2(S_{\mathcal{S}_{u,v})}}$ (recall
that~$\varpi\equiv \beta-\bar\alpha$). One has
\begin{align*}
  D(\varpi\delta^2\rho)&=D\varpi(\delta^2\rho)+\varpi D\delta^2\rho \\
  & =(\Psi_1+\Gamma^2+\delta\epsilon
  -\delta\bar\epsilon)\delta^2\rho+\Gamma^5
  +\Gamma^3\delta\Gamma+\varpi(\delta\Gamma)^2+\Gamma^2\delta^2\Gamma,
\end{align*}
 from which we can conclude that
\begin{align*}
  ||\varpi\delta^2\rho||_{L^2(\mathcal{S}_{u,v})}\leq
  C(I,\Delta_{e_\star},\Delta_{\Gamma_\star},\Delta_{\Psi_\star})
  +C(I,\Delta_{e_\star},\Delta_{\Gamma_\star},\Delta_{\Psi_\star})
  \int_0^v||\nablasl^3\rho||_{L^2(\mathcal{S}_{u,v'})}\mathrm{d}v',
\end{align*}
by Sobolev embedding as before. The third representative term
is~$||\delta\varpi\delta\rho||_{L^2(\mathcal{S}_{u,v})}$ for which we
have
\begin{align*}
  D(\delta\varpi\delta\rho)=-\Psi_1\bar\pi\delta\rho+\Gamma^3\delta\rho
  +\delta\Psi_1\delta\rho+\Gamma(\delta\Gamma)^2
  +\delta^2(\epsilon-\bar\epsilon)\delta\rho,
\end{align*}
so that
\begin{align*}
  ||\delta\varpi\delta\rho||_{L^2(\mathcal{S}_{u,v})}\leq
  C(I,\Delta_{e_\star},\Delta_{\Gamma_\star},\Delta_{\Psi_\star})
  +C(I,\Delta_{e_\star},\Delta_{\Gamma_\star},\Delta_{\Psi_\star})
  \int_0^v||\nablasl^3\rho||_{L^2(\mathcal{S}_{u,v'})}\mathrm{d}v'.
\end{align*}
The fourth representative term is~$\varpi^2\delta\rho$ for which we
can compute
\begin{align*}
  D(\varpi^2\delta\rho)=2\varpi\Psi_1\delta\rho+\Gamma^3\delta\Gamma
  +\Gamma(\delta\Gamma)^2+\Gamma^5.
\end{align*}
Consequently, one finds that
\begin{align*}
  ||\varpi^2\delta\rho||_{L^2(\mathcal{S}_{u,v})}\leq
  C(I,\Delta_{e_\star},\Delta_{\Gamma_\star},\Delta_{\Psi_\star}).
\end{align*}
Combining all the expressions arising in the expansion
of~$\nablasl^3\rho$ one then concludes,
\begin{align*}
  ||\nablasl^3\rho||_{L^2(\mathcal{S}_{u,v})}\leq
  C(I,\Delta_{e_\star},\Delta_{\Gamma_\star},\Delta_{\Psi_\star})
  +C(I,\Delta_{e_\star},\Delta_{\Gamma_\star},\Delta_{\Psi_\star})
  \int_0^v\left(||\nablasl^3\rho||_{L^2(\mathcal{S}_{u,v'})}
  +||\nablasl^3\sigma||_{L^2(\mathcal{S}_{u,v'})}\right)\mathrm{d}v',
\end{align*}
and Gr\"onwall's inequality finally gives
\begin{align*}
  ||\nablasl^3\rho||_{L^2(\mathcal{S}_{u,v})}\leq
  C(I,\Delta_{e_\star},\Delta_{\Gamma_\star},\Delta_{\Psi_\star})
  +C(I,\Delta_{e_\star},\Delta_{\Gamma_\star},\Delta_{\Psi_\star})
  \int_0^v||\nablasl^3\sigma||_{L^2(\mathcal{S}_{u,v'})}\mathrm{d}v'.
\end{align*}
In order to estimate~$||\nablasl^3\sigma||_{L^2(\mathcal{S}_{u,v})}$,
we make use, again, of the general expressions contained in
Appendix~\ref{nablaf}. For brevity we focus our attention
on~$||\delta^3\sigma||_{L^2(\mathcal{S}_{u,v})}$. Making use of the
integration identity in Appendix~\ref{IntegralIdentity} and the
commutators one finds that
\begin{align*}
  ||\delta^3\sigma||_{L^2(\mathcal{S}_{u,v})}=
  ||\bar\delta\delta^2\sigma||_{L^2(\mathcal{S}_{u,v})}
  =||\delta^2\bar\delta\sigma||_{L^2(\mathcal{S}_{u,v})}+ \cdots
\end{align*}
where the ellipsis denotes lower order derivative terms. Now, the
constraint structure equation (Codazzi equation)~\eqref{structureeq17}
lets us transform this norm further to a norm of the same order
for$\rho$. Thus, one concludes that
\begin{align*}
  ||\nablasl^3\rho||_{L^2(\mathcal{S}_{u,v})}\leq
  C(I,\Delta_{e_\star},\Delta_{\Gamma_\star},\Delta_{\Psi_\star})
  +C(I,\Delta_{e_\star},\Delta_{\Gamma_\star},\Delta_{\Psi_\star})
  \int_0^v||\nablasl^3\rho||_{L^2(\mathcal{S}_{u,v'})}\mathrm{d}v'.
\end{align*}
This inequality in turn implies that
\begin{align*}
  &||\nablasl^3\rho||_{L^2(\mathcal{S}_{u,v})}\leq
  C(I,\Delta_{e_\star},\Delta_{\Gamma_\star},\Delta_{\Psi_\star}),\\
  &||\nablasl^3\sigma||_{L^2(\mathcal{S}_{u,v})}\leq
  C(I,\Delta_{e_\star},\Delta_{\Gamma_\star},\Delta_{\Psi_\star}).
\end{align*}

\smallskip
\noindent
\textbf{Estimates for $\tau$ and $\chi$. } Making use of the structure
equation~\eqref{structureeq2} and the commutators we obtain
\begin{align*}
  D\delta^3\tau&=\delta^3\Psi_1+\Gamma\delta^3\Gamma_1+\Gamma\delta^3\tau
  +\Gamma\delta^2\Psi_1+\delta\Gamma\delta^2\Gamma+\Gamma^2\delta^2\Gamma \\
  &\quad+\Gamma^2\delta\Psi_1+\delta\Gamma\delta\Psi_1+\Gamma^3\delta\Gamma
  +\Gamma(\delta\Gamma)^2,
\end{align*}
where~$\Gamma_1$ contains combinations
of~$\epsilon$,~$\alpha$,~$\beta$,~$\rho$ and~$\sigma$. Thus, using the
main bootstrap assumption and the definition of~$\Delta_{\Psi}$ we
obtain that
\begin{align*}
  ||\nablasl^3\tau||_{L^2(\mathcal{S}_{u,v})}\leq
  C(I,\Delta_{e_\star},\Delta_{\Gamma_\star},\Delta_{\Psi_\star},\Delta_{\Psi})
  +C(I,\Delta_{e_\star},\Delta_{\Gamma_\star},\Delta_{\Psi_\star})
  \int_0^v||\nablasl^3\tau||_{L^2(\mathcal{S}_{u,v'})}\mathrm{d}v'.
\end{align*}
Accordingly, using Gr\"onwall's inequality one arrives to
\begin{align*}
  ||\nablasl^3\tau||_{L^2(\mathcal{S}_{u,v})}\leq
  C(I,\Delta_{e_\star},\Delta_{\Gamma_\star},\Delta_{\Psi_\star},\Delta_{\Psi}).
\end{align*}
The construction of an estimate for~$\chi$ is similar. In this case we
obtain that
\begin{align*}
  ||\nablasl^3\chi||_{L^2(\mathcal{S}_{u,v})}\leq
  C(I,\Delta_{e_\star},\Delta_{\Gamma_\star},\Delta_{\Psi_\star},\Delta_{\Psi}).
\end{align*}

\smallskip
\noindent
\textbf{Estimates for the remaining connection coefficients.} In order
to provide estimates for
\begin{align*}
  ||\nablasl^3\{\mu,\lambda,\alpha,\beta,\epsilon\}||_{L^2(\mathcal{S}_{u,v})},
\end{align*}
we make use of the transport equations along the short direction. The
proofs for the various coefficients are similar so for brevity we
discuss only the argument for~$\epsilon$. In this case one can readily
compute that
\begin{align*}
  \Delta\delta^3\epsilon&=-\delta^3\Psi_2+\Gamma\delta^3\Gamma_1
  +\Gamma\delta^3\epsilon+\Psi_1\delta^2\Gamma+\delta\Gamma\delta^2\Gamma
  +\Gamma^2\delta^2\Gamma\\
  &\quad+\Gamma\delta^2\Psi_2+\Gamma^2\delta\Psi_2+\Gamma^3\delta\Gamma
  +\Gamma(\delta\Gamma)^2+\Gamma^3\Psi_2+\Gamma^5,
\end{align*}
where the coefficients~$\Gamma_1$ do not contain $\epsilon$. Making
use of the short direction inequality of
Proposition~\ref{Proposition:TransportLpEstimates} we obtain that
\begin{align*}
  ||\nablasl^3\epsilon||_{L^2(\mathcal{S}_{u,v})}&\leq
  2||\nablasl^3\epsilon||_{L^2(\mathcal{S}_{0,v})}
  +C(I,\Delta_{e_\star},\Delta_{\Gamma_\star},\Delta_{\Psi_\star})
  \Delta_{\Psi}\varepsilon^{1/2}\\
  &\quad+C(I,\Delta_{e_\star},\Delta_{\Gamma_\star},\Delta_{\Psi_\star})
  \int_0^u||\nablasl^3\epsilon||_{L^2(\mathcal{S}_{u',v})}\mathrm{d}u'.
\end{align*}
In particular, we can choose the range of integration sufficiently
small so that
\begin{align*}
||\nablasl^3\epsilon||_{L^2(\mathcal{S}_{u,v})}\leq 3\Delta_{\Gamma_\star}.
\end{align*}
The argument
for~$||\nablasl^3\{\mu,\lambda,\alpha,\beta\}||_{L^2(\mathcal{S}_{u,v})}$
is the same.

\smallskip
\noindent
\textbf{Concluding the argument.} An inspection of the estimates
obtained in the previous paragraphs shows that we have improved the
initial bootstrap assumption. This concludes the proof of the
proposition.
\end{proof}

\subsection{Main estimates for the curvature}

We are now in the position to obtain the main estimates for the
components of the Weyl tensor. We start with an estimate on a given
pair of null hypersurfaces in terms of their value at hypersurfaces in
the past.

\begin{proposition}[\textbf{\em basic control of components of the
    Weyl tensor on the light cones in terms of its values on
    causal diamonds}]
  \label{Proposition:FirstMainEstimateCurvature}
  Suppose that we are given a solution to the vacuum EFEs in Stewart's
  gauge and that~$\mathcal{D}_{u,v}$ is contained in the existence
  area. The following~$L^2$ estimates for the Weyl curvature hold:
\begin{align*}
  &\sum_{i=0,1,2}\int_{\mathcal{N}_u(0,v)}|\Psi_i|^2+\sum_{j=1,2,3}
  \int_{\mathcal{N}'_v(0,u)}Q^{-1}|\Psi_j|^2\\ &\qquad
  \leq\sum_{i=0,1,2}\int_{\mathcal{N}_{0}(0,v)}|\Psi_i|^2
  +\sum_{j=1,2,3}\int_{\mathcal{N}'_{0}(0,u)}Q^{-1}|\Psi_j|^2 +
  \int_{\mathcal{D}_{u,v}}|\Psi_H\Psi\Gamma + \mbox{cc}|,
\end{align*}
where~$\Psi$ contains~$\Psi_k$, $k=0,...,4$, $\Psi_H$ denotes the
components~$\Psi_k$, $k=0,...,3$, ``cc'' denotes the complex conjugate
of the last term on the right-hand side and~$\Gamma$ stands for
  arbitrary connection coefficients from the collection~$\{\mu,
  \lambda, \alpha, \beta, \epsilon, \rho, \sigma, \tau\}$.
\end{proposition}

\begin{proof}
Assuming, as always that the vacuum field equations of GR are
satisfied, we start considering the Bianchi
identities~\eqref{Bianchi2} and~\eqref{Bianchi1} written schematically
as
\begin{align*}
& \Delta\Psi_0=\delta\Psi_1+\Gamma\Psi, \\
& D\Psi_1=\bar\delta\Psi_0+\Gamma\Psi.
\end{align*}
Then, integration by parts one obtains (again, using schematic
notation) that
\begin{align*}
  \int_{\mathcal{D}_{u,v}}\bar\Psi_0\Delta\Psi_0&
  =\int_{\mathcal{D}_{u,v}}\bar\Psi_0\delta\Psi_1
  +\int_{\mathcal{D}_{u,v}}\bar\Psi_0\Gamma\Psi \\
  &=-\int_{\mathcal{D}_{u,v}}\Psi_1\delta\bar\Psi_0
  -\int_{\mathcal{D}_{u,v}}\Psi_1\bar\Psi_0\varpi+\int\bar\Psi_0\Gamma\Psi  \\
  &=-\int_{\mathcal{D}_{u,v}}\Psi_1D\bar\Psi_1
  + \int_{\mathcal{D}_{u,v}}\{\bar\Psi_0,\Psi_1\}\Gamma\Psi .
\end{align*}
Hence, using the identities in Lemma~\ref{Lemma:IntegralIdentities},
we conclude that
\begin{align*}
  \int_{\mathcal{N}_u(0,v)}|\Psi_0|^2+\int_{\mathcal{N}'_v(0,u)}Q^{-1}|\Psi_1|^2
  \leq\int_{\mathcal{N}_{0}(0,v)}|\Psi_0|^2+\int_{\mathcal{N}'_{0}(0,u)}Q^{-1}
  |\Psi_1|^2+ \int_{\mathcal{D}_{u,v}}(|\{\Psi_0,\Psi_1\}\Psi\Gamma+\mbox{cc}|),
\end{align*}
where in the previous expression~$\Psi$
contains~$\Psi_{0,1,2}$. Analogous inequalities can be obtained for
the pairs~$\Delta\Psi_1$, $D\Psi_2$, and~$\Delta\Psi_2$, $D\Psi_3$.
\end{proof}

Similar estimates can be obtained for the first angular derivatives of
the components of the Weyl tensor.

\begin{proposition}[\textbf{\em control of the first angular
    derivatives of the components of the Weyl tensor}]
  \label{Proposition:SecondMainEstimateCurvature} Again let~$\mathcal{D}_{u,v}$
  be contained in the existence area, then we have that
\begin{align*}
  &\sum_{i=0,1,2}\int_{\mathcal{N}_u(0,v)}|\nablasl\Psi_i|^2
  +\sum_{j=1,2,3}\int_{\mathcal{N}'_v(0,u)}Q^{-1}|\nablasl\Psi_j|^2\\
  &\qquad\qquad \leq\sum_{i=0,1,2}\int_{\mathcal{N}_{0}(0,v)}|\nablasl\Psi_i|^2
  +\sum_{j=1,2,3}\int_{\mathcal{N}'_{0}(0,u)}Q^{-1}|\nablasl\Psi_j|^2
  + \int_{\mathcal{D}_{u,v}}|\nablasl\Psi_H|(|\Psi\Gamma^2|
  +|\Gamma\nablasl\Psi|+|\Psi\nablasl\Gamma|),
\end{align*}
where~$\Psi$ contains~$\Psi_k$, $k=0,...,4$, and~$\Psi_H$ contains
$\Psi_k$, $k=0,...,3$, and again ~$\Gamma$ stands for some
  combination of the connection coefficients~$\{\mu, \lambda, \alpha,
  \beta, \epsilon, \rho, \sigma, \tau\}$.
\end{proposition}

\begin{proof}
Again, we make use of integration by parts. Consider for example
\begin{align*}
  \int_{\mathcal{D}_{u,v}}\bar\delta\bar\Psi_0\Delta\delta\Psi_0&=
  \int_{\mathcal{D}_{u,v}}\bar\delta\bar\Psi_0\delta^2\Psi_1
  +\int_{\mathcal{D}_{u,v}}\bar\delta\bar\Psi_0(\Gamma^2\Psi_i
  +\Gamma\delta\Psi_i+\Psi_i\delta\Gamma) \\
  &=-\int_{\mathcal{D}_{u,v}}\delta\bar\delta\bar\Psi_0\delta\Psi_1
  +\int_{\mathcal{D}_{u,v}}\bar\delta\bar\Psi_0(\Gamma^2\Psi_i
  +\Gamma\delta\Psi_i+\Psi_i\delta\Gamma) \\
  & =-\int_{\mathcal{D}_{u,v}}\delta\Psi_1D\bar\delta\bar\Psi_1
  +\int_{\mathcal{D}_{u,v}}(\bar\delta\bar\Psi_0,\delta\Psi_1)
  (\Gamma^2\Psi_i+\Gamma\delta\Psi_i+\Psi_i\delta\Gamma)
\end{align*}
with~$i=0,\,1,\,2$. A similar expression can be obtained for the
combination
\begin{align*}
  \int_{\mathcal{D}_{u,v}}\delta\bar\Psi_0\Delta\bar\delta\Psi_0
  +\int_{\mathcal{D}_{u,v}}\bar\delta\Psi_1D\delta\bar\Psi_1.
\end{align*}
Thus, using Lemma~\ref{Lemma:IntegralIdentities} can conclude that
\begin{align*}
  \int_{\mathcal{N}_u(0,v)}|\nablasl\Psi_0|^2+\int_{\mathcal{N}'_v(0,u)}
  Q^{-1}|\nablasl\Psi_1|^2&\leq\int_{\mathcal{N}_{0}(0,v)}|\nablasl\Psi_0|^2
  +\int_{\mathcal{N}'_{0}(0,v)}Q^{-1}|\nablasl\Psi_1|^2  \\
  &\quad + \int_{\mathcal{D}_{u,v}}|\nablasl\{\Psi_0,\Psi_1\}|
  (|\Psi\Gamma^2|+|\Gamma\nablasl\Psi|+|\Psi\nablasl\Gamma|),
\end{align*}
where~$\Psi$ contains the components~$\Psi_0$, $\Psi_1$
and~$\Psi_2$. A similar computation for the other pairs of components
renders the desired result.
\end{proof}

The previous result can be extended to include higher order
derivatives. More precisely:

\begin{proposition} [\textbf{\em control of the higher angular
    derivatives of the components of the Weyl tensor}]
\label{Proposition:EstimatesDerivativesWeyl0123}
Let~$\mathcal{D}_{u,v}$ again be contained in the existence
area. Given a non-negative integer~$m$, one has
\begin{align*}
&\sum_{i=0,1,2}\int_{\mathcal{N}_u(0,v)}|\nablasl^m\Psi_i|^2
+\sum_{j=1,2,3}\int_{\mathcal{N}'_v(0,u)}Q^{-1}|\nablasl^m\Psi_j|^2\\
& \qquad\qquad\leq\sum_{i=0,1,2}\int_{\mathcal{N}_{0}(0,v)}|
\nablasl^m\Psi_i|^2+\sum_{j=1,2,3}
\int_{\mathcal{N}'_{0}(0,v)}Q^{-1}|\nablasl^m\Psi_j|^2\\
&\qquad\qquad\qquad\quad+ \int_{\mathcal{D}_{u,v}}|\nablasl^m\Psi_H|
\sum_{i_1+i_2+i_3+i_4=m}|\nablasl^{i_1}\Gamma^{i_2}||
\nablasl^{i_3}\Gamma||\nablasl^{i_4}\Psi|.
\end{align*}
where~$\Psi$ contains the components~$\Psi_k$, $k=0,...,4$,
and~$\Psi_H$ contains the components~$\Psi_k$, $k=0,...,3$.
Again~$\Gamma$ stands for some combination of the connection
  coefficients~$\{\mu, \lambda, \alpha, \beta, \epsilon, \rho, \sigma,
  \tau\}$.
\end{proposition}

To wrap up the argument we also need estimates on the
components~$\Psi_3$ and~$\Psi_4$. These follow from the Bianchi
identities
\begin{align}
  & \Delta\Psi_3-\delta\Psi_4=4\Psi_4\beta-\Psi_4\tau
  -4\Psi_3\mu,\nonumber\\
  & D\Psi_4-\bar\delta\Psi_3=\Psi_4(\rho-4\epsilon)
  +2\Psi_3(3\alpha+2\beta)-3\Psi_2\lambda.\label{Bianchi_Prop17}
\end{align}
Using a similar approach to the one
used in the previous propositions one can prove the following:
\begin{proposition}[\textbf{\em control of the higher angular
    derivatives of the ``bad'' components of the Weyl tensor}] 
\label{Proposition:EstimatesDerivativesWeyl34}
Let~$\mathcal{D}_{u,v}$ be contained in the existence area. Given a
non-negative integer~$m$, one has that 
\begin{align*}
  &\int_{\mathcal{N}_u(0,v)}|\nablasl^m\Psi_3|^2
  +\int_{\mathcal{N}'_v(0,u)}Q^{-1}|\nablasl^m\Psi_4|^2\\
  & \qquad\qquad\leq\int_{\mathcal{N}_{0}(0,v)}|\nablasl^m\Psi_3|^2
  +\int_{\mathcal{N}'_{0}(0,u)}Q^{-1}
  |\nablasl^m\Psi_4|^2  \\
  & \qquad\qquad\qquad\quad
  + \int_{\mathcal{D}_{u,v}}|\nablasl^m\Psi_4|\sum_{i_1+i_2+i_3+i_4=m}|
  \nablasl^{i_1}\Gamma'^{\,i_2}
  ||\nablasl^{i_3}\Gamma'||\nablasl^{i_4}\Psi_4| \\
  & \qquad\qquad\qquad\quad+ \int_{\mathcal{D}_{u,v}}|\nablasl^m\Psi_3
  |\sum_{i_1+i_2+i_3+i_4=m}|\nablasl^{i_1}\Gamma^{i_2}||\nablasl^{i_3}\Gamma||
  \nablasl^{i_4}\Psi| \\
  & \qquad\qquad\qquad\quad+\int_{\mathcal{D}_{u,v}}|\nablasl^m\Psi_4|
  \sum_{i_1+i_2+i_3+i_4=m}|\nablasl^{i_1}\Gamma^{i_2}
  ||\nablasl^{i_3}\Gamma||\nablasl^{i_4}\Psi'_H|, 
\end{align*}
where~$\Psi$ contains the components~$\Psi_3$ and~$\Psi_4$,
while~$\Psi'_H$ contains the components~$\Psi_2$
and~$\Psi_3$. Here~$\Gamma$ stands for some combination of the
connection coefficients~$\{\mu, \lambda, \alpha, \beta, \epsilon,
\rho, \tau,\sigma\}$. Because neither the coefficient of~$\Psi_4$ on
the right hand side of~\eqref{Bianchi_Prop17} nor the
NP~$\delta\bar{\delta}$-commutator~\eqref{NPCommutator4}
contain~$\tau,\chi$ terms, neither does~$\Gamma'$.
\end{proposition}

Propositions~\ref{Proposition:FirstMainEstimateCurvature}-\ref{Proposition:EstimatesDerivativesWeyl34}
clearly make no use of the estimates demonstrated in the previous
sections. Finally, we therefore conclude this section with the main
estimate for the components of the Weyl tensor employing our earlier
work. This proposition makes only assumptions on the initial data.

\begin{proposition}[\textbf{\em control of the components of the Weyl
    tensor in terms of the initial data}]
\label{Proposition:FinalEstimateCurvature}
Suppose we are given a solution to the vacuum EFE's in Stewart's gauge
emanating from data for the CIVP as prepared in
Lemma~\ref{Lemma:FreeDataCIVP}, satisfying
\begin{align*}
\Delta_{e_\star},\;\Delta_{\Gamma_\star}, \;\Delta_{\Psi_\star} <\infty,
\end{align*}
with the solution itself satisfying 
\begin{align*}
  \sup_{u,v}||\{\mu, \lambda, \alpha, \beta, \epsilon, \rho, \sigma,
  \tau, \chi\}||_{L^\infty(\mathcal{S}_{u,v})}&< \infty\,,\quad
  \sup_{u,v}||\nablasl\{\mu, \lambda, \alpha, \beta, \epsilon, \rho,
  \sigma\}
  ||_{L^4(\mathcal{S}_{u,v})}<\infty\,,\\
  \sup_{u,v}||\nablasl^2\{\mu,
  \lambda, \alpha, \beta, \epsilon, \rho, \sigma,
  \tau\}||_{L^2(\mathcal{S}_{u,v})}&<\infty\,,\quad
  \sup_{u,v}||\nablasl^3\{\mu,\lambda,\alpha,\beta,\epsilon,\tau\}
  ||_{L^2(\mathcal{S}_{u,v})}<\infty\,,\\
  \Delta_{\Psi}(\mathcal{S})<\infty\,, \quad
  \Delta_{\Psi}<\infty\,,\quad\quad\quad\quad&
\end{align*}
on some truncated causal diamond~$\mathcal{D}_{u,v_\bullet}^{\,t}$.
Then there
exists~$\varepsilon_\star=\varepsilon_\star(I,\Delta_{e_\star},\Delta_{\Gamma_\star},
\Delta_{\Psi_\star})$ such that for~$\varepsilon_\star\leq\varepsilon$
we have
\begin{align*}
\Delta_{\Psi}\leq C(I,\Delta_{e_\star},\Delta_{\Gamma_\star},\Delta_{\Psi_\star}).
\end{align*}
\end{proposition}

\begin{proof}
The aim in this proof is to control the terms involving integrals on
the diamond~$\mathcal{D}_{u,v}$ arising in
Propositions~\ref{Proposition:EstimatesDerivativesWeyl0123}
and~\ref{Proposition:EstimatesDerivativesWeyl34}
for~$m\leq3$. Starting with
Proposition~\ref{Proposition:EstimatesDerivativesWeyl0123} one has
that the relevant integral is given by
\begin{align}
\label{IntegralPropWeyl0123}
\int_{\mathcal{D}_{u,v}}|\nablasl^m\Psi_H|\sum_{i_1+i_2+i_3+i_4=m}|
\nablasl^{i_1}\Gamma^{i_2}||\nablasl^{i_3}\Gamma||\nablasl^{i_4}\Psi|,
\end{align}
for~$(u,v)$ in~$\mathcal{D}_{\varepsilon,v_\bullet}^{\,t}$. On the one hand,
for the first factor in this integral, given
that~$\Psi_H\in\{\Psi_0,\Psi_1, \Psi_2, \Psi_3 \}$ can be controlled
in~$L^2(\mathcal{N}_u(0,v))$, one readily obtains
\begin{align*}
  ||\nablasl^m\Psi_H||_{L^2(\mathcal{D}_{u,v})}=
  \left(\int_0^u\int_0^v\int_{\mathcal{S}_{u',v'}}|\nablasl^m\Psi_H|^2
  \mathrm{d}v'\mathrm{d}u'\right)^{1/2}
  \leq C\Delta_{\Psi}\varepsilon^{1/2},
\end{align*}
On the other, for the factors contains~$\Psi_4$, one only has control
on~$\mathcal{N}'_v(0,u)$ ---that is,
\begin{align*}
||\nablasl^m\Psi||_{L^2(\mathcal{D}_{u,v})}\leq C\Delta_{\Psi}.
\end{align*}
It then follows that the integral~\eqref{IntegralPropWeyl0123} can be
estimated as,
\begin{align}
  &\int_{\mathcal{D}_{u,v}}|\nablasl^m\Psi_H|\sum_{i_1+i_2+i_3+i_4=m}
  |\nablasl^{i_1}\Gamma^{i_2}||\nablasl^{i_3}\Gamma||\nablasl^{i_4}\Psi|
  \nonumber\\
  &\qquad\qquad\qquad\leq C\varepsilon^{1/2}\Delta_{\Psi}\sum_{i_1+i_2+i_3+i_4\leq3}
  ||\nablasl^{i_1}\Gamma^{i_2}\nablasl^{i_3}\Gamma\nablasl^{i_4}
  \Psi||_{L^2(\mathcal{D}_{u,v})}.\label{IntermediateInequalityPsi}
\end{align}
In particular, for~$m=0$, the right-hand side of the above inequality
gives
\begin{align*}
  C\varepsilon^{1/2}\Delta_{\Psi}||\Gamma\Psi||_{L^2(\mathcal{D}_{u,v})}\leq
  C\epsilon^{1/2}\Delta_{\Psi}||\Gamma||_{L^{\infty}(S)}||\Psi||_{L^2(\mathcal{D}_{u,v})}
  \leq C(I,\Delta_{e_\star},\Delta_{\Gamma_\star},\Delta_{\Psi_\star},
  \Delta_{\Psi})\varepsilon^{1/2}.
\end{align*}
Next, when~$m=1$, we have that the right-hand of
inequality~\eqref{IntermediateInequalityPsi} gives
\begin{align*}
  C\varepsilon^{1/2}\Delta_{\Psi}||\Gamma^2\Psi+\Gamma|\nablasl\Psi|
  +\Psi|\nablasl\Gamma|||_{L^2(\mathcal{D}_{u,v})}.
\end{align*}
The first two terms can be controlled like the case~$m=0$, and the
third can be controlled by means of Sobolev embedding:
\begin{align*}
  ||\Psi|\nablasl\Gamma|||_{L^2(\mathcal{S}_{u,v})}
  &\leq||\nablasl\Gamma||_{L^{\infty}(\mathcal{S}_{u,v})}
  ||\Psi||_{L^2(\mathcal{D}_{u,v})} \\
  &\leq\left(||\nablasl\Gamma||_{L^2(\mathcal{S}_{u,v})}
  +||\nablasl^2\Gamma||_{L^2(\mathcal{S}_{u,v})}
  +||\nablasl^3\Gamma||_{L^2(\mathcal{S}_{u,v})}\right)
  ||\Psi||_{L^2(\mathcal{D}_{u,v})}. 
\end{align*}
For the case~$m=2$, the terms on the right-hand side of
inequality~\eqref{IntermediateInequalityPsi} give
\begin{align}
  C\varepsilon^{1/2}\Delta_{\Psi}||\Gamma|\nablasl^2\Psi|+\Gamma^3\Psi
  +\Gamma^2|\nablasl\Psi|+\Psi\Gamma|\nablasl\Gamma|+|\nablasl\Psi
  ||\nablasl\Gamma|+\Psi|\nablasl^2\Gamma|||_{L^2(\mathcal{D}_{u,v})}.
\label{IntermediateEqualityPsim2}
\end{align}
All terms, save last one, can be controlled by analysis analogous to
that used in the previous cases. To see this, we split
the~$L^\infty$-norm of the connection coefficient and the~$L^2$-normal
of the curvature. The~$L^\infty$-normal can then be controlled by
means of Sobolev embedding. For the last term, we have
\begin{align*}
  \left(\int_0^u\int_0^v
  \int_{\mathcal{S}_{u',v'}}(\Psi|\nablasl^2\Gamma|)^2
    \mathrm{d}v'\mathrm{d}u'\right)^{1/2}&\leq
  \left(\int_0^u\int_0^v||\Psi||^2_{L^{\infty}(\mathcal{S}_{u',v'})}
  ||\nablasl^2\Gamma||^2_{L^2(\mathcal{S}_{u',v'})}\mathrm{d}v'
  \mathrm{d}u'\right)^{1/2}\\
  &\leq\left(\sup_{\mathcal{D}_{u,v}}
  ||\nablasl^2\Gamma||_{L^2(\mathcal{S}_{u',v'})}\right)\sum_{i=0}^{2}
  ||\nablasl^i\Psi||_{L^2(\mathcal{D}_{u,v})},
\end{align*}
hence~\eqref{IntermediateEqualityPsim2} under control.

Finally, when~$m=3$ the terms on the right-hand side of
inequality~\eqref{IntermediateInequalityPsi} give
\begin{align*}
  &C\varepsilon^{1/2}\Delta_{\Psi}||(\Gamma|\nablasl^3\Psi|
  +\Psi|\nablasl^3\Gamma|+|\nablasl\Gamma||\nablasl^2\Psi|
  +|\nablasl\Psi||\nablasl^2\Gamma|+\Gamma^2|\nablasl^2\Psi|
  +\Gamma\Psi|\nablasl^2\Gamma| \\
  &+\Gamma|\nablasl\Gamma||\nablasl\Psi|+\Psi|\nablasl\Gamma|^2
  +\Gamma^3|\nablasl\Psi|+\Psi\Gamma^2|\nablasl\Gamma|
  +\Gamma^4\Psi )||_{L^2(\mathcal{D}_{u,v})}.
\end{align*}
The various terms in this expression can be estimated in a manner
analogous to the previous cases. We conclude that the integral
over~$\mathcal{D}_{u,v}$ can be controlled by
\begin{align*}
  \int_{\mathcal{D}_{u,v}}|\nablasl^m\Psi_H|\sum_{i_1+i_2+i_3+i_4=m}
  |\nablasl^{i_1}\Gamma^{i_2}||\nablasl^{i_3}\Gamma||\nablasl^{i_4}\Psi|
  \leq C(I,\Delta_{e_\star},\Delta_{\Gamma_\star},\Delta_{\Psi_\star}, \Delta_{\Psi})
  \varepsilon^{1/2}.
\end{align*}
We now proceed to examine the estimate from
Proposition~\ref{Proposition:EstimatesDerivativesWeyl34}. The terms in
\begin{align*}
  \int_{\mathcal{D}_{u,v}}|\nablasl^m\Psi_3|\sum_{i_1+i_2+i_3+i_4=m}
  |\nablasl^{i_1}\Gamma^{i_2}||\nablasl^{i_3}\Gamma||\nablasl^{i_4}\Psi|
\end{align*}
are identical to those already analysed and can be controlled
by
\begin{align*}
  C(I,\Delta_{e_\star},\Delta_{\Gamma_\star},\Delta_{\Psi_\star}, \Delta_{\Psi})
  \varepsilon^{1/2}.
\end{align*}
The terms
\begin{align*}
  \int_{\mathcal{D}_{u,v}}|\nablasl^m\Psi_4|\sum_{i_1+i_2+i_3+i_4=m}
  |\nablasl^{i_1}\Gamma^{i_2}||\nablasl^{i_3}\Gamma||\nablasl^{i_4}\Psi'_H|
\end{align*}
can also be controlled because the components of the Weyl tensor
contained in~$\Psi'_H=\{\Psi_2,\Psi_3\}$ have already been shown to be
controlled. The remaining terms are
\begin{align*}
  \int_{\mathcal{D}_{u,v}}|\nablasl^m\Psi_4|\sum_{i_1+i_2+i_3+i_4=m}|
  \nablasl^{i_1}\Gamma'^{\,i_2}||\nablasl^{i_3}(\rho+\epsilon)
  ||\nablasl^{i_4}\Psi_4|.
\end{align*}
We proceed to by treating~$m=0,\ldots,3$ individually. Notice in
particular, that~$\Gamma'$ does contains neither~$\tau$
nor~$\chi$. Crucially the weakest bounds of
  Proposition~\ref{Proposition:SummaryBasicEstmatesConnectionCurvature}
  and~Proposition~\ref{Proposition:ImprovedEstimates}
  involving~$\Delta_\Psi$ are therefore not invoked in the resulting
  computation, and so after a lengthy analysis one concludes that these
  terms satisfy
\begin{align*}
  &\int_{\mathcal{D}_{u,v}}|\nablasl^m\Psi_4|\sum_{i_1+i_2+i_3+i_4=m}
  |\nablasl^{i_1}\Gamma^{i_2}||\nablasl^{i_3}(\rho+\epsilon)
  ||\nablasl^{i_4}\Psi_4|\\
  &\qquad\qquad\leq C(I,\Delta_{e_\star},
  \Delta_{\Gamma_\star},\Delta_{\Psi_\star})
  \int_0^v||\nablasl^m\Psi_4||_{L^2(\mathcal{N}'_{v'}(0,u))}
  \sum_{i=0}^m||\nablasl^i\Psi_4||_{L^2(\mathcal{N}'_{v'}(0,u))}\mathrm{d}v'\\
  &\qquad\qquad \leq C(I,\Delta_{e_\star},\Delta_{\Gamma_\star},\Delta_{\Psi_\star})
  \int_0^v\sum_{i=0}^m||\nablasl^i\Psi_4||^2_{L^2(\mathcal{N}'_{v'}(0,u))}\mathrm{d}v' .
\end{align*}
Substituting the previous expressions into the inequality of
Proposition~\ref{Proposition:EstimatesDerivativesWeyl34} one concludes
that
\begin{align*}
\sum_{i=0}^3||\nablasl^i\Psi_4||^2_{L^2(\mathcal{N}'_v(0,u))} &\leq
C\Delta_{\Psi_\star}+C(I,\Delta_{e_\star},\Delta_{\Gamma_\star},\Delta_{\Psi_\star},
\Delta_{\Psi})\varepsilon^{1/2}\\
&\quad+C(I,\Delta_{e_\star},\Delta_{\Gamma_\star},
\Delta_{\Psi_\star})
\int_0^v\sum_{i=0}^m||\nablasl^i\Psi_4||^2_{L^2(\mathcal{N}'_v(0,u))}\mathrm{d}v'. 
\end{align*}
Accordingly, using Gr\"onwall's inequality and taking~$\varepsilon$
sufficiently small one finds,
\begin{align*}
  \sum_{i=0}^3||\nablasl^i\Psi_4||^2_{L^2(\mathcal{N}'_v(0,u))} \leq
  C\Delta_{\Psi_\star}+C(I,\Delta_{e_\star},\Delta_{\Gamma_\star},\Delta_{\Psi_\star},
  \Delta_{\Psi})\varepsilon^{1/2}\leq C(I,\Delta_{e_\star},
  \Delta_{\Gamma_\star},\Delta_{\Psi_\star}).
\end{align*}
Using this estimate, it follows that
\begin{align*}
\Delta_{\Psi}\leq C(I,\Delta_{e_\star},\Delta_{\Gamma_\star},\Delta_{\Psi_\star})
+C(I,\Delta_{e_\star},\Delta_{\Gamma_\star},\Delta_{\Psi_\star},
\Delta_{\Psi})\varepsilon^{1/2}.
\end{align*}
Taking~$\varepsilon$ small enough we have proven the proposition.
\end{proof}

\section{Last slice argument and the end of the proof}

In this section we make use of the estimates developed in the previous
sections to show the existence of solutions to the vacuum Einstein
field equations exists in the rectangular domain
\begin{align*}
\mathcal{D}=\{0\leq
u\leq\varepsilon,\ 0\leq v\leq v_\bullet\}.
\end{align*}
The strategy makes use of an argument by contradiction known as the
\emph{last slice argument}, in which it is assumed that the solution
does not fill the whole of~$\mathcal{D}$ and, accordingly, there
exists a hypersurface (the last slice) which bounds the domain of
existence of the solution. The estimates we have constructed in the
previous sections allow then to show that, in fact, on this slice the
solution and its derivatives are bounded. Thus, it is possible to make
use of the standard Cauchy problem for the Einstein field equations to
show that the solution extends beyond the hypersurface~$t^*$ ---an
observation which contradicts the original assumption.

\subsection{Setup}

In order to implement the above strategy one foliates the
rectangle~$\mathcal{D}$ by means of spacelike hypersurfaces. To this
end recall definition~\eqref{eqn:time_coord_defn} of the \emph{time
  function}
\begin{align*}
t\equiv u+v
\end{align*}
so that~$\nabla t$ is timelike. Let~$\Sigma_t$ denote the level sets
of~$t$.

The last slice argument starts by invoking the local existence result
for the CIVP based on Rendall's reduction strategy. This result
ensures the existence of a solution to evolution equations in a
neighbourhood~$\mathcal{V}$ of~$\mathcal{S}_\star$
on~$J^+(\mathcal{S}_\star)$ ---see
Theorem~\ref{Theorem:RendallLocalExistence}. Within this neighbourhood
there exists a truncated causal diamond on which all the bootstrap
assumptions required to obtain the estimates from the previous
sections hold. Thus, we know that the set on which the bootstrap
hypotheses hold is non-empty, and hence render our estimates
applicable. The rest of the last slice argument proceeds now to show
that this basic truncated causal diamond can be progressively enlarged
as long as one has control on the initial data on the null
cone~$\mathcal{N}_\star$ thus exhausting the domain~$\mathcal{D}$.

 If the solution does not exist in the whole of~$\mathcal{D}$, we must
 have~$t^*\in(0,I+\varepsilon)$ such that
\begin{align*}
t^*=\sup\{t:\ \mbox{the spacetime exists  in }
\mathcal{D}\cap\cup_{\tau\in[0,t)}\Sigma_{\tau}\}.
\end{align*}
Let~$\bmh_t$ and~$\bmK_t$ be, respectively, the induced metric and
second fundamental form on~$\Sigma_t$. A schematic depiction of the
geometric set-up is shown in Figure~\ref{Fig:LastSlice}.

\begin{figure}[t]
\centering
\includegraphics[width=0.8\textwidth]{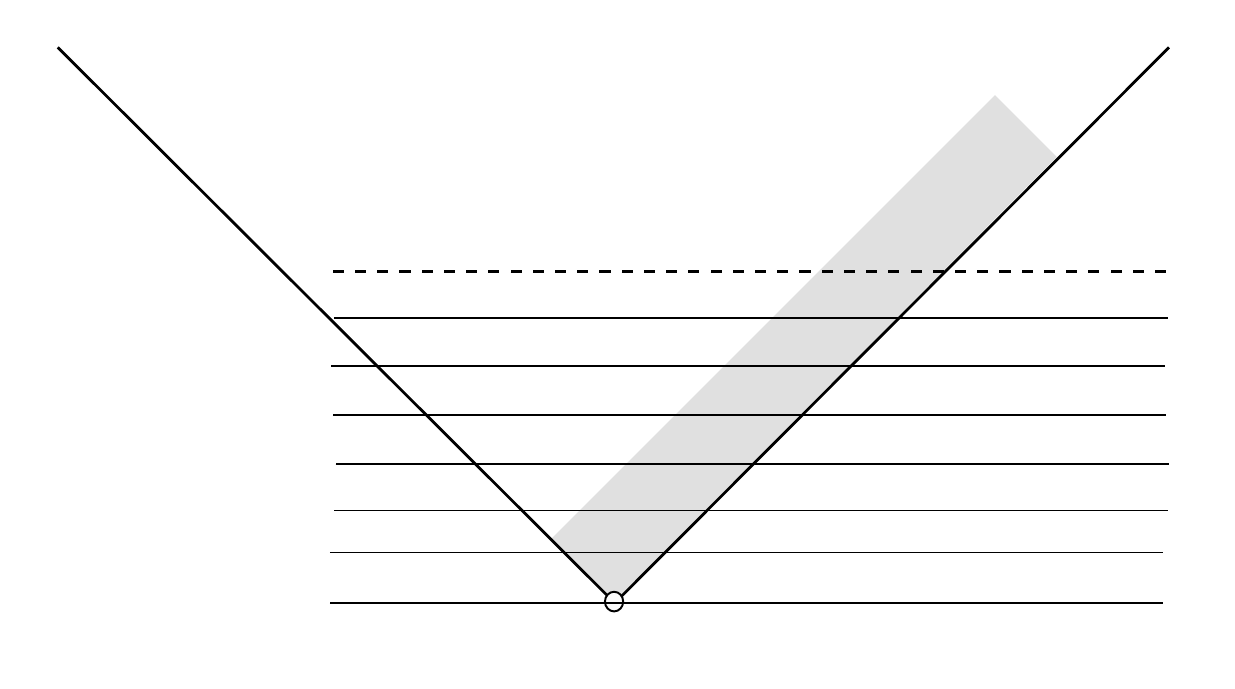}
\put(-40,140){$\mathcal{N}_\star$}
\put(-80,130){$\mathcal{D}$}
\put(-320,140){$\mathcal{N}'_\star$}
\put(-180,5){$\mathcal{S}_{u_\star,v_\star}$}
\put(-15,112){$t^*$}
\put(-345,35){$t=u+v=\mbox{constant}$}
\put(-260,60){$\Sigma_t$}
\caption{Setup for the \emph{last slice argument}. On each slice of
  the family of hypersurfaces~$\Sigma_t$ one has a smooth initial data
  set~$(\bmh_t,\bmK_t)$ for the vacuum Einstein field equations. The
  estimates of Proposition~\ref{Proposition:FinalEstimateCurvature}
  then show that even on the \emph{last slice}~$\Sigma_{t^*}$ one has
  a well initial data set. Thus, the solution can be extended beyond
  this slice ---a contradiction! }
\label{Fig:LastSlice}
\end{figure}

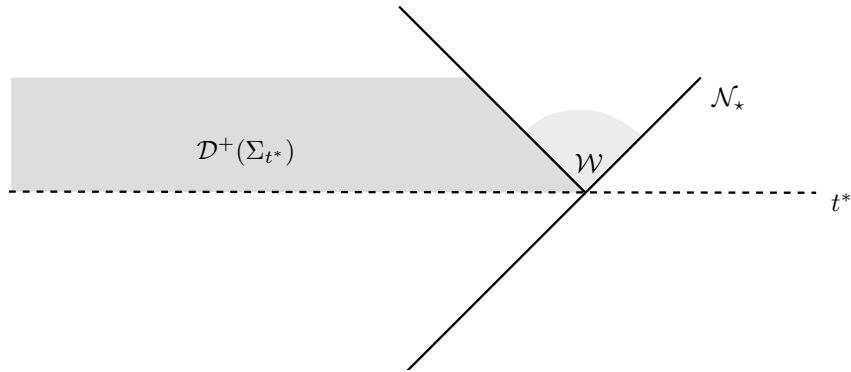
\begin{figure}[t]
\centering
\includegraphics[width=0.8\textwidth]{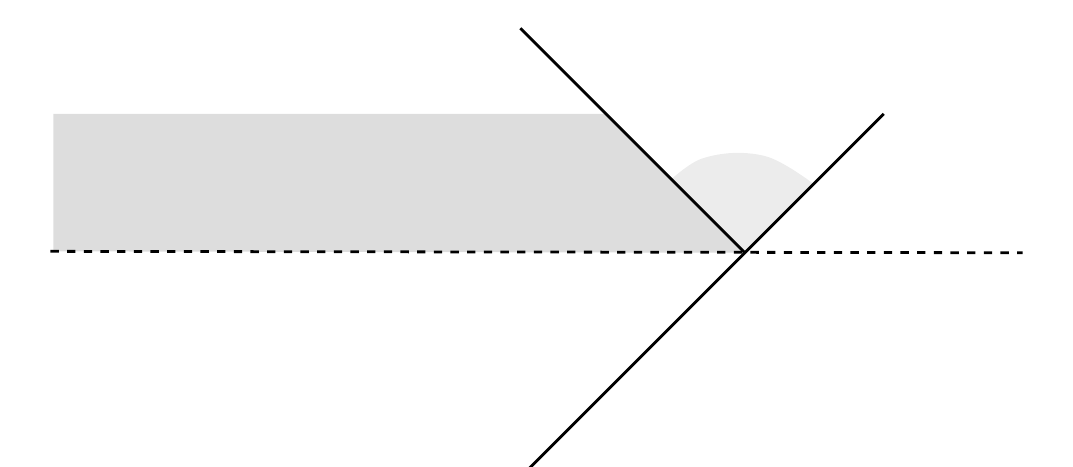}
\put(-10,60){$t^*$}
\put(-250,80){$\mathcal{D}^+(\Sigma_{t^*})$}
\put(-55,100){$\mathcal{N}_\star$}
\put(-107,75){$\mathcal{W}$}
\caption{Zoom in on the hypothetical last slice. Regular Cauchy
  initial data on $\Sigma_{t^*}$ allows to extend the solution to, at
  least, a slab on $\mathcal{D}^+(\Sigma_{t^*})$ making use of the
  standard Cauchy problem for the Einstein field equations. On the
  wedge $\mathcal{W}$, a solution can be recovered by appealing to
  Rendall's formulation of the local CIVP.}
\label{Fig:LastSliceZoom}
\end{figure}

\subsection{Main argument}

In the following we will show that the fields~$\bmh_t$ and~$\bmK_t$
converge in~$C^{\infty}$ to fields~$\bmh_{t^*}$
and~$\bmK_{t^*}$. Moreover, it will be shown that the
pair~$(\bmh_{t^*},\bmK_{t^*})$ satisfy the Einstein constraint
equations on~$\Sigma_{t^*}$. In order to show this, it is necessary to
show that all derivatives of~$\bmh_t$ are bounded uniformly
in~$L^2(\Sigma_t)$ for all~$t<t^*$. The method proceeds by induction:

\smallskip
\noindent
\textbf{Base step.} The first step corresponds, in essence, to the
estimates obtained in the previous sections. More precisely, we have
first derived uniform estimates for the~$L^{\infty}$-norm of the
\emph{zeroth order derivatives} of connection on~$\mathcal{D}$ ---see
Proposition~\ref{Proposition:FirstEstimateConnection}. For
this we needed to assume that
\begin{align}
\label{ListAssumptionsBaseStep}
\sup_{u,v}||\nablasl^2\tau||_{L^2(\mathcal{S}_{u,v})}<\infty, \qquad
\sup_{u,v}||\nablasl^3\tau||_{L^2(\mathcal{S}_{u,v})}<\infty,
\qquad \Delta_{\Psi}(\mathcal{S})<\infty, \qquad \Delta_{\Psi}<\infty
\end{align}
on the truncated causal diamond. These conditions also lead to the
analysis the~$L^4$-norms of the first order derivatives 
(Proposition \ref{PropositionSecondEstimateConnection})
and~$L^2$-norms of the second order derivatives of the connection
---see Proposition~\ref{Proposition:ThirdEstimateConnection}. Now, using the
bootstrap assumptions, it follows
that~$\Delta_{\Psi}(\mathcal{S})<\infty$ uniformly on $\mathcal{D}$
with bounds given in terms of the initial data ---thus, this condition
can be removed from the list
in~\eqref{ListAssumptionsBaseStep}. Similarly, we can also drop the
condition~$||\nablasl^3\tau||_{L^2(\mathcal{S}_{u,v})}<\infty$ and
estimate the~$L^2$-norm of the third order angular derivatives of the
connection. In order to do so, we make use of the~$D$-direction
(i.e. the \emph{long direction}) equations for the NP
coefficients~$\rho$ and~$\sigma$, rather than the equations along the
short direction as we want to avoid dealing with the higher order
derivative of~$\tau$ on spheres~$\mathcal{S}_{u,v}$. Now, using
integration by parts, one concludes that~$\Delta_{\Psi}$ satisfies a
similar uniform bound on~$\mathcal{D}$. Thus, it has been shown that
given some initial data on the initial light cone, it is possible to
estimate the~$L^2$-norm on the spheres~$\mathcal{S}_{u,v}$ of the
connection coefficients and their derivatives up to third order.

\smallskip
\noindent
\textbf{Intermediate step.} The previous analysis is the \emph{base
  step} of the induction. As an intermediate induction step one
analyses the fourth order derivatives of the connection
coefficients. To this end, we make use of the same approach used in
the analysis of the third order derivatives in
Proposition~\ref{Proposition:ImprovedEstimates} This approach requires
the control of the norms of the fourth order derivatives of the
components of the Weyl tensor on the light cone. As in the case of the
Base Step, the required bounds need to be uniform on the truncated
causal diamond with bounds given in terms of the initial data. This
control can be achieved by the using integration by parts as in the
analysis of Proposition~\ref{Proposition:FinalEstimateCurvature}.

\begin{remark}
{\em The reason the method to analyse the fourth order derivatives of
  the connection coefficients is different from that of the third and
  lower orders lies in the structural properties of the equations
  ---these properties become manifest when considering higher order
  derivatives. In particular, one has that:
\begin{enumerate}[i).]
\item For zeroth-order derivatives, we cannot make use of the Codazzi
  equation to access the norms of~$\rho$ and~$\sigma$, since the
  Codazzi equation is a first order equation for the derivatives
  of~$\rho$ and~$\sigma$. Further difficulties arise from the
  nonlinear term~$\rho^2$ in the~$D$-direction
  equation~\eqref{structureeq13} for the coefficient~$\rho$.

\item For the first-order derivatives, we can readily estimate
  the~$L^2$-norm of the connection. However, this is not enough for
  the second order derivatives.  In the~$L^2$ estimate for the second
  order derivatives of the connections, we need H\"older's inequality
  to separate products of the form~$\delta\Gamma\times
  \delta\Gamma$. This procedure leads to estimates involving
  the~$L^4$-norm.
\end{enumerate}
}
\end{remark}

\smallskip
\noindent
\textbf{Induction step.} A procedure analogous to the one used to
control the fourth order derivatives of the connection coefficients is
employed to estimate the~$k+1$-th order derivatives of the connection
if control on the derivatives of~$k$-th order is assumed. This
calculation, requires, in particular, control of the value of such
norms on the initial light cone ---this control follows readily from
the procedure used to evaluate the \emph{formal derivatives} on the
initial light cone ---see Lemma~\ref{Lemma:FormalDrivatives}.

\smallskip
\noindent
\textbf{Concluding the argument.} The previous step shows that it is
possible to obtain control over the~$L^2$-norms of all angular
derivatives of the connection over the rectangular
domain~$\mathcal{D}$. Control of the derivatives respect to the
optical functions~$u$ and~$v$ can be obtained by applying, as
required, the directional covariant derivatives~$D$ and~$\Delta$ to
the evolution equations and commuting. Since the domain is bounded,
then all derivatives of~$\bmh_t$ and~$\bmK_t$ are bounded uniformly
in~$L^2(\Sigma_t)$ for~$t<t^*$. Moreover, one has that the 1-parameter
family of data~$(\bmh_t,\bmK_t)$ converges uniformly in~$C^\infty$ to
a pair~$(\bmh_{t^*},\bmK_{t^*})$. The pair~$(\bmh_{t^*},\bmK_{t^*})$
satisfies the Einstein constraint equations on the hypersurface
defined~$t=t^*$ ---see~\cite{Luk12}. This leads to a contradiction
with the assumption of the existence of a last slice as the theory of
the Cauchy problem for the Einstein field equations allows us to
readily obtain a (future) development of the data
set~$(\bmh_{t^*},\bmK_{t^*})$ ---see Figure \ref{Fig:LastSliceZoom}
Thus, no such \emph{last slice} exists and the solution to the
Einstein vacuum equations exists on the whole of the rectangular
domain~$\mathcal{D}$.

\subsection{Statement of the main result}

The long analysis of the preceding sections leads to the following:

\begin{theorem}[\textbf{\em main result ---improved local existence
  for the CIVP for the EFE}] Given regular initial data for the vacuum
  Einstein field equations as contructed in
  Lemma~\ref{Lemma:FreeDataCIVP} on the null
  hypersurfaces~$\mathcal{N}_\star\cup\mathcal{N}'_\star$
  for~$I\equiv\{0\leq v\leq v_\bullet\}$, there exists~$\varepsilon>0$
  such that a unique smooth solution to the vacuum Einstein field
  equations exists in the region where~$v\in I$ and~$0\leq
  u\leq \varepsilon$ defined by the null coordinates~$(u,v)$. The
  number~$\varepsilon$ can be chosen to depend only on~$I$,
  $\Delta_{e_\star}$, $\Delta_{\Gamma_\star}$
  and~$\Delta_{\Psi_\star}$. Furthermore, in this area one has that,
  \begin{align*}
  &\sup_{u,v}\sup_{\Gamma\in\{\mu,\lambda,\rho,\sigma,\alpha,\beta,\epsilon,\tau,\chi\}}
    \max\bigg\{\sum_{i=0}^1
    ||\nablasl^i\Gamma||_{L^{\infty}(\mathcal{S}_{u,v})},\sum_{i=0}^2
  ||\nablasl^i\Gamma||_{L^4(\mathcal{S}_{u,v})},
    \sum_{i=0}^3||\nablasl^i\Gamma||_{L^2(\mathcal{S}_{u,v})}\bigg\}\\
  &+\sum_{i=0}^3\sup_{\Psi\in\{\Psi_0,\Psi_1,\Psi_2,\Psi_3\}}\sup_u
  ||\nablasl^i\Psi||_{L^2(\mathcal{N}_u)}+\sup_{\Psi\in\{\Psi_1,\Psi_2,\Psi_3,\Psi_4\}}
  \sup_v||\nablasl^i\Psi||_{L^2(\mathcal{N}'_v)}\leq C(I,\Delta_{e_\star},
  \Delta_{\Gamma_\star},\Delta_{\Psi_\star}).
\end{align*}
\end{theorem}

\section*{Acknowledgements}

We are grateful to Edgar Gasperin and Arick Shao for helpful
discussions. We also gratefully acknowledge support offered by the
London Mathematical Society research in pairs scheme. DH was supported
by the FCT (Portugal) IF Program IF/00577/2015 and
PTDC/MAT-APL/30043/2017. PZ acknowledges the support of the China
Scholarship Council.

\appendix

\section{The Einstein field equations in the NP formalism}
\label{App:NP-formalism}

This appendix serves as quick reference of the basic equations of our
analysis. Throughout we make use of the NP formalism in the
conventions used in the book by J. Stewart~\cite{Ste91} which, in
turn, follows the conventions in~\cite{PenRin84}.

\subsubsection*{The spin connection coefficients}

Given a NP frame~$\{l^a,\ n^a,\ m^a,\ \bar{m}^a \}$, we define the
complex spin connection coefficients as,
\begin{align*}
\kappa&\equiv -m^al^b\nabla_bl_a,\qquad \rho\equiv
  -m^a\bar{m}^b\nabla_bl_a, &\qquad  \sigma&\equiv -m^am^b\nabla_bl_a, \qquad
  \tau\equiv -m^an^b\nabla_bl_a, \\
  \nu&\equiv \bar{m}^an^b\nabla_bn_a, \qquad \mu\equiv\bar{m}^am^b\nabla_bn_a,
  &\qquad
  \lambda&\equiv \bar{m}^a\bar{m}^b\nabla_bn_a, \qquad
  \pi\equiv \bar{m}^al^b\nabla_bn_a, \\
\alpha&\equiv
  \frac{1}{2}(l^a\bar{m}^b\nabla_bn_a-m^a\bar{m}^b\nabla_b\bar{m}_a),
  &\qquad \beta&\equiv \frac{1}{2}(\bar{m}^am^b\nabla_bm_a-n^am^b\nabla_bl_a), \\
\epsilon&\equiv
  \frac{1}{2}(\bar{m}^al^b\nabla_bm_a-n^al^b\nabla_bl_a), &\qquad 
  \gamma&\equiv \frac{1}{2}(l^an^b\nabla_bn_a-m^an^b\nabla_b\bar{m}_a).
\end{align*}

\subsubsection*{The directional covariant derivatives}

The directional covariant derivatives along the directions given by
the elements of the NP frame are given by
\begin{align*}
D\equiv l^a\nabla_a, \qquad \Delta\equiv n^a\nabla_a \qquad
\delta\equiv m^a\nabla_a, \qquad \bar{\delta}=\bar{m}^a\nabla_a.
\end{align*}

\subsubsection*{The commutators}

The NP directional covariant derivatives satisfy the commutator
relations
\begin{subequations}
\begin{align}
 (\Delta D - D\Delta) \psi &= \big( (\pink{\gamma}+\pink{\bar{\gamma}}) D +
   (\epsilon+\bar{\epsilon})\Delta -(\bar{\tau} + \pi)\delta
   -(\tau+\bar{\pi})\bar{\delta}\big)\psi, \label{NPCommutator1}\\
   (\delta D - D \delta) \psi &=\big( (\pink{\bar{\alpha}}
   +\pink{\beta} -\pink{\bar{\pi}})D
   +\pink{\kappa}\Delta -(\bar{\rho}+\epsilon-\bar{\epsilon})\delta
   -\sigma\bar{\delta}  \big)\psi, \label{NPCommutator2}\\
 (\delta \Delta -\Delta \delta)\psi &= \big( -\pink{\bar{\nu}} D +
   (\tau-\bar{\alpha}-\beta)\Delta + (\mu -\pink{\gamma} +\pink{\bar{\gamma}})
   \delta
   +\bar{\lambda} \bar{\delta}\big)\psi, \label{NPCommutator3}\\
 (\bar{\delta}\delta - \delta \bar{\delta})\psi &= \big(
   (\pink{\bar{\mu}}-\pink{\mu})D + (\pink{\bar{\rho}}-\pink{\rho}) \Delta +
   (\alpha-\bar{\beta})\delta -(\bar{\alpha}-\beta)\bar{\delta} \big)\psi.
   \label{NPCommutator4},
\end{align}
\end{subequations}
where~$\psi$ is any scalar field. Here we have highlighted the terms
which vanish in our gauge.

\subsubsection*{The components of the curvature}

The components of the Weyl tensor~$C_{abcd}$, trace-free Ricci
tensor~$\Phi_{ab}$ and the Ricci scalar~$R$, namely~\{$\Psi_0$,
$\Psi_1$, $\Psi_2$, $\Psi_3$, $\Psi_4$\}, \{$\Phi_{00}$, $\Phi_{01}$,
$\Phi_{02}$, $\Phi_{11}$, $\Phi_{12}$, $\Phi_{22}$\} and $\Lambda$ are
defined as
\begin{align*}
  \Psi_0&\equiv C_{abcd}l^am^bl^cm^d, &\qquad \Psi_1&\equiv
  C_{abcd}l^an^bl^cm^d, &\qquad  \Psi_2&\equiv \frac{1}{2}C_{abcd}l^an^b
  (l^cn^d-m^c\bar{m}^d)\\
  \Psi_3&\equiv C_{abcd}n^al^bn^c\bar{m}^d, &\qquad \Psi_4&\equiv
  C_{abcd}n^a\bar{m}^bn^c\bar{m}^d, &&\\
  \Phi_{00}&\equiv \frac{1}{2}R_{\{ab\}}l^al^b, &\qquad \Phi_{01}&\equiv
  \frac{1}{2}R_{\{ab\}}l^am^b,
   &\qquad \Phi_{02}&\equiv\frac{1}{2}R_{\{ab\}}m^am^b,\\
\Phi_{11}&\equiv \frac{1}{4}R_{\{ab\}}(l^an^b+m^a\bar{m}^b), &\qquad
\Phi_{12}&\equiv \frac{1}{2}R_{\{ab\}}n^am^b, &\qquad \Phi_{22}&\equiv
\frac{1}{2}R_{\{ab\}}n^an^b,\\
\Lambda&\equiv-\frac{R}{24} &&&&
\end{align*}
where the curly brackets denote the symmetric, trace-free part.

\subsubsection{The NP Ricci identities}

The NP Ricci identities (also known as the second structure equations)
take the form: 
\begin{subequations}
\begin{align}
  \Delta\epsilon-D \pink{\gamma}&= \pink{\Lambda}- \pink{\Phi_{11}}-\Psi_2
  +\epsilon( \pink{2\gamma+\bar{\gamma}})+ \pink{\gamma}\bar{\epsilon}
  + \pink{\kappa\nu}-\beta\pi-\alpha\bar{\pi}-\alpha\tau-\pi\tau
  -\beta\bar{\tau} \label{structureeq1},\\
  \Delta \pink{\kappa}-D\tau&=- \pink{\Phi_{01}}-\Psi_1+3 \pink{\gamma\kappa}
  + \pink{\bar{\gamma}\kappa}-\bar{\pi}\rho-\pi\sigma-\epsilon\tau
  +\bar{\epsilon}\tau-\rho\tau-\sigma\bar{\tau} \label{structureeq2},\\
  \Delta\pi-D \pink{\nu}&=- \pink{\Phi_{21}}-\Psi_3+3\epsilon \pink{\nu}
  +\bar{\epsilon} \pink{\nu}- \pink{\gamma}\pi+ \pink{\bar{\gamma}}\pi
  -\mu\pi-\lambda\bar{\pi}-\lambda\tau-\mu\bar{\tau} \label{structureeq3},\\
  \delta\pink{\gamma}-\Delta\beta&=\pink{\Phi_{12}}-\bar{\alpha}\pink{\gamma}
  -2\beta\pink{\gamma}+\beta\pink{\bar{\gamma}}+\alpha\bar{\lambda}+\beta\mu
  -\epsilon\pink{\bar{\nu}}-\pink{\nu}\sigma+\pink{\gamma}\tau+\mu\tau
  \label{structureeq4},\\
  \delta\epsilon-D\beta&=-\Psi_{1}+\bar{\alpha}\epsilon+\beta\bar{\epsilon}
  +\pink{\gamma\kappa}+\pink{\kappa}\mu-\epsilon\bar{\pi}-\beta\bar{\rho}
  -\alpha\sigma-\pi\sigma \label{structureeq5},\\
  \delta\pink{\kappa}-D\sigma&=-\Psi_{0}+\bar{\alpha}\pink{\kappa}
  +3\beta\pink{\kappa}-\pink{\kappa}\bar{\pi}-3\epsilon\sigma
  +\bar{\epsilon}\sigma-\rho\sigma-\bar{\rho}\sigma+\pink{\kappa}\tau
  \label{structureeq6},\\
  \delta\pink{\nu}-\Delta\mu&=\pink{\Phi_{22}}+\lambda\bar{\lambda}+\pink{\gamma}\mu
  +\pink{\bar{\gamma}}\mu+\mu^2-\bar{\alpha}\pink{\nu}-3\beta\pink{\nu}
  -\pink{\bar{\nu}}\pi+\pink{\nu}\tau \label{structureeq7},\\
  \delta\pi-D\mu&=-2\pink{\Lambda}-\Psi_2+\epsilon\mu+\bar{\epsilon}\mu
  +\pink{\kappa}\nu+\bar{\alpha}\pi-\beta\pi-\pi\bar{\pi}-\mu\bar{\rho}
  -\lambda\sigma \label{structureeq8},\\
  \delta\tau-\Delta\sigma&=\pink{\Phi_{02}}-\pink{\kappa\bar{\nu}}+\bar{\lambda}\rho
  -3\pink{\gamma}\sigma+\pink{\bar{\gamma}}\sigma+\mu\sigma-\bar{\alpha}\tau
  +\beta\tau+\tau^2 \label{structureeq9},\\
  \bar{\delta}\beta-\delta\alpha&=-\pink{\Lambda}-\pink{\Phi_{11}}+\Psi_2
  -\alpha\bar{\alpha}+2\alpha\beta-\beta\bar{\beta}-\epsilon\mu+\epsilon\bar{\mu}
  -\pink{\gamma}\rho-\mu\rho+\pink{\gamma}\bar{\rho}+\lambda\sigma
  \label{structureeq10},\\
  \bar{\delta}\pink{\gamma}-\Delta\alpha&=\Psi_3-\bar{\beta}\pink{\gamma}
  -\alpha\pink{\bar{\gamma}}+\beta\lambda+\alpha\bar{\mu}-\epsilon\pink{\nu}
  -\pink{\nu}\rho+\lambda\tau+\pink{\gamma}\bar{\tau} \label{structureeq11},\\
  \bar{\delta}\epsilon-D\alpha&=-\pink{\Phi_{10}}+2\alpha\epsilon+\bar{\beta}\epsilon
  -\alpha\bar{\epsilon}+\pink{\gamma\bar{\kappa}}+\pink{\kappa}\lambda-\epsilon\pi
  -\alpha\rho-\pi\rho-\beta\bar{\sigma} \label{structureeq12},\\
  \bar{\delta}\pink{\kappa}-D\rho&=-\pink{\Phi_{00}}+3\alpha\pink{\kappa}
  +\bar{\beta}\pink{\kappa}-\pink{\kappa}\pi-\epsilon\rho-\bar{\epsilon}\rho
  -\rho^2-\sigma\bar{\sigma}+\pink{\bar{\kappa}}\tau \label{structureeq13},\\
  \bar{\delta}\mu-\delta\lambda&=-\pink{\Phi_{21}}+\Psi_3-\bar{\alpha}\lambda
  +3\beta\lambda-\alpha\mu-\bar{\beta}\mu-\mu\pi+\bar{\mu}\pi-\pink{\nu}\rho
  +\pink{\nu}\bar{\rho} \label{structureeq14},\\
  \bar{\delta}\pink{\nu}-\Delta\lambda&=\Psi_4+3\pink{\gamma}\lambda
  -\pink{\bar{\gamma}}\lambda+\lambda\mu+\lambda\bar{\mu}-3\alpha\pink{\nu}
  -\bar{\beta}\pink{\nu}-\pink{\nu}\pi+\pink{\nu}\bar{\tau} \label{structureeq15},\\
  \bar{\delta}\pi-D\lambda&=-\pink{\Phi_{20}}+3\epsilon\lambda-\bar{\epsilon}\lambda
  +\pink{\bar{\kappa}\nu}-\alpha\pi+\bar{\beta}\pi-\pi^2-\lambda\rho-\mu\bar{\sigma}
  \label{structureeq16},\\
  \bar{\delta}\sigma-\delta\rho&=-\pink{\Phi_{01}}+\Psi_1-\pink{\kappa}\mu
  +\pink{\kappa}\bar{\mu}-\bar{\alpha}\rho-\beta\rho+3\alpha\sigma
  -\bar{\beta}\sigma-\rho\tau-\bar{\rho}\tau \label{structureeq17},\\
  \bar{\delta}\tau-\Delta\rho&=2\pink{\Lambda}+\Psi_2-\pink{\kappa\nu}
  -\pink{\gamma}\rho-\pink{\bar{\gamma}}\rho+\bar{\mu}\rho+\lambda\sigma
  +\alpha\tau-\bar{\beta}\tau+\tau\bar{\tau} \label{structureeq18}.
\end{align}
\end{subequations}
Observe that the above are the full NP Ricci identities. The
highlighted terms are those that vanish either because of the gauge or
vacuum conditions. Further simplification following from Stewart's
gauge~(Lemma~\ref{Lemma1}) are~$\rho=\bar{\rho},\mu=\bar{\mu}$
and~$\pi=\alpha+\bar{\beta}$. Note furthermore that the three
equations~\eqref{structureeq10}, \eqref{structureeq14},
\eqref{structureeq17}, there are constraints on the
2-spheres~$\mathcal{S}_{u,v}$.

\subsubsection*{The Bianchi identities}

Finally, the (second) Bianchi identities take the form:
\begin{subequations}
\begin{align}
  &\bar{\delta}\Psi_0-D\Psi_1+D\pink{\Phi_{01}}-\delta\pink{\Phi_{00}}
  =(4\alpha-\pi)\Psi_0
  -2(2\rho+\epsilon)\Psi_1+3\pink{\kappa}\Psi_2+(\bar{\pi}-2\bar{\alpha}
  -2\beta)\pink{\Phi_{00}}
  \nonumber \\
  &\qquad\qquad\qquad\qquad\qquad\qquad\quad+2(\epsilon
  +\bar{\rho})\pink{\Phi_{01}}+2\sigma\pink{\Phi_{10}}-2\pink{\kappa\Phi_{11}}
  -\pink{\bar{\kappa}\Phi_{02}}, \label{Bianchi1} \\
  &\Delta\Psi_0-\delta\Psi_1+D\pink{\Phi_{02}}-\delta\pink{\Phi_{01}}
  =(4\pink{\gamma}-\mu)\Psi_0-2(2\tau+\beta)\Psi_1+3\sigma\Psi_2
  -\bar{\lambda}\pink{\Phi_{00}}
  +2(\bar{\pi}-\beta)\pink{\Phi_{01}} \nonumber \\
  &\qquad\qquad\qquad\qquad\qquad\qquad\quad+2\sigma\pink{\Phi_{11}}
  +(\bar{\rho}+2\epsilon-2\bar{\epsilon})\pink{\Phi_{02}}
  -2\pink{\kappa\Phi_{12}},
  \label{Bianchi2} \\
  &\bar{\delta}\Psi_3-D\Psi_4+\bar{\delta}\pink{\Phi_{21}}-\Delta\pink{\Phi_{20}}
  =(4\epsilon-\rho)\Psi_4-2(2\pi+\alpha)\Psi_3+3\lambda\Psi_2
  +2\lambda\pink{\Phi_{11}}
  -2\pink{\nu\Phi_{10}}-\bar{\sigma}\pink{\Phi_{22}}\nonumber \\
  &\qquad\qquad\qquad\qquad\qquad\qquad\quad +(2\pink{\gamma}
  -2\pink{\bar{\gamma}}
  +\bar{\mu})\pink{\Phi_{20}}+2(\bar{\tau}-\alpha)\pink{\Phi_{21}},
  \label{Bianchi3} \\
  &\Delta\Psi_3-\delta\Psi_4+\bar{\delta}\pink{\Phi_{22}}-\Delta\pink{\Phi_{21}}
  =(4\beta-\tau)\Psi_4-2(2\mu+\pink{\gamma})\Psi_3+3\pink{\nu}\Psi_2
  +2\lambda\pink{\Phi_{12}}
  -2\pink{\nu\Phi_{11}}-\pink{\bar{\nu}\Phi_{20}} \nonumber \\
  &\qquad\qquad\qquad\qquad\qquad\qquad\quad+(\bar{\tau}-2\bar{\beta}-2\alpha)
  \pink{\Phi_{22}}
  +2(\pink{\gamma}+\bar{\mu})\pink{\Phi_{21}} ,\label{Bianchi4} \\
  &D\Psi_2-\bar{\delta}\Psi_1+\Delta\pink{\Phi_{00}}-\bar{\delta}\pink{\Phi_{01}}
  +2D\pink{\Lambda}
  =-\lambda\Psi_0+2(\pi-\alpha)\Psi_1+3\rho\Psi_2-2\pink{\kappa}\Psi_3
  -2\tau\pink{\Phi_{10}}
  +2\rho\pink{\Phi_{11}} \nonumber \\
  &\qquad\qquad\qquad\qquad\qquad\qquad\qquad\quad\quad+\bar{\sigma}\pink{\Phi_{02}}
  +(2\pink{\gamma}+2\pink{\bar{\gamma}}-\bar{\mu})\pink{\Phi_{00}}
  -2(\bar{\tau}+\alpha)\pink{\Phi_{01}}, \label{Bianchi5} \\
  &\Delta\Psi_2-\delta\Psi_3+D\pink{\Phi_{22}}-\delta\pink{\Phi_{21}}
  +2\Delta\pink{\Lambda }
  =\sigma\Psi_4+2(\beta-\tau)\Psi_3-3\mu\Psi_2+2\pink{\nu}\Psi_1+2\pi\pink{\Phi_{12}}
  -2\mu\pink{\Phi_{11}}\nonumber \\
  & \qquad\qquad\qquad\qquad\qquad\qquad\qquad\quad\quad
  -\bar{\lambda}\pink{\Phi_{20}}+(\bar{\rho}-2\epsilon-2\bar{\epsilon})\pink{\Phi_{22}}
  +2(\bar{\pi}+\beta)\pink{\Phi_{21}},\label{Bianchi6} \\
  &D\Psi_3-\bar{\delta}\Psi_2-D\pink{\Phi_{21}}+\delta\pink{\Phi_{20}}
  -2\bar{\delta}\pink{\Lambda}=-\pink{\kappa}\Psi_4+2(\rho-\epsilon)\Psi_3
  +3\pi\Psi_2-2\lambda\Psi_1-2\pi\pink{\Phi_{11}}+2\mu\pink{\Phi_{10}}
  \nonumber \\
  & \qquad\qquad\qquad\qquad\qquad\qquad\qquad\quad\quad
  +\pink{\bar{\kappa}\Phi_{22}}+(2\bar{\alpha}-2\beta-\bar{\pi})\pink{\Phi_{20}}
  -2(\bar{\rho}-\epsilon)\pink{\Phi_{21}} ,\label{Bianchi7} \\
  &\Delta\Psi_1-\delta\Psi_2-\Delta\pink{\Phi_{01}}+\bar{\delta}\pink{\Phi_{02}}
  -2\delta\pink{\Lambda}=\pink{\nu}\Psi_0+2(\pink{\gamma}-\mu)\Psi_1-3\tau\Psi_2
  +2\sigma\Psi_3
  +2\tau\pink{\Phi_{11}} -2\rho\pink{\Phi_{12}}\nonumber \\
  & \qquad\qquad\qquad\qquad\qquad\qquad\qquad\quad\quad
  -\pink{\bar{\nu}\Phi_{00}}+(\bar{\tau}-2\bar{\beta}+2\alpha)\pink{\Phi_{02}}
  +2(\bar{\mu}-\pink{\gamma})\pink{\Phi_{01}}, \label{Bianchi8} \\
  &D\pink{\Phi_{11}}-\delta\pink{\Phi_{10}}-\bar{\delta}\pink{\Phi_{01}}
  +\Delta\pink{\Phi_{00}}+3D\pink{\Lambda}=(2\pink{\gamma}-\mu+2\pink{\bar{\gamma}}
  -\bar{\mu})\pink{\Phi_{00}}+(\pi-2\alpha-2\bar{\tau})\pink{\Phi_{01}}
  +\bar{\sigma}\pink{\Phi_{02}} \nonumber \\
  & \qquad\qquad\qquad\qquad\qquad\qquad\qquad\quad\quad\quad
  +\sigma\pink{\Phi_{20}}+(\bar{\pi}-2\bar{\alpha}-2\tau)\pink{\Phi_{10}}
  +2(\rho+\bar{\rho})\pink{\Phi_{11}}-\pink{\bar{\kappa}\Phi_{12}}\nonumber\\
  & \qquad\qquad\qquad\qquad\qquad\qquad\qquad\quad\quad\quad
  -\pink{\kappa\Phi_{21}}, \label{Bianchi9} \\
  &D\pink{\Phi_{12}}-\delta\pink{\Phi_{11}}-\bar{\delta}\pink{\Phi_{02}}
  +\Delta\pink{\Phi_{01}}
  +3\delta\pink{\Lambda}=(-2\alpha+2\bar{\beta}+\pi-\bar{\tau})\pink{\Phi_{02}}
  +(\bar{\rho}+2\rho-2\bar{\epsilon})\pink{\Phi_{12}}+\pink{\bar{\nu}\Phi_{00}}
  \nonumber \\
  & \qquad\qquad\qquad\qquad\qquad\qquad\qquad\quad\quad\quad
  -\bar{\lambda}\pink{\Phi_{10}}+2(\bar{\pi}-\tau)\pink{\Phi_{11}}
  +(2\pink{\gamma}-2\bar{\mu}-\mu)\pink{\Phi_{01}}+\sigma\pink{\Phi_{21}}\nonumber\\
  & \qquad\qquad\qquad\qquad\qquad\qquad\qquad\quad\quad\quad
  -\pink{\kappa\Phi_{22}}, \label{Bianchi10} \\
  &D\pink{\Phi_{22}}-\delta\pink{\Phi_{21}}-\bar{\delta}\pink{\Phi_{12}}
  +\Delta\pink{\Phi_{11}}
  +3\Delta\pink{\Lambda}=(\rho+\bar{\rho}-2\epsilon-2\bar{\epsilon})\pink{\Phi_{22}}
  +(2\bar{\beta}+2\pi-\bar{\tau})\pink{\Phi_{12}}+\pink{\nu\Phi_{01}}
  \nonumber \\
  & \qquad\qquad\qquad\qquad\qquad\qquad\qquad\quad\quad\quad
  +\pink{\bar{\nu}\Phi_{10}}+(2\beta+2\bar{\pi}-\tau)\pink{\Phi_{21}}
  -2(\mu+\bar{\mu})\pink{\Phi_{11}}
  -\bar{\lambda}\pink{\Phi_{20}}\nonumber\\
  & \qquad\qquad\qquad\qquad\qquad\qquad\qquad\quad\quad\quad
  -\lambda\pink{\Phi_{02}}. \label{Bianchi11}
\end{align}
\end{subequations}
As in the case of the Ricci identities we have highlighted the
vanishing terms. Note that the last three equations make no
contribution to our analysis as they are satisifed identically.

\section{Inequalities}\label{App:Inequalities}

In this appendix, as a quick reference, we list the key inequalities
which are used routinely in our analysis. These inequalities are
standard and proofs can be found, e.g. in~\cite{Eva98}.

\smallskip
\noindent
\textbf{Cauchy-Schwarz inequality. } 
If~$u_1, ... , u_n\in\mathbb{C}$ and $v_1, ... , v_n\in\mathbb{C}$, we have
\begin{align*}
|u_1v_1+...+u_nv_n|^2\leq(|u_1|^2+...+|u_n|^2)(|v_1|^2+...+|v_n|^2).
\end{align*}

\smallskip
\noindent
\textbf{Gr\"onwall's inequality.} If~$\beta(t)$ is a non-negative
continuous function and $u(t)$ satisfies
\begin{align*}
  u(t)\leq\alpha(t)+\int_a^t\beta(s)u(s)\mathrm{d}s,
  \ \ \forall t\in[a,\ b],
\end{align*}
then 
\begin{align*}
  u(t)\leq\alpha(t)+\int_a^t\alpha(s)\beta(s)
  \exp\left(\int_s^t\beta(r)\mathrm{d}r\right)\mathrm{d}s,
  \ \ t\in[a,\ b].
\end{align*}
In addition, if the function~$\alpha$ is non-decreasing, then
\begin{align*}
  u(t)\leq\alpha(t)\exp\left(\int_a^t\beta(s)\mathrm{d}s\right),
  \ \ t\in[a,\ b].
\end{align*}
Moreover, if~$\beta\equiv C$ where~$C$ is a positive constant, then
\begin{align*}
u(t)\leq C(b-a)\alpha(t).
\end{align*}

\smallskip
\noindent
\textbf{Young's inequality.} If~$a$ and~$b$ are non negative real
numbers and~$p$ and~$q$ are positive real numbers such
that~$1/p+1/q=1$, then
\begin{align*}
ab\leq\frac{a^p}{p}+\frac{b^q}{q}.
\end{align*}
The equality holds if and only if~$a^p=b^q$. Moreover, if~$a$ and~$b$
are non negative real numbers and~$p\geq1$, then
\begin{align*}
a^p+b^p\leq(a+b)^p.
\end{align*}
Finally, if~$f(x)$ is non-negative continuous function and $p\geq1$,
then
\begin{align*}
\int_Kf^p\leq\left(\int_Kf\right)^p,
\end{align*}
where~$K$ is a compact set.

\smallskip
\noindent
\textbf{Generalised H\"older's inequality.} Let~$K$ be a measurable
space. Assume~$f\in L^p(K)$ and $g\in L^q(K)$ with~$1\leq p,
q\leq\infty$ and~$1/r=1/p+1/q\leq1$, then
\begin{align*}
||fg||_{L^r(K)}\leq||f||_{L^p(K)}||g||_{L^q(K)}.
\end{align*}

\smallskip
\noindent
\textbf{Gagliardo-Nirenberg-Sobolev inequality.} Let~$U$ be a bounded,
open subset of~$\mathbb{R}^n$, and assume~$\p U$ is~$C^1$. Let~$1\leq
p<n$, and suppose that~$u\in W^{1,p}(U)$. Then~$u\in L^{p*}(U)$, with
the estimate,
\begin{align*}
||u||_{L^{p*}(U)}\leq C||u||_{W^{1,p}(U)}
\end{align*}
the constant~$C$ depending only on~$p$, $n$ and~$U$
and~$1/p+1/p^*=1/n$.

\section{Angular derivatives of a scalar function}\label{nablaf}

In our analysis we make repeated use of properties of the angular
derivatives of a scalar field over the 2-spheres~$\mathcal{S}_{u,v}$
of constant $u$, $v$. In the following let~$f:
\mathcal{S}_{u,v}\rightarrow \mathbb{C}$ denote a sufficiently smooth
complex scalar field.

\subsubsection*{Definitions and basic inequalities}

In terms of the NP vectors~$m^a$ and~$\bar{m}^a$ one has that
\begin{align*}
 \nablasl_af&=-m_a\bar m^b\nablasl_bf-\bar m_am^b\nablasl_bf
 =-m_a\bar\delta f-\bar m_a\delta f.
\end{align*}
Moreover, we have that
\begin{align*}
  |\nablasl f|^2\equiv -\sigma^{ab}\overline{\nablasl_af}
  \nablasl_bf=\bar\delta f\delta\bar f+\bar\delta\bar f\delta f.
\end{align*}
A direct computation shows that,
\begin{align*}
  ||\nablasl f||_{L^p(\mathcal{S}_{u,v})}&=\left(\int_{\mathcal{S}_{u,v}}|
  \bar\delta f\delta\bar f+\bar\delta\bar f\delta f|^{p/2}
  \right)^{1/p}=|||\delta f|^2+|\bar
  \delta f|^2||^{1/2}_{L^{p/2}(\mathcal{S}_{u,v})} \\
  &\leq\left(|||\delta
  f|^2||_{L^{p/2}(\mathcal{S}_{u,v})}+|||\bar\delta
  f|^2||_{L^{p/2}(\mathcal{S}_{u,v})}\right)^{1/2}\\
  & \leq|||\delta f|^2||_{L^{p/2}(\mathcal{S}_{u,v})}^{1/2}
  +|||\bar\delta f|^2||_{L^{p/2}(\mathcal{S}_{u,v})}^{1/2}  \\
  &=||\delta f||_{L^{p}(\mathcal{S}_{u,v})}
  +||\bar\delta f||_{L^{p}(\mathcal{S}_{u,v})}.
\end{align*}
Conversely, we have
\begin{align*}
  ||\delta f||_{L^{p}(\mathcal{S}_{u,v})},
  \quad ||\bar\delta f||_{L^{p}(\mathcal{S}_{u,v})} \leq
  \left(\int_{\mathcal{S}_{u,v}}|\bar\delta f\delta\bar f
  +\bar\delta\bar f\delta
  f|^{p/2} \right)^{1/p} \leq ||\nablasl
f||_{L^p(\mathcal{S}_{u,v})}. 
\end{align*}
Thus, we can estimate~$\nablasl f$ in terms of~$\delta f$
and~$\bar\delta f$ and vice versa. \emph{This observation is used
  repeatedly in the main text.}

\subsubsection*{The Hessian}

The Hessian~$\nablasl_a\nablasl_bf$ of the
scalar function~$f$ can be expanded in terms of NP objects as
\begin{align*}
  \nablasl_a\nablasl_bf&=\left(\bar\delta\bar\delta f
  +(\bar\beta-\alpha)\bar\delta f \right)m_am_b
  +\left(\delta\delta f+(\beta-\bar\alpha)\delta f \right)
  \bar m_a\bar m_b \\
  & +\left(\bar\delta\delta f+(\alpha-\bar\beta)\delta f \right)
  m_a\bar m_b+\left(\delta\bar\delta f
  +(\bar\alpha-\beta)\bar\delta f \right)\bar m_am_b,
\end{align*}
where we have made use of the expansion
\begin{align*}
\nablasl_am_b=(\alpha-\bar\beta)m_am_b+(\beta-\bar\alpha)\bar
m_am_b.
\end{align*}
Defining, for convenience, the scalars 
\begin{align*}
&& T_1\equiv \bar\delta\bar\delta
f+(\bar\beta-\alpha)\bar\delta f, \qquad T_2 \equiv \bar\delta\delta
f+(\alpha-\bar\beta)\delta f, \\
&& T_3 \equiv \delta\bar\delta
f+(\bar\alpha-\beta)\bar\delta f, \qquad T_4\equiv \delta\delta
f+(\beta-\bar\alpha)\delta f,
\end{align*}
one can then write
\begin{align*}
  |\nablasl^2f|^2\equiv \sigma^{ab}\sigma^{cd}\overline{\nablasl_a\nablasl_cf}
  \nablasl_b\nablasl_df=|T_1|^2+|T_2|^2+|T_3|^2+|T_4|^2.
\end{align*}
Making use of the above decomposition we then have that 
\begin{align} 
  ||\nablasl^2
  f||_{L^p(\mathcal{S}_{u,v})}&=\left(\int_{\mathcal{S}_{u,v}}
  (|T_1|^2+|T_2|^2+|T_3|^2+|T_4|^2)^{p/2}\right)^{1/p}\leq\sum_{i=1}^4
  ||T_i||_{L^{p}(\mathcal{S}_{u,v})} \nonumber \\
  &\leq||\delta^2f||_{L^{p}(\mathcal{S}_{u,v})}+||\bar\delta^2f
  ||_{L^{p}(\mathcal{S}_{u,v})}+||\delta\bar\delta
  f||_{L^{p}(\mathcal{S}_{u,v})}+||\bar\delta\delta
  f||_{L^{p}(\mathcal{S}_{u,v})} \nonumber \\
  &\quad+4\Delta_{\Gamma}(||\delta f||_{L^{p}(\mathcal{S}_{u,v})}
  +||\bar\delta f||_{L^{p}(\mathcal{S}_{u,v})}), \label{nabla2f}
\end{align}
where~$\Delta_{\Gamma}$ is defined as in the main text. Also, observe
that~$||\nablasl^2 f||_{L^p(\mathcal{S}_{u,v})}$ is not smaller than
any of the individual terms in the right side of the first
inequality~\eqref{nabla2f}.

A final observation following the \emph{irreducible decomposition}
\begin{align} 
\nablasl_a\nablasl_bf = \nablasl_{\{a}\nablasl_{b\}}f +
\frac{1}{2}\sigma_{ab} \Deltasl f + \nablasl_{[a}\nablasl_{b]}f
\end{align} 
of the Hessian, where the curly brackets denote the
symmetric-tracefree part with respect to the metric $\sigma_{ab}$, is
that
\begin{align} 
| \nablasl_a\nablasl_bf |^2 = |\nablasl_{\{a}\nablasl_{b\}}f|^2 +
\frac{1}{2}|\Deltasl f|^2 + |\nablasl_{[a}\nablasl_{b]}f|^2,
\end{align} 
so that
\begin{align}
|\Deltasl f|^2 \leq 2 | \nablasl_a\nablasl_bf |^2
\label{InequalityLaplacian}
\end{align}

\subsubsection*{Third derivatives of a scalar field}

As in the main text denote by~$\varpi\equiv \beta-\bar\alpha$ the
simple independent component of the connection of
the~$2$-sphere~$\mathcal{S}_{u,v}$. It follows from the from the
structure equation~\eqref{structureeq10} and its complex conjugate,
that the Gaussian curvature curvature
\begin{align*}
K\equiv 2\varpi\bar\varpi+2\delta\bar\varpi+2\bar\delta\varpi
\end{align*}
satisfies the relation
\begin{align*}
  K=\sigma\lambda+\bar\sigma\bar\lambda-\rho\mu-\bar\rho\bar\mu
  +\Psi_2+\bar\Psi_2,
\end{align*}
see~\cite{PenRin84} for details.

Now, the third order covariant derivative of~$f$
on~$\mathcal{S}_{u,v}$ can be expanded as
\begin{align*}
  \nablasl_a\nablasl_b\nablasl_cf&=M_1m_am_bm_c
  +M_5\bar{m}_a\bar{m}_b\bar{m}_c
  +M_2\bar{m}_am_bm_c+M_6m_a\bar{m}_b\bar{m}_c \\
  &+M_3m_am_b\bar{m}_c+M_7\bar{m}_a\bar{m}_bm_c
  +M_4\bar{m}_am_b\bar{m}_c+M_8m_a\bar{m}_bm_c,
\end{align*}
where,
\begin{align*}
  M_1&\equiv -(\bar\delta^3f+3\bar\varpi\bar\delta^2f
  +\bar\delta\bar\varpi\bar\delta f+2\bar\varpi^2\bar\delta f), \\
  M_2&\equiv-\delta\bar\delta^2f-\bar\varpi\delta\bar\delta f
  +2\varpi\bar\delta^2f-\delta\bar\varpi\bar\delta f
  +2\varpi\bar\varpi\bar\delta f, \\
  M_3&\equiv-\bar\delta^2\delta f+\bar\varpi\bar\delta\delta f
  +\bar\delta\bar\varpi\delta f,\\
  M_4&\equiv-\delta\bar\delta\delta f+\bar\varpi\delta^2f
  +\delta\bar\varpi\delta f, \\
  M_5&\equiv-(\delta^3f+3\varpi\delta^2f+\delta\varpi\delta f
  +2\varpi^2\delta f), \\
  M_6&\equiv-\bar\delta\delta^2f-\varpi\bar\delta\delta f
  +2\bar\varpi\delta^2f-\bar\delta\varpi\delta f+2\varpi\bar\varpi\delta f, \\
  M_7&\equiv-\delta^2\bar\delta f+\varpi\delta\bar\delta f
  +\delta\varpi\bar\delta f,\\
  M_8&\equiv-\bar\delta\delta\bar\delta f+\varpi\bar\delta^2f
  +\bar\delta\varpi\bar\delta f.
\end{align*}
It follows then that, 
\begin{align*}
|\nablasl^3f|^2=\sum_{i=1}^8|M_i|^2.
\end{align*}
From the above expression one finds that 
\begin{align*}\label{nabla3f}
  ||\nablasl^3 f||_{L^p(\mathcal{S}_{u,v})}&\leq
  ||\delta^3f||_{L^p(\mathcal{S}_{u,v})}+||\bar\delta^3f||_{L^p(\mathcal{S}_{u,v})}
  +||\delta^2\bar\delta f||_{L^p(\mathcal{S}_{u,v})}
  +||\delta\bar\delta^2f||_{L^p(\mathcal{S}_{u,v})} \\
  &\quad+||\bar\delta^2\delta f||_{L^p(\mathcal{S}_{u,v})}
  +||\bar\delta\delta^2f||_{L^p(\mathcal{S}_{u,v})}+||\delta\bar\delta\delta
  f||_{L^p(\mathcal{S}_{u,v})}+||\bar\delta\delta\bar\delta f||_{L^p(\mathcal{S}_{u,v})} \\
  &\quad+3||\bar\varpi\bar\delta^2f||_{L^p(\mathcal{S}_{u,v})}
  +||\bar\varpi\delta\bar\delta f||_{L^p(\mathcal{S}_{u,v})}
  +2||\varpi\bar\delta^2f||_{L^p(\mathcal{S}_{u,v})}+||\bar\varpi\bar\delta\delta
  f||_{L^p(\mathcal{S}_{u,v})}\\
  &\quad+||\bar\varpi\delta^2f||_{L^p(\mathcal{S}_{u,v})}
  +3||\varpi\delta^2f||_{L^p(\mathcal{S}_{u,v})}+||\varpi\bar\delta\delta
  f||_{L^p(\mathcal{S}_{u,v})}+2||\bar\varpi\delta^2f||_{L^p(\mathcal{S}_{u,v})}\\ 
  &\quad  +||\varpi\delta\bar\delta f||_{L^p(\mathcal{S}_{u,v})}
  +||\varpi\bar\delta^2f||_{L^p(\mathcal{S}_{u,v})}+||\bar\delta\bar\varpi\bar\delta
  f||_{L^p(\mathcal{S}_{u,v})}
  +||\delta\bar\varpi\bar\delta f||_{L^p(\mathcal{S}_{u,v})}\\
  &\quad  +||\bar\delta\bar\varpi\delta f||_{L^p(\mathcal{S}_{u,v})}
  +||\delta\bar\varpi\delta f||_{L^p(\mathcal{S}_{u,v})}
  +||\delta\varpi\delta f||_{L^p(\mathcal{S}_{u,v})}+||\bar\delta\varpi\delta
  f||_{L^p(\mathcal{S}_{u,v})}  \\
  &\quad+||\delta\varpi\bar\delta
  f||_{L^p(\mathcal{S}_{u,v})}+||\bar\delta\varpi\bar\delta
  f||_{L^p(\mathcal{S}_{u,v})}+2||\bar\varpi^2\bar\delta
  f||_{L^p(\mathcal{S}_{u,v})}+2||\varpi\bar\varpi\bar\delta
  f||_{L^p(\mathcal{S}_{u,v})}\nonumber\\
  &\quad+2||\varpi^2\delta
  f||_{L^p(\mathcal{S}_{u,v})}+2||\varpi\bar\varpi^2\delta f||_{L^p(\mathcal{S}_{u,v})}. 
\end{align*}
The above expression contains four representative terms,
namely~$||\delta^3f||_{L^p(\mathcal{S}_{u,v})}$,
$||\varpi\delta^2f||_{L^p(\mathcal{S}_{u,v})}$, $||\delta\varpi\delta
f||_{L^p(\mathcal{S}_{u,v})}$ and $||\varpi^2\delta
f||_{L^p(\mathcal{S}_{u,v})}$ which will be used to illustrate the
analysis in the main text.

\section{Integration Identities}\label{IntegralIdentity}

In this appendix we prove some integration identities which are
routinely used in the main text.

\smallskip
First we observe that a direct calculation yields
\begin{align*}
  \nablasl_{\mathcal{A}}P^{\mathcal{A}}&=\frac{1}{\sqrt{\det\bmsigma}}
  \p_{\mathcal{A}}(\sqrt{\det\bmsigma}P^{\mathcal{A}})
  =\frac{P^{\mathcal{A}}}{\sqrt{\det\bmsigma}}\p_{\mathcal{A}}
  \sqrt{\det\bmsigma}+\p_{\mathcal{A}}P^{\mathcal{A}}\\
  &=\frac{1}{2}P^{\mathcal{A}}\sigma^{\mathcal{B}\mathcal{C}}\p_{\mathcal{A}}
  \sigma_{\mathcal{B}\mathcal{C}}+\p_{\mathcal{A}}P^{\mathcal{A}}
  =-\frac{1}{2}P^{\mathcal{A}}\sigma_{\mathcal{B}\mathcal{C}}
  \p_{\mathcal{A}}\sigma^{\mathcal{B}\mathcal{C}}+\p_{\mathcal{A}}P^{\mathcal{A}}\\
  &=P^{\mathcal{A}}P^{\mathcal{B}}\sigma_{\mathcal{B}\mathcal{C}}\p_{\mathcal{A}}
  \bar P^{\mathcal{C}}+P^{\mathcal{A}}\bar P^{\mathcal{C}}
  \sigma_{\mathcal{B}\mathcal{C}}\p_{\mathcal{A}}P^{\mathcal{B}}-\bar m_{\mathcal{C}}
  \delta P^{\mathcal{C}}-m_{\mathcal{C}}\bar\delta P^{\mathcal{C}}, \\
  &=\sigma_{\mathcal{B}\mathcal{C}}P^{\mathcal{B}}(\delta\bar
  P^{\mathcal{C}}-\bar\delta P^{\mathcal{C}})=\varpi.
\end{align*}
In the last step we have made use of
equation~\eqref{framecoefficient5}. Consequently, we also have
that~$\nablasl_{\mathcal{A}}\bar P^{\mathcal{A}}=\bar\varpi$. Making
use of these results we further compute on the arbitrary
sphere~$\mathcal{S}$ that
\begin{align*}
  \int_{\mathcal{S}}|\delta f|^2&=\int_{\mathcal{S}} \delta f\bar\delta\bar
  f\sqrt{\det\bmsigma}\mathrm{d}^2x=\int_{\mathcal{S}}\delta
  f\bar P^{\mathcal{A}}\p_{\mathcal{A}}\bar f\sqrt{\det\bmsigma}\mathrm{d}^2x,\\
  &=-\int_{\mathcal{S}} \bar f \nablasl_{\mathcal{A}}
  (\delta f\bar P^{\mathcal{A}})
  =-\int_{\mathcal{S}}\bar f\bar\delta\delta f-\int_{\mathcal{S}}\bar f
  \nablasl_{\mathcal{A}}\bar P^{\mathcal{A}}\delta f.
\end{align*}
On the one hand, the first integral in the last equality can be
further expanded as
\begin{align*}
  \int_{\mathcal{S}}\bar f\bar\delta\delta f
  &=\int_{\mathcal{S}}\bar f\delta\bar\delta f-\int_{\mathcal{S}}\bar
  f\bar\varpi\delta f+\int_{\mathcal{S}}\bar f\varpi\bar\delta f \\
  &=\int_{\mathcal{S}}\bar
  fP^{\mathcal{A}}\p_{\mathcal{A}}\bar\delta
  f-\int_{\mathcal{S}}\bar f\bar\varpi
  P^{\mathcal{A}}\p_{\mathcal{A}}f+\int_{\mathcal{S}}\bar
  f\varpi\bar\delta f \\
  &=-\int_{\mathcal{S}}\bar\delta f \nablasl_{\mathcal{A}}(\bar fP^{\mathcal{A}})
  +\int_{\mathcal{S}}f\ \nablasl_{\mathcal{A}}(\bar f\bar\varpi P^{\mathcal{A}})
  +\int_{\mathcal{S}}\bar f\varpi\bar\delta f \\
  &=-\int_{\mathcal{S}}|\bar\delta f|^2+\int_{\mathcal{S}}|f|^2|\varpi|^2
  +\int_{\mathcal{S}}|f|^2\delta\bar\varpi
  +\int_{\mathcal{S}}f\bar\varpi\delta\bar f .
\end{align*}
On the other hand, the second integral can be expanded as
\begin{align*}
  \int_{\mathcal{S}}\bar f \nablasl_{\mathcal{A}}\bar P^{\mathcal{A}}\delta f
  &=\int_{\mathcal{S}}\bar f\bar\varpi\delta f=\int_{\mathcal{S}}\bar f\bar\varpi
  P^{\mathcal{A}}\nablasl_{\mathcal{A}}f=-\int_{\mathcal{S}}f \nablasl_{\mathcal{A}}
  (\bar f\bar\varpi P^{\mathcal{A}}),\\
  &=-\int_{\mathcal{S}} f\bar\varpi
  P^{\mathcal{A}}\nablasl_{\mathcal{A}}\bar
  f-\int_{\mathcal{S}}
  |f|^2P^{\mathcal{A}}\nablasl_{\mathcal{A}}\bar\varpi-\int_{\mathcal{S}}
  |f|^2\bar\varpi\ \nablasl_{\mathcal{A}}P^{\mathcal{A}} \\
  &=-\int_{\mathcal{S}} f\bar\varpi\delta\bar f-\int_{\mathcal{S}}
  |f|^2\delta\bar\varpi-\int_{\mathcal{S}} |f|^2|\varpi|^2.
\end{align*}
Combining these last expressions one finds that 
\begin{align*}
\int_{\mathcal{S}}|\delta f|^2=\int_{\mathcal{S}}|\bar\delta f|^2.
\end{align*}
In other words we have found that 
\begin{align*}
||\delta f||_{L^2(\mathcal{S}_{u,v})}=||\bar\delta f||_{L^{2}(\mathcal{S}_{u,v})}.
\end{align*}

\section{Details in Propositions \ref{PropositionSecondEstimateConnection}
  and \ref{Proposition:ThirdEstimateConnection}}
\label{Appendix:Propositions8-9}

In this appendix we provide further details regarding the lengthy
computations arising in the analysis of
Propositions~\ref{PropositionSecondEstimateConnection}
and~\ref{Proposition:ThirdEstimateConnection}.

\subsubsection*{Estimates on the $L^4$-norm of connection
  coefficients}

In the following we consider, for conciseness, the NP spin connection
coefficient~$\lambda$. Making use of
Proposition~\ref{Proposition:TransportLpEstimates} to
estimate~$||\lambda||_{L^4(S)}$ one finds that
\begin{align*}
  ||\nablasl\lambda||_{L^4(\mathcal{S}_{u,v})}\leq 2
  \left(||\nablasl\lambda||_{L^4(\mathcal{S}_{0,v})}
  +C(\Delta_{e_\star}, \Delta_{\Gamma_\star})\int_0^u
  \left(\int_{\mathcal{S}_{u',v}}\Delta\left\langle\nablasl\lambda,
  \nablasl\lambda\right\rangle_{\bmsigma}^2\right)^{1/4} \mathrm{d}u'\right).
\end{align*}
One can then estimate 
\begin{align*}
  \int_{\mathcal{S}_{u',v}}|\Delta\left\langle\nablasl\lambda,\nablasl\lambda
  \right\rangle_{\bmsigma}^2|&=\int_{\mathcal{S}_{u',v}}|\nablasl\lambda|^2
  |\Delta(\bar\delta\lambda\delta\bar\lambda+\bar\delta\bar\lambda
  \delta\lambda)| \\
  &=\int_{\mathcal{S}_{u',v}}|\nablasl\lambda|^2|(\Delta\delta\lambda)\bar
  \delta\bar\lambda+\delta\lambda\Delta\bar\delta\bar\lambda
  +\delta\bar\lambda\Delta\bar\delta\lambda+\bar\delta\lambda
  \Delta\delta\bar\lambda| \\
  &\leq\int_{S_{u',v}}|\nablasl\lambda|^2\sqrt{2|\delta\lambda|^2
  +2|\bar\delta\lambda|^2}\sqrt{2|\Delta\delta\lambda|^2+2|\Delta
  \bar\delta\lambda|^2} \\
  &\leq2\int_{\mathcal{S}_{u',v}}|\nablasl\lambda|^3\left(|\Delta\delta\lambda|
  +|\Delta\bar\delta\lambda|\right),
\end{align*}
where we have made use of the Cauchy-Schwarz inequality in the first
inequality. Now, making use of the expressions
for~$\Delta\delta\lambda$ and~$\Delta\bar\delta\lambda$ one further
finds that,
\begin{align*}
  &\int_{\mathcal{S}_{u',v}}|\Delta\left\langle\nablasl\lambda,\nablasl\lambda
  \right\rangle_{\bmsigma}^2|\leq 2\int_{\mathcal{S}_{u',v}}|\nablasl\lambda|^3
  \left(|\Gamma|^3+|\Gamma||\Psi_4|+|\Gamma'||\nablasl\lambda|
  +4|\lambda||\nablasl\mu|+|\nablasl\Psi_4| \right)\\
  &\qquad\leq
  C(I,\Delta_{e_{\star}},\Delta_{\Gamma_\star},\Delta_{\Psi}(\mathcal{S}))\left(
  \int_{\mathcal{S}_{u',v}}|\nablasl\lambda|^3
  +||\Psi_4||_{L^{\infty}(\mathcal{S}_{u',v})}\int_{\mathcal{S}_{u',v}}|
  \nablasl\lambda|^3 \right)\\
  &\qquad\quad+C(\Delta_{\Gamma_\star})
  \int_{\mathcal{S}_{u',v}}|\nablasl\lambda|^4
  +C(\Delta_{\Gamma_\star})\int_{\mathcal{S}_{u',v}}|\nablasl\lambda|^3
  |\nablasl\mu|
  +2||\nablasl\Psi_4||_{L^{\infty}(\mathcal{S}_{u',v})}\int_{\mathcal{S}_{u',v}}|
  \nablasl\lambda|^3 \\
  &\qquad\leq
  C(I,\Delta_{e_{\star}},\Delta_{\Gamma_\star},\Delta_{\Psi}(\mathcal{S}))
  \mbox{Area}(\mathcal{S}_{u',v})^{1/4}
  ||\nablasl\lambda||_{L^4(\mathcal{S}_{u',v})}^3
  \left(1+\Big(\sum_{i=0}^2
  ||\nablasl^i\Psi_4||_{L^2(\mathcal{S}_{u',v})}\Big)\right)\\
  &\qquad\quad+C(\Delta_{\Gamma_\star})
  ||\nablasl\lambda||_{L^4(\mathcal{S}_{u',v})}^3
  ||\nablasl\mu||_{L^4(\mathcal{S}_{u',v})}+C(\Delta_{\Gamma_\star})
  ||\nablasl\lambda||_{L^4(\mathcal{S}_{u',v})}^4 \\
  &\qquad\quad+C(\Delta_{\Gamma_\star}) \mbox{Area}(\mathcal{S}_{u',v})^{1/4}
  ||\nablasl\lambda||_{L^4(\mathcal{S}_{u',v})}^3\Big(\sum_{i=1}^3
  ||\nablasl^i\Psi_4||_{L^2(\mathcal{S}_{u',v})}\Big),
\end{align*}
where in the previous chain of inequalities we have made use of
H\"older's inequality and the Sobolev's embedding. Moreover,
here~$\Gamma$ represents a linear combination of the NP spin
connection coefficients~$\tau, \alpha, \beta, \mu, \lambda$
whereas~$\Gamma'$ contains no~$\tau$ term, which allows the use of
sharper estimates. Both~$\Gamma$ and~$\Gamma'$ are controlled
in~$L^{\infty}(\mathcal{S}_{u',v})$ as a result of
Proposition~\ref{Proposition:FirstEstimateConnection}.

Making use of the latter estimate and of the bootstrap assumption in
Proposition~\ref{PropositionSecondEstimateConnection}, one readily
obtains that
\begin{align*}
  ||\nablasl\lambda||_{L^4(\mathcal{S}_{u,v})}\leq 2\Delta_{\Gamma_\star}
  +C(I,\Delta_{e_\star},\Delta_{\Gamma_\star},
  \Delta_{\Psi}(\mathcal{S}))\varepsilon
  +C(I,\Delta_{e_\star},\Delta_{\Gamma_\star},\Delta_{\Psi}(\mathcal{S}))
  \Delta_{\Psi}\varepsilon^{7/8},
\end{align*}
where it has been used that
\begin{align*}
  \int_0^u\left(\int_{\mathcal{S}_{u',v}}|\Psi_4|^2\right)^{1/8}\mathrm{d}u'\leq
  \left(\int_0^u\int_{\mathcal{S}_{u',v}}|\Psi_4|^2\mathrm{d}u'\right)^{1/8}
  \left(\int_0^u 1\mathrm{d}u'\right)^{7/8}\leq\varepsilon^{7/8}
  ||\Psi_4||_{L^2(\mathcal{N}'_v(0,u))}^{1/4}.
\end{align*}
Thus, we can choose a suitable~$\varepsilon>0$ such
that~$||\nablasl\lambda||_{L^4(\mathcal{S}_{u,v})}\leq3\Delta_{\Gamma_\star}$. This
improves the starting bootstrap assumption.

\subsubsection*{Estimates on $||\nablasl^2\lambda||_{L^2(\mathcal{S}_{u,v})}$}

In this case we start from
\begin{align*}
  \int_{\mathcal{S}_{u,v}}|\Delta\left\langle\nablasl^2\lambda,\nablasl^2
  \lambda\right\rangle_{\bmsigma}|&=\int_{\mathcal{S}_{u,v}}2|
  \Delta(T_1\bar T_1+T_2\bar T_2+T_3\bar T_3+T_4\bar T_4)| \\
  &\leq2\sqrt{2}\int_{\mathcal{S}_{u,v}}|\nablasl^2\lambda|
  \left(|\Delta T_1|+|\Delta T_2|+|\Delta T_3|+|\Delta T_4| \right).
\end{align*}
we can then further expand to obtain (in schematic notation for
simplicity) that
\begin{align*}
  &\int_{\mathcal{S}_{u',v}}|\Delta\left\langle\nablasl^2\lambda,\nablasl^2
  \lambda\right\rangle_{\bmsigma}|\leq2\sqrt{2}\int_{\mathcal{S}_{u',v}}|
  \nablasl^2\lambda|(|\Gamma'||\nablasl^2\lambda|+|\Gamma'||\nablasl^2
  \Gamma|+|\nablasl^2\Psi_4|+|\nablasl\Gamma||\nablasl\Gamma| \\
  &\qquad\quad+|\Gamma^2||\nablasl\Gamma|+|\Psi_4||\nablasl\Gamma|+|\Psi_3||
  \nablasl\lambda|+|\Gamma||\nablasl\Psi_4|+|\Psi_4||\Gamma^2|+|\Gamma^4|) \\
  &\qquad\leq C(\Delta_{\Gamma_\star})\int_{\mathcal{S}_{u',v}}|\nablasl^2\lambda|^2
  +C(\Delta_{\Gamma_\star})\int_{\mathcal{S}_{u',v}}|\nablasl^2\lambda||
  \nablasl^2\Gamma|+\int_{\mathcal{S}_{u',v}}|\nablasl^2
  \lambda||\nablasl^2\Psi_4| \\
  &\qquad\quad+\int_{\mathcal{S}_{u',v}}|\nablasl^2\lambda||\nablasl\Gamma||\nablasl\Gamma|
  +C(I,\Delta_{e_{\star}},\Delta_{\Gamma_\star},\Delta_{\Psi}(\mathcal{S}))
  \int_{\mathcal{S}_{u',v}}|\nablasl^2
  \lambda||\nablasl\Gamma|+C(\Delta_{\Psi}(\mathcal{S}))\int_{\mathcal{S}_{u',v}}
  |\nablasl^2\lambda||\nablasl\Gamma|  \\
  &\qquad\quad
  +C(I,\Delta_{e_{\star}},\Delta_{\Gamma_\star},
    \Delta_{\Psi}(\mathcal{S}))\int_{\mathcal{S}_{u',v}}|\nablasl^2\lambda
  ||\nablasl\Psi_4|+C(I,\Delta_{e_{\star}},\Delta_{\Gamma_\star},
    \Delta_{\Psi}(\mathcal{S}))\int_{\mathcal{S}_{u',v}}
  |\nablasl^2\lambda|\\
  &\qquad\quad+\int_{\mathcal{S}_{u',v}}|\nablasl^2\lambda||\nablasl\Gamma||\Psi_4|
  +C(I,\Delta_{e_{\star}},\Delta_{\Gamma_\star},\Delta_{\Psi}(\mathcal{S}))
  \int_{\mathcal{S}_{u',v}}|\nablasl^2\lambda||\Psi_4|  \\
  &\qquad\leq C(\Delta_{\Gamma_\star})||\nablasl^2\lambda||_{L^2(\mathcal{S}_{u',v})}^2
  +C(\Delta_{\Gamma_\star})||\nablasl^2\lambda||_{L^2(\mathcal{S}_{u',v})}
  ||\nablasl^2\Gamma||_{L^2(\mathcal{S}_{u',v})}\\
  &\qquad\quad+||\nablasl^2\lambda||_{L^2(\mathcal{S}_{u',v})}||
  \nablasl^2\Psi_4||_{L^2(\mathcal{S}_{u',v})}
  +||\nablasl^2\lambda||_{L^2(\mathcal{S}_{u',v})}||\nablasl\Gamma||^2_{L^4(\mathcal{S}_{u',v})}\\
  &\qquad\quad+C(I,\Delta_{e_{\star}},\Delta_{\Gamma_\star},\Delta_{\Psi}(\mathcal{S}))
  ||\nablasl^2\lambda||_{L^2(\mathcal{S}_{u',v})}
  ||\nablasl\Gamma||_{L^4(\mathcal{S}_{u',v})}+C(\Delta_{e_\star},\Delta_{\Gamma_\star})
  ||\Psi_4||_{L^{\infty}(\mathcal{S}_{u',v})}||\nablasl^2\lambda||_{L^2(\mathcal{S}_{u',v})} \\
  &\qquad\quad+C(I,\Delta_{e_{\star}},
    \Delta_{\Gamma_\star},\Delta_{\Psi}(\mathcal{S}))
  ||\nablasl^2\lambda||_{L^2(\mathcal{S}_{u',v})}
  ||\nablasl\Psi_4||_{L^2(\mathcal{S}_{u',v})}
  +C(I,\Delta_{e_{\star}},\Delta_{\Gamma_\star},\Delta_{\Psi}(\mathcal{S}))
  ||\nablasl^2\lambda||_{L^2(\mathcal{S}_{u',v})}\\
  &\qquad\quad+C(I,\Delta_{e_{\star}},\Delta_{\Gamma_\star},
    \Delta_{\Psi}(\mathcal{S}))||\nablasl^2\lambda||_{L^2(\mathcal{S}_{u',v})}
  ||\nablasl\Gamma||_{L^4(\mathcal{S}_{u',v})}
  ||\Psi_4||_{L^{\infty}(\mathcal{S}_{u',v})} \\
  &\qquad\leq C(I,\Delta_{e_\star},\Delta_{\Gamma_\star},\Delta_{\Psi}(\mathcal{S}))
  (1+||\Psi_4||_{L^2(\mathcal{S}_{u',v})}+||\nablasl\Psi_4||_{L^2(\mathcal{S}_{u',v})}
  +||\nablasl^2\Psi_4||_{L^2(\mathcal{S}_{u',v})}).
\end{align*}
In the previous chain of inequalities we have made repeated use of our
bootstrap assumption, the results in
Proposition~\ref{PropositionL4Estimates} and of H\"older's inequality.
Finally, combining with the short direction estimate in
Proposition~\ref{Proposition:TransportLpEstimates} we conclude that
\begin{align*}
  ||\nablasl^2\lambda||_{L^2(\mathcal{S}_{u,v})}\leq 2\Delta_{\Gamma_\star}
  +C(I,\Delta_{e_\star},
    \Delta_{\Gamma_\star},\Delta_{\Psi}(\mathcal{S}))\varepsilon
  +C(I,\Delta_{e_{\star}},\Delta_{\Gamma_\star},\Delta_{\Psi}(\mathcal{S}))
  \Delta_{\Psi}\varepsilon^{3/4}.
\end{align*}
The factor~$\varepsilon^{3/4}$ results from the transferring of
the~$2$-sphere estimate of~$\Psi_4$ to the light cone.


\begin{thebibliography}{10}

\bibitem{BonBurMet62}
H.~Bondi, M.~G.~J. van~der Burg, and A.~W.~K. Metzner.
\newblock Gravitational waves in general relativity {V}{I}{I}. {Waves} from
  axi-symmetric isolated systems.
\newblock {\em Proc. Roy. Soc. Lond. A}, 269:21, 1962.

\bibitem{Sac62c}
R.~K. Sachs.
\newblock Gravitational waves in general relativity {V}{I}{I}{I}. {Waves} in
  asymptotically flat space-time.
\newblock {\em Proc. Roy. Soc. Lond. A}, 270:103, 1962.

\bibitem{Sac62b}
R.~K. Sachs.
\newblock On the characteristic initial value problem in gravitational theory.
\newblock {\em J. Math. Phys.}, 3:908, 1962.

\bibitem{NewPen62}
E.~T. Newman and R.~Penrose.
\newblock An approach to gravitational radiation by a method of spin
  coefficients.
\newblock {\em J. Math. Phys.}, 3:566, 1962.

\bibitem{Pen65a}
R.~Penrose.
\newblock Zero rest-mass fields including gravitation: asymptotic behaviour.
\newblock {\em Proc. Roy. Soc. Lond. A}, 284:159, 1965.

\bibitem{MulSei77}
H.~{M\"uller zu Hagen} and H.-J. Seifert.
\newblock On characteristic initial-value and mixed problems.
\newblock {\em Gen. Rel. Grav.}, 8:259, 1977.

\bibitem{Fri81a}
H.~Friedrich.
\newblock On the regular and the asymptotic characteristic initial value
  problem for {Einstein}'s vacuum field equations.
\newblock {\em Proc. Roy. Soc. Lond. A}, 375:169, 1981.

\bibitem{Fri81b}
H.~Friedrich.
\newblock The asymptotic characteristic initial value problem for {Einstein}'s
  vacuum field equations as an initial value problem for a first-order
  quasilinear symmetric hyperbolic system.
\newblock {\em Proc. Roy. Soc. Lond. A}, 378:401, 1981.

\bibitem{Fri82}
H.~Friedrich.
\newblock On the existence of analytic null asymptotically flat solutions of
  {Einstein}'s vacuum field equations.
\newblock {\em Proc. Roy. Soc. Lond. A}, 381:361, 1982.

\bibitem{SteFri82}
J.~M. Stewart and H.~Friedrich.
\newblock Numerical relativity. {T}he characteristic initial value problem.
\newblock {\em Proc. Roy. Soc. Lond. A}, 384:427, 1982.

\bibitem{IsaWelWin83}
R.~A. Isaacson, J.~S. Welling, and J.~Winicour.
\newblock Null cone computation for gravitational radiation.
\newblock {\em J. Math. Phys.}, 24:1824, 1983.

\bibitem{FriSte83}
H.~Friedrich and J.~Stewart.
\newblock Characteristic initial data and wavefront singularities in general
  relativity.
\newblock {\em Proc. Roy. Soc. Lond. A}, 385:345, 1983.

\bibitem{Ren90}
A.~D. Rendall.
\newblock Reduction of the characteristic initial value problem to the cauchy
  problem and its application to the einstein equations.
\newblock {\em Proc. Roy. Soc. Lond. A}, 427:221, 1990.

\bibitem{Fou52}
Y.~Four\`{e}s-Bruhat.
\newblock Th\'{e}or\`{e}me d'existence pour certains syst\`emes d'\'{e}quations
  aux deriv\'{e}es partielles non lin\'{e}aires.
\newblock {\em Acta Mathematica}, 88:141, 1952.

\bibitem{ChrKla90}
D.~Christodoulou and S.~Klainerman.
\newblock Asymptotic properties of linear field equations in minkowski
  spacetime.
\newblock {\em Comm. Pure Appl. Math.}, 43:137, 1990.

\bibitem{ChrKla93}
D.~Christodoulou and S.~Klainerman.
\newblock {\em The global nonlinear stability of the {Minkowski} space}.
\newblock Princeton University Press, 1993.

\bibitem{Chr08}
D.~Christodoulou.
\newblock {\em The Formation of Black Holes in General Relativity}.
\newblock Number v. 4 in EMS monographs in mathematics. European Mathematical
  Society Publishing House, 2009.

\bibitem{Luk12}
J.~Luk.
\newblock On the local existence for the characteristic initial value problem
  in general relativity.
\newblock {\em Int. Math. Res. Not.}, 20:4625, 2012.

\bibitem{LiZhu18}
J.~Li and X.-P. Zhu.
\newblock On the local extension of the future null infinity.
\newblock {\em J. Diff. Geom.}, 110:73, 2018.

\bibitem{ChoChrMar11}
Y.~Choquet-Bruhat, P.~T. Chru\'{s}ciel, and J.~M. Mart\'{\i}n-Garc\'{\i}a.
\newblock The cauchy problem on a characteristic cone for the einstein
  equations in arbitrary dimensions.
\newblock {\em Ann. Henri Poincar\'e}, 12:419, 2011.

\bibitem{ChrPae12}
P.~T. Chru\'{s}ciel and T.-T. Paetz.
\newblock The many ways of the characteristic cauchy problem.
\newblock {\em Class. Quantum Grav.}, 29:145006, 2012.

\bibitem{CabChrWaf14}
Aurore Cabet, Piotr~T. Chruściel, and Roger~Tagne Wafo.
\newblock {On the characteristic initial value problem for nonlinear symmetric
  hyperbolic systems, including Einstein equations}.
\newblock 2014.

\bibitem{Cag80}
F.~Cagnac.
\newblock Probl\`eme de cauchy sur un cono\"{\i}de caract\'eristique.
\newblock {\em Ann. Fac. Sci. Toulouse $2^e$ s\'{e}rie}, 2:11, 1980.

\bibitem{Cag81}
F.~Cagnac.
\newblock Probl\`eme de cauchy sur un cono\"{\i}de caract\'eristique pour des
  equations quasi-lineaires.
\newblock {\em Ann. Mat. Pure Appl.}, 129:13, 1981.

\bibitem{LinRod05}
H.~Lindblad and I.~Rodnianski.
\newblock Global existence for the {Einstein} vacuum equations in wave
  coordinates.
\newblock {\em Comm. Math. Phys.}, 256:43, 2005.

\bibitem{Ste91}
J.~Stewart.
\newblock {\em Advanced general relativity}.
\newblock Cambridge University Press, 1991.

\bibitem{PenRin84}
R.~Penrose and W.~Rindler.
\newblock {\em Spinors and space-time. {V}olume 1. {T}wo-spinor calculus and
  relativistic fields}.
\newblock Cambridge University Press, 1984.

\bibitem{CFEBook}
J.~A. {Valiente Kroon}.
\newblock {\em Conformal Methods in General Relativity}.
\newblock Cambridge University Press, 2016.

\bibitem{Eva98}
L.~C. Evans.
\newblock {\em Partial Differential Equations}.
\newblock American Mathematical Society, 1998.

\end{thebibliography}


\end{document}